\documentclass[11pt]{article}
\usepackage[colorlinks=true,allcolors=blue]{hyperref}
\usepackage[symbol]{footmisc}
\usepackage[utf8]{inputenc}
\usepackage{xcolor}
\usepackage{enumitem}
\usepackage{fullpage}
\usepackage{tabularx}
\usepackage{braket}
\usepackage{amsmath}
\usepackage{amsfonts}
\usepackage{amsthm}
\usepackage{tikz-cd}
\usepackage{bbm}
\usepackage{amssymb}
\usepackage{ulem}
\usepackage{authblk}
\usepackage{tcolorbox}
\usepackage{mathabx}
\usepackage{tikz}
\def\basegraph{\tikz[baseline=.1ex]{
\fill (0,2ex) circle (1pt) coordinate (A);
\fill (3ex,4ex) circle (1pt) coordinate (B);
\fill (3ex,0ex) circle (1pt) coordinate (C);
\fill (6ex,2ex) circle (1pt) coordinate (D);
\fill (9ex,0) circle (1pt) coordinate (E);
\fill (9ex,4ex) circle (1pt) coordinate (F);
\fill (12ex,2ex) circle (1pt) coordinate (G);
\draw (A)--(B);
\draw (A)--(C);
\draw (D)--(B);
\draw (D)--(C);
\draw (D)--(E);
\draw (D)--(F);
\draw (G)--(E);
\draw (G)--(F);
}}

\def\basegraphtwo{\tikz[baseline=.1ex]{
\fill (0,2ex) circle (1pt) coordinate (A);
\fill (3ex,4ex) circle (1pt) coordinate (B);
\fill (3ex,0ex) circle (1pt) coordinate (C);
\fill (6ex,2ex) circle (1pt) coordinate (D);
\fill (12ex,0) circle (1pt) coordinate (E);
\fill (12ex,4ex) circle (1pt) coordinate (F);
\fill (15ex,2ex) circle (1pt) coordinate (G);
\fill (9ex,2ex) circle (1pt) coordinate (H);
\draw (A)--(B);
\draw (A)--(C);
\draw (D)--(B);
\draw (D)--(C);
\draw (H)--(E);
\draw (H)--(F);
\draw (G)--(E);
\draw (G)--(F);
\draw (D)--(H);
}}

\def\lsinglegraph{\tikz[baseline=.1ex]{
\fill (0,2ex) circle (1pt) coordinate (A) node [left]{\small{$3^0_i$}};
\fill (3ex,4ex) circle (1pt) coordinate (B) node [above]{\small{$1^0_i$}};
\fill (3ex,0ex) circle (1pt) coordinate (C) node [below]{\small{$2^0_i$}};
\fill (6ex,2ex) circle (1pt) coordinate (D) node [right]{\small{$4^0_i$}};
\draw (A)--(B);
\draw (A)--(C);
\draw (D)--(B);
\draw (D)--(C);
}}

\def\rsinglegraph{\tikz[baseline=.1ex]{
\fill (0,2ex) circle (1pt) coordinate (A) node [left]{\small{$3^1_i$}};
\fill (3ex,4ex) circle (1pt) coordinate (B) node [above]{\small{$1^1_i$}};
\fill (3ex,0ex) circle (1pt) coordinate (C) node [below]{\small{$2^1_i$}};
\fill (6ex,2ex) circle (1pt) coordinate (D) node [right]{\small{$4^1_i$}};
\draw (A)--(B);
\draw (A)--(C);
\draw (D)--(B);
\draw (D)--(C);
}}

\def\xsinglegraph{\tikz[baseline=.1ex]{
\fill (0,2ex) circle (1pt) coordinate (A) node [left]{\small{$3^{x_i}_i$}};
\fill (3ex,4ex) circle (1pt) coordinate (B) node [above]{\small{$1^{x_i}_i$}};
\fill (3ex,0ex) circle (1pt) coordinate (C) node [below]{\small{$2^{x_i}_i$}};
\fill (6ex,2ex) circle (1pt) coordinate (D) node [right]{\small{$4^{x_i}_i$}};
\draw (A)--(B);
\draw (A)--(C);
\draw (D)--(B);
\draw (D)--(C);
}}

\def\xxsinglegraph{\tikz[baseline=.1ex]{
\fill (0,2ex) circle (1pt) coordinate (A) node [left]{\small{$3'^x_i$}};
\fill (3ex,4ex) circle (1pt) coordinate (B) node [above]{\small{$1'^x_i$}};
\fill (3ex,0ex) circle (1pt) coordinate (C) node [below]{\small{$2'^x_i$}};
\fill (6ex,2ex) circle (1pt) coordinate (D) node [right]{\small{$4'^x_i$}};
\draw (A)--(B);
\draw (A)--(C);
\draw (D)--(B);
\draw (D)--(C);
}}

\def\ysinglegraph{\tikz[baseline=.1ex]{
\fill (0,2ex) circle (1pt) coordinate (A) node [left]{\small{$3^{y_i}_i$}};
\fill (3ex,4ex) circle (1pt) coordinate (B) node [above]{\small{$1^{y_i}_i$}};
\fill (3ex,0ex) circle (1pt) coordinate (C) node [below]{\small{$2^{y_i}_i$}};
\fill (6ex,2ex) circle (1pt) coordinate (D) node [right]{\small{$4^{y_i}_i$}};
\draw (A)--(B);
\draw (A)--(C);
\draw (D)--(B);
\draw (D)--(C);
}}

\def\yysinglegraph{\tikz[baseline=.1ex]{
\fill (0,2ex) circle (1pt) coordinate (A) node [left]{\small{$3'^y_i$}};
\fill (3ex,4ex) circle (1pt) coordinate (B) node [above]{\small{$1'^y_i$}};
\fill (3ex,0ex) circle (1pt) coordinate (C) node [below]{\small{$2'^y_i$}};
\fill (6ex,2ex) circle (1pt) coordinate (D) node [right]{\small{$4'^y_i$}};
\draw (A)--(B);
\draw (A)--(C);
\draw (D)--(B);
\draw (D)--(C);
}}

\def\primesinglegraph{\tikz[baseline=.1ex]{
\fill (0,2ex) circle (1pt) coordinate (A) node [left]{\small{$3''_i$}};
\fill (3ex,4ex) circle (1pt) coordinate (B) node [above]{\small{$1''_i$}};
\fill (3ex,0ex) circle (1pt) coordinate (C) node [below]{\small{$2''_i$}};
\fill (6ex,2ex) circle (1pt) coordinate (D) node [right]{\small{$4''_i$}};
\draw (A)--(B);
\draw (A)--(C);
\draw (D)--(B);
\draw (D)--(C);
}}

\theoremstyle{definition}
\newtheorem{definition}{Definition}
\newtheorem*{definition*}{Definition}
\newtheorem{theorem}{Theorem}
\newtheorem{lemma}{Lemma}
\newtheorem{claim}{Claim}
\newtheorem{proposition}{Proposition}

\newtheorem*{lemma*}{Lemma}
\newtheorem*{theorem*}{Theorem}
\newtheorem*{proposition*}{Proposition}
\theoremstyle{remark}
\newtheorem*{remark}{Remark}

\title{
Computational complexity of the homology problem with orientable filtration:
 MA-completeness}
\author[1]{Ryu Hayakawa
\footnote{ryu.hayakawa@yukawa.kyoto-u.ac.jp}
}
\author[2]{Casper Gyurik}
\author[2]{Mahtab Yaghubi Rad}
\author[2]{Vedran Dunjko}
\affil[1]{Yukawa Institute for Theoretical Physics \& The Hakubi Center, Kyoto University, Japan}
\affil[2]{applied Quantum algorithms (aQa), Leiden University, 2300 RA Leiden, The Netherlands}

\date{}

\begin{document}

\noindent
\hspace{\fill} YITP-25-139
\begingroup
\let\newpage\relax
\maketitle
\endgroup


\begin{abstract}
We show the existence of an MA-complete homology problem for a certain subclass of simplicial complexes. The problem is defined through a new concept of orientability of simplicial complexes that we call a ``uniform orientable filtration'', which is related to sign-problem freeness in homology. The containment in MA is achieved through the design of new, higher-order random walks on simplicial complexes associated with the filtration. For the MA-hardness, we design a new gadget with which we can reduce from an MA-hard stoquastic satisfiability problem. Therefore, our result provides the first natural MA-complete problem for higher-order random walks on simplicial complexes, combining the concepts of topology, persistent homology, and quantum computing.
\end{abstract}

\vspace{0.5cm}
\section{Introduction}

Higher-order random walks are higher-order generalizations of random walks on graphs to higher-order objects such as simplicial complexes. 
Recently, higher-order random walks have been actively studied in the field of {\it high-dimensional expansion}, which is related to the mixing time of the higher-order random walks~\cite{kaufman2017high,alev2020improved,kaufman2020high,anari2021spectral,chen2021optimal,chen2022localization}. 
The basic philosophy of high-dimensional expansion is the so-called {\it local-to-global} expansion, in which certain expansion properties of the local structure of a higher-order object will be utilized to derive a global and high-dimensional expansion property. 
Higher-order random walks have also attracted attention from the perspective of the topology hidden in simplicial complexes~\cite{mukherjee2016random,parzanchevski2017simplicial,schaub2020random}. Motivation behind these lines of study is the application of higher-order random walks to problems in {\it topological data analysis} (TDA).

However, the usage of higher-order random walks for the analysis of topology may be limited due to the NP-hardness of the homology problem~\cite{adamaszek2016complexity}. Moreover, it has been recently discovered that the homology problem is indeed QMA$_1$-hard~\cite{crichigno2024clique}.
Here, the homology problem is a problem to decide if a given simplicial complex (more specifically, a clique complex described by a graph) has a $d$-dimensional hole or not.
This provided solid indications of the quantum-mechanical nature of the underlying problem. 
Subsequent work has shown that a promise variant of the homology problem is contained in QMA while still being QMA$_1$ hard under the same promise~\cite{king:qma}. 
In quantum topological data analysis (QTDA), both algorithmic \cite{lloyd2016quantum,hayakawa2022quantum,mcardle2022streamlined,black2023incremental,leditto2025topological,akhalwaya2024comparing,hayakawa2024quantum} and complexity theoretic aspects \cite{gyurik2022towards,crichigno2024clique,berry2024analyzing,king:qma,rayudu2024fermionic,gyurik2024quantum,rudolph2024towards}, as well as extension to thermal states~\cite{scali2024topology}, extension to Khovanov homology \cite{schmidhuber2025quantum}, and dequantization algorithms~\cite{apers2023simple} have been actively studied. 

In this work, we identify a natural subclass of the homology problem that falls in the complexity class MA. Here, MA is a complexity class where the prover sends a classical witness to the verifier, and the verifier performs an efficient randomized verification procedure. 
Therefore, it naturally fits the framework of higher-order random walks. 

A known canonical MA-complete problem is the {stoquastic satisfiability problem}~\cite{bravyi2010complexity}. It is also extended to several other problems or other settings~\cite{aharonov2020two, raza2024complexity}. 
The stoquastic property characterizes a particular class of Hamiltonians with which the {\it sign-problem} in quantum Monte-Carlo does not occur. 
Inspired by the stoquastic satisfiability problem, we investigate a class of simplicial complexes within which we can escape {\it a sign-problem in the homology problem}. 
Indeed, it can be said that previous works on higher-random walks suffer from a certain type of sign-problem that are deeply related to the orientation of simplices: previous works~\cite{mukherjee2016random,parzanchevski2017simplicial,schaub2020random} perform random walks on simplices with both positive and negative orientations by treating them as distinct states\footnote{We say two oriented simplices $\sigma=v_0v_1...v_d$ and $\sigma'=v_0'v_1'...v_d'$ have opposite (positive or negative) orientations if they are composed of the same vertices and match with an odd permutation.}. 
Roughly, the doubling of state space and the necessary post-processing cause inefficiency. 
We remark that, as there are no natural orientations for the vertices in graphs, the sign-problem in the homology problem appeared {\it as a consequence of a higher-order generalization} of random walks. 

{\color{blue}

}

\subsection{Our contributions}

The main contribution of this paper is the establishment of an MA-completeness for a class of the homology problem (Theorem~\ref{thm:main}).
Conceptually, the problem is formulated through a new notion of orientability of simplices that we call a \textbf{uniform orientable filtration}. 
With a uniform orientable filtration, we establish a particular way of escaping the sign-problem, which leads to the existence of an efficient MA-verification protocol based on higher-order random walks while keeping it as hard as possible. 
The MA-hardness is established through a new gadget construction that allows a reduction from a stoquastic satisfiability problem to a homology problem on simplicial complexes that allows a uniform orientable filtration, as well as other conditions required in the problem. 
Therefore, our result extends the class of simplicial complexes for which classical random walks efficiently reveal the topology of simplicial complexes without suffering from the sign-problem. 
{Although the MA-completeness result does not immediately provide a practical application, we believe that our result opens a new possibility of performing high-dimensional TDA with random walks by finding a suitable filtration of the data in terms of orientability.}
Moreover, our result provides the first MA-complete problem for a natural problem in simplicial complexes relevant for higher-order random walks and high-dimensional expansion.



\subsection{Overview of problem definition and proofs}

\paragraph{Uniform orientable filtration and problem definition}

Before introducing a uniform orientable filtration, we recall the usual notion of orientability. 
An (abstract) simplicial complex is a family of sets of vertices that is closed under taking subsets. 
{Any $d$-dimensional simplex (that is composed of $d+1$-number of vertices) with $d\geq 1$ can have two orientations. We regard two oriented simplices (ordered vertices) with the same vertices as having the same/opposite orientations if they match under even/odd permutations.}   
A simplicial complex is said to be $d$-orientable if there is a choice of orientations s.t. whenever two $d$-simplices intersect (i.e., share a common $d-1$-dimensional face), they induce opposite orientations on it. 
For a simplicial complex with such orientations, the combinatorial Laplacian indeed becomes stoquastic, which is a starting point for a complexity theoretic analysis.
This is desirable for the containment in MA because a stoquastic Hamiltonian leads to a design of a Markov transition matrix. However, it is limited in terms of hardness: a simplicial complex can only be fully orientable if $d-1$-simplices are degree-2, which seems to be too simple to be MA-hard\footnote{In fact, the currently known MA-hard stoquastic Hamiltonian itself is degree-4 in the lowest case~\cite{raza2024complexity}.}.
This motivates us to look for a broader class of simplicial complexes while maintaining containment in MA. 
While achieving stoquasticity would be a clear way of avoiding the sign problem and achieving containment in MA, we solve the problem by generalizing the class in a way that is natural from the perspective of homology. 
However, as a consequence of such generalization, the combinatorial Laplacian will no longer be stoquastic. 
Nonetheless, we show that it suffices by relying on the fixed-node Hamiltonian construction to effectively obtain a stoquastic operator and by constructing a Markov transition matrix from the fixed-node Hamiltonian.


The new notion of orientability that we call a {\it orientable filtration}, concerns both {\it upward-orientability} and {\it downward-orientability}. 
The usual notion of $d$-orientability that we discussed above concerns the adjacency through $d-1$-dimensional simplices. Therefore, from our point of view, the usual orientability concerns downward-orientability. 
Unlike the vertices in graphs, which are connected only by edges (cofaces for vertices), simplices can be connected through both faces and cofaces. 
Here, for a $d$-simplex $\sigma$, $d+1$-simplex that contains $\sigma$ as a face is called coface of $\sigma$. 
Therefore, we would like to also consider orientability through upward-adjacency (upward-orientability). 
Consider a sequence of simplicial complexes
$$
X^0 \subseteq X^1 \subseteq \cdots \subseteq X^N. 
$$
Such a sequence of sub-complexes is called a {\it filtration} of $X^N$. 
Note that $X^N\subseteq X$ need not be equal to the full simplicial complex. 
Then, the orientable filtration requires $d$-simplices that are added at the $i$-th step of the filtration 
(we denote these added simplices by $\tilde X^i_d$), to be orientable, i.e., whenever two simplices in $\tilde X^i_d$ share a face, they induce opposite orientations on it. This is an orientability that cares about downward relationships inside $\tilde X^i_d$. 
This is analogous to the usual notion of orientability for $\tilde X^i_d$. 
Another orientability that we impose is an upward orientability among the simplices in different stages of the filtration: we require that whenever simplices in $\tilde X^i_d$ and $\tilde X^j_d$ share a coface, they induce opposite orientations on it. An example can be seen in Figure~\ref{fig:orientable_filtration}. 
Note that the above explanation is rather informal and does not completely capture all the situations that we need to care about. 
We also impose some uniformity on the orientable filtration in the formal definition, and the orientability will be called a uniform orientable filtration. 
\\

Then, the problem is to decide whether $X$ has a $d$-dimensional hole or not under the following conditions and promises: \vspace{-0.2cm}
\begin{itemize}
    \item Input: An oriented clique complex $X$ on $n$ vertices, description of a uniform orientable filtration\\ $X^0 \subseteq X^1 \subseteq \cdots \subseteq X^N$, target dimension $d$, and $\epsilon >1/poly(n)$. \vspace{-0.2cm}
    \item YES instances: $H_d(X)$ is non-trivial.  It is promised that there is a homologous cycle $\ket{c}$ in $X_d^0$ and the corresponding harmonic state $\ket{\phi}=\ket{c}+\ket{b}$ where $\ket{b}\in \mathrm{Im}\ \partial_{d+1}$ is supported on $X_d^N$. \vspace{-0.2cm}
    \item NO instances:  $H_d(X)$ is trivial. It is promised that $\lambda(\Delta_d)\geq \epsilon$, where $\lambda(\Delta_d)$ is the minimal eigenvalue of the combinatorial Laplacian $\Delta_d$. 
\end{itemize}

The formal definition of the problem can be seen in Definition~\ref{def:problem}. 
In the ``description of a uniform orientable filtration'' we assume that we are given a description of a classical circuit that returns the index in the filtration and the orientation by inputting a $d$-simplex that satisfies the conditions in the orientable filtration. 
We also assume that the orientation of simplices is given as an input of the problem. 
It is essentially a {\it persistence} problem because in YES instances, a hole in $X_d^0$ persists as a hole across the uniform orientable filtration. 

Our result can be informally stated as follows.
\begin{theorem}[Main (Informal statement of Theorem~\ref{thm:main})]
    The homology problem for simplicial complexes with orientable filtration is MA-complete. 
\end{theorem}

\paragraph{Containment in MA}
The containment in MA is achieved through the design of a stoquastic Hermitian operator on the chain space spanned by simplices of the target dimension. 
Although the combinatorial Laplacian for the given simplicial complexes may not be fully stoquastic, we can turn it into a stoquastic Hamiltonian with some modifications. 
We modify the Laplacian with a construction known as the fixed-node Hamiltonian construction~\cite{ten1995proof,bravyi2023rapidly}. 
Nice properties of the fixed-node construction are that they maintain the ground state and spectral gap of the original Hamiltonian (See Lemma~\ref{lemma:fixed_node_Hamiltonian}). 
In the modification of the Laplacian, a non-negative state is utilized as a ``fixed-node''. The state is indeed a (non-negative) harmonic state in the kernel of the combinatorial Laplacian in YES instances. Here, non-negative states refer to the states with non-negative amplitudes in the computational basis. 
In NO instances, even though there are no harmonics, we utilize a state which is not energetically penalized in some local sense, i.e., a state that is composed of connected components of ``locally-good'' simplices. There are two types of ``badness'' associated with the orientable filtration. One is ``non-cocycleness'' as described in Figure~\ref{fig:eg_good_bad} (b). The other type of badness is ``non-cycleness'' as described in Figure~\ref{fig:eg_good_bad} (c). 
\begin{figure}
    \centering
    \includegraphics[width=0.9\linewidth]{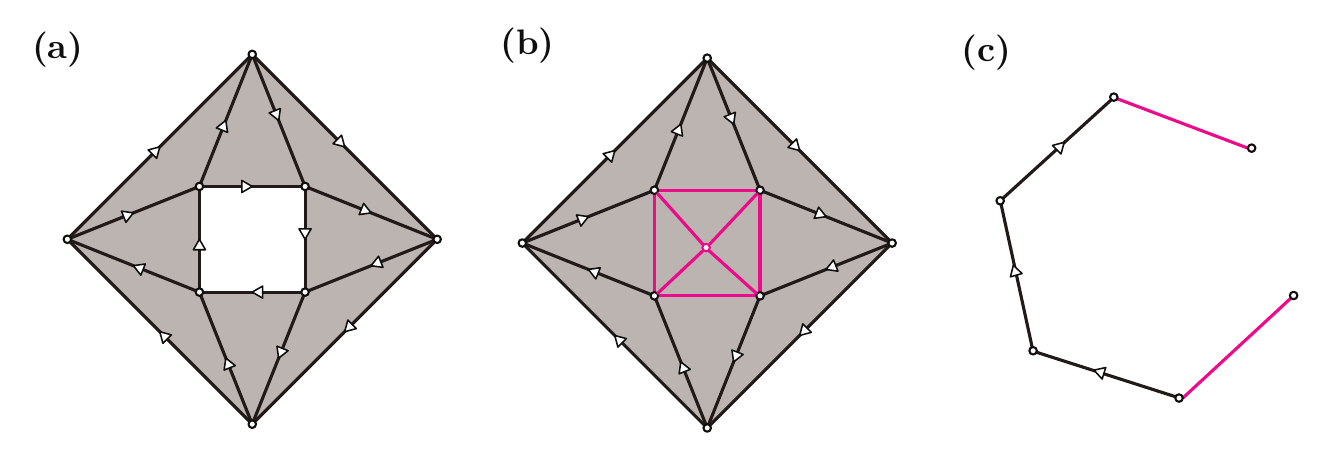}
    \caption{Example of states that are utilized in the fixed-node construction. In this figure, we are assuming that the simplices are filtrated in a cyclic way as in Figure~\ref{fig:orientable_filtration}. (a) A simplicial complex with a 1-dimensional hole. (b)  Center triangles fill the cycle. Therefore, the red edges possess a badness of ``non-cocycleness'' (c) There is no cycle. The red edges possess a badness of ``non-cycleness''. }
    \label{fig:eg_good_bad}
\end{figure}
The fixed-node state consists of the connected components of good simplices (edges colored in black in Figure~\ref{fig:eg_good_bad}). 
Note that the concept of goodness/badness is valid in the existence of an orientable filtration. 
Then, the verifier's strategy is to find bad simplices with random walks on simplicial complexes associated with the modified Laplacian with a fixed-node, starting from an initial simplex provided by the prover. 
We can show that the transition probabilities for adjacent simplices can be determined efficiently only by looking at local neighbors of the simplex of the current position. (Two adjacent simplices only differ by one vertex.) 
In YES instances, the random walk is performed in the support of a harmonic, and bad simplices never appear. 
In NO instances, we show that the verifier can reject with high probability.

\paragraph{MA-hardness}
The MA-hardness is shown through a reduction from a certain MA-hard stoquastic SAT problem. 
This stoquastic SAT problem that we use consists of projectors as $H=\sum_i h_i$ onto either computational basis states or the difference of two computational basis states. 
Therefore, a $m$-qubit projector $h_i$ can be written as either $h_i=\ket{x}\bra{x}$ or  $h_i=\frac{1}{2}(\ket{x}-\ket{y})(\bra{x}-\bra{y})$ for some $x,y \in \{0,1\}^m$, see  Lemma~\ref{lemma:setofprojectors}. 

In order to map an instance of stoquastic satisfiability problems, we identify the correct Hilbert space in the simplices. 
The $n$-qubit Hilbert space will be encoded into the harmonic space of $X_{2n-1}^0$ whose dimension is $2^n$. 
In order to construct such a clique complex, we utilize an $n$-fold join product of graphs with two 1-dimensional holes:$\basegraphtwo$. 
Then, $X^0$ is the clique complex of the graph. 
All the vertices in $X^0$ are weighted by $1$. 
The base graph that we use is different from that used in the previous work of \cite{king:qma}, which was  $\basegraph$. The reason that we use a different graph is that the graph used in~\cite{king:qma} is not degree-2 and thus not suitable in terms of the orientability of $X_{2n-1}^0$\footnote{Our graph is not degree-2 as well. However, there is an appropriate way to make it degree-2 within the filtration.}. The target dimension is $2n-1$ because we are considering holes that are composed of $n$-fold products of edges. 
Then the strategy to reduce from stoquastic SAT is to fill the holes that correspond to states penalized by local terms of the Hamiltonian. 
This allows us to ensure that only the encoding of a state $\ket{\phi}$ s.t. $H\ket{\phi}=0$ into chain space remains a valid ``hole'' after the gadget construction corresponds to local projectors. 

For projectors $h_i=\ket{x}\bra{x}$, we use the same gadget construction with~\cite{king:qma}: copy the target hole and fill it by $2n$-dimensional simplices with an axial vertex so that the copied hole appears as a boundary of the added $2n$-dimensional simplices. 
For projectors $h_i=\frac{1}{2}(\ket{x}-\ket{y})(\bra{x}-\bra{y})$, we take a strategy to ``glue the holes'' corresponding to $\ket{x}$ and $\ket{y}$. 
The ``wormhole'' that connects the holes of $\ket{x}$ and $\ket{y}$, makes $\ket{x}-\ket{y}$ as boundary while $\ket{x}+\ket{y}$ remains to be a hole. Intuitively, it can be described as Figure~\ref{fig:glue_xy}. 
Of course, we need to be careful about the orientations of the cycles when we glue holes because $\ket{y}$ and $-\ket{y}$ are represented by cycles with opposite orientations. 
\begin{figure}[t]

    \centering
    \includegraphics[width=0.3\linewidth]{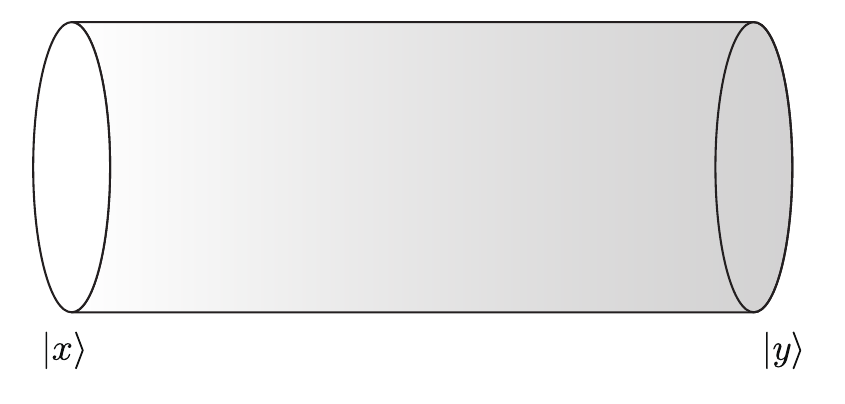}
    \caption{Intuitive gadget construction for $h_i=\frac{1}{2}(\ket{x}-\ket{y})(\bra{x}-\bra{y})$.}
    \label{fig:glue_xy}
\end{figure}
This construction is suitable in terms of orientable filtration because we can filtrate the wormhole by slicing it step by step, and we can systematically maintain orientations. 
Although the vertices in $X^0$ are weighted by $1$, we impose a weight $\lambda \ll 1$ for all the vertices that are added for the gadget construction. 

A technical concern arises here: a new homology class can be born in the procedure of gluing two holes. 
This occurs under two circumstances. The intuition can be obtained by looking at the ``block'' corresponds to the $i$-th qubit.
\begin{itemize}
    \item For the block where $x_i=y_i$, we will end up in creating a {\it torus} for that part.
    \item For the block s.t. $x_i\neq y_i$, the edge that connects two rectangles $\basegraphtwo$ leads to the generation of a new 1-dimensional cycle. 
\end{itemize}
For a more detailed illustration, see Figure~\ref{fig:newhole}. 
The holes in such a local block will lead to the generation of new $2m-1$-dimensional holes. 
Therefore, our attempts to make the simplicial complex orientable with filtration cause to make new holes in the gadget construction.  

Because the newly generated holes are not the ones that persist from $X_{2n-1}^0$, there might be no problem in terms of the persistence of the initial holes. However, this situation does not match the conditions in NO instances where there are no holes at all. 
Therefore, in order to prevent new holes to born, we add axially simplices to fill such undesirable holes. 
Importantly, the gadgets are designed so that the harmonics in the gadget complex in YES instances are not supported on such axial simplices. That is, holes in $X_{2n-1}^{0}$ do not spread out to axial simplices in the persistence, which is a required condition in the problem and important for containment in MA.

In order to conclude the proof, we analyze the lower bound of the minimal eigenvalue of the combinatorial Laplacian in NO instances and the structure of homologous holes in $X_{2n-1}^0$ and harmonics in YES instances. 
The analysis of the minimal eigenvalue is performed through the analysis of spectral sequences. 
The analysis of the spectral sequences is done with the filtration associated with the weight of simplices, where we add weight $\lambda \ll 1$ to vertices that are not present in $X^0$ and we consider the weighting of simplices with the product of weights of vertices as is done in~\cite{king:qma}. 
Finally, we show (1) the existence of filtration $X_{2n-1}^0\subseteq \cdots \subseteq X_{2n-1}^N$ where $X_{2n-1}^N\subseteq X_{2n-1}$ that allows a uniform orientable filtration, (2) the existence of a homologous cycle in $X_{2n-1}^0$ in YES instances, and (3) that the harmonics corresponding to that homologous cycle in $X_{2n-1}^0$ is only supported on $X_{2n-1}^N$.

\subsection{Discussion}

In this paper, we have established the MA-completeness of a homology problem with uniform orientable filtration on simplicial complexes. 
This provides a natural complete problem for higher-order random walks. We give several discussions and open problems.

\paragraph{Related works}

In~\cite{jiang2023local}, it is shown that the local Hamiltonian problem with ``succinct ground state'' is MA-complete. 
Here, the succinct ground state refers to an access to a classical circuit that outputs the amplitudes in the ground state by inputting an orthonormal basis state up to some global constant. 
The fixed-node Hamiltonian construction and continuous time version of the Markov chain Monte-Carlo method were utilized. 
It can be said that our analysis shows that the simplicial complexes with an orientable filtration are the class of simplicial complexes that allow a{\it succinct ground state} for combinatorial Laplacians. 
Combined with the result of~\cite{jiang2023local}, the existence of a succinct ground state would directly lead to the containment of the problem in MA.
However, in our case, we do not need to introduce the demanding continuous-time Markov chain Monte-Carlo because we can avoid an exponential blowup of matrix elements in the fixed-node construction. 
Therefore, we have shown the containment in MA through a direct proof with a protocol that is similar to the original verification protocol of \cite{bravyi2010complexity} combined with our fixed-node Laplacian construction.

In~\cite{eidi2023irreducibility}, random walks on simplicial complexes that can be applied to the topology of simplicial complexes are considered. However, their algorithms are limited to down-walks on simplicial complexes with at most degree-2. Therefore, our algorithm, based on orientable filtration, applies to a broader class of simplicial complexes. 

In~\cite{bonis2024random}, random walks on simplicial complexes whose state space is given by {\it cycles}, i.e., a cycle is updated to another cycle in the random walks. 
However, it is not easy to describe general cycles succinctly and keep updating them. In~\cite{bonis2024random}, the random walks of cycles are performed related to the up Laplacian. 
In contrast, our state space is just simplices; those are the natural basis for the chain space, and they are updated according to a Markov transition matrix that is constructed with a fixed-node modification of the full combinatorial Laplacian. 


Next, we highlight some open questions. 

\paragraph{Non-uniform stoquastic projector}
Although we have shown MA-hardness for projectors onto states $\ket{x}$ and $\ket{x}-\ket{y}$, we are not able to construct gadgets that allows uniform orientable filtration for more general {\it non-uniform} stoquastic projectors such as projectors onto $\ket{x}-2\ket{y}$.  
Is it possible to identify a problem that allows reduction from such non-uniform stoquastic projectors and remain in MA? 
If that is the case, we should be able to identify a broader class of homology problems that is MA-complete. 

\paragraph{MA-hard gadget construction with uniform weights}
Our gadget construction for the MA-hardness relies on artificially imposed inverse-polynomially small weight on simplices, which is similar to the case in~\cite{king:qma}. Can one make the weighting to be uniform while keeping the MA-hardness? 

 \paragraph{Derandomization, high-dimensional expansion of higher-order random walks}
 In~\cite{aharonov2019stoquastic}, it has been shown that a constant gap version of the uniform stoquastic SAT problem is contained in NP. 
 Therefore, a hypothetical PCP theorem for stoquastic SAT with gap amplification, i.e., the MA-hardness of the uniform stoquastic SAT problem, would lead to the derandomization of MA.  
 Derandomization of higher-order random walk is also an important open problem~\cite{feng2023towards}. 
This work provided an MA-complete problem for simplicial complexes. 
 Can we utilize advanced techniques in high-dimensional expansion, such as the spectral independence~\cite{anari2021spectral}, to design an improved or derandomized algorithm for simplicial complexes with orientable filtration?

\subsection{Organization}
The remainder of the paper is organized as follows. 
In Section~\ref{sec:preliminaries}, we introduce preliminaries. 
In Section~\ref{sec:problem_main}, we formally introduce our main problem with the concept of the uniform orientable filtration. 
In Section~\ref{sec:containment_MA}, we introduce a classical verification protocol and show a containment in MA. 
In Sections~\ref{sec:hardness_1},~\ref{sec:hardness_2},~\ref{sec:hardness_3} we provide a proof of MA-hardness. 

\section{Preliminaries}
\label{sec:preliminaries}

\subsection{Stoquastic Hamiltonians and non-negative states }
\label{sec:preliminaries:stoq}

We first introduce non-negative states, with respect to the computational basis, as 
\begin{definition}[non-negative state]
    We say that $\ket{\psi}$ is a non-negative state if $\ket{\psi}$ can be written as 
    $$
    \ket{\psi}=\sum_{x\in\{0,1\}^n} a_x \ket{x}
    $$
    and any coefficients are non-negative real numbers $a_x\in \mathbb{R}_+$ up to a global phase. 
\end{definition}

We denote the support of a non-negative state $\ket{\psi}$ by \begin{equation}\label{eq:def:support}\mathrm{Supp}(\ket{\psi}):= \{x\in\{0,1\}^*: \braket{x|\psi}\neq 0\}\end{equation}. 
An $n$-qubit Hamiltonian $H$ is said to be stoquastic in the computational basis if 
$$
\bra{x}H\ket{y} \leq 0
$$
for all $x\neq y \in \{0,1\}^n$. 
It is known that all stoquastic Hamiltonian has non-negative ground states \cite{bravyi2010complexity}. 

Next, we introduce the construction of the ``fixed-node'' Hamiltonian as follows. 

\begin{definition}[Fixed-node Hamiltonian~\cite{ten1995proof,bravyi2023rapidly}]
\label{def:fixed_node_Hamiltonian}
Let $H\in \mathrm{\mathbb{R}^{2^n\times 2^n}}$ be a real Hamiltonian, and $\ket{\psi}$ be a state with real amplitudes with $\ket{\psi}=\sum_{x\in\{0,1\}^n}a_x \ket{x}$ where $a_x \in \mathrm{\mathbb{R}}$ for all $x\in \{0,1\}^n$ and  $\|\ket{\psi}\|\neq 0$. Let also
$$
S^+=\{
(x,y)\ :\ x\neq y \text{ and } \braket{\psi|x}\bra{x}H\ket{y}\braket{y|\psi}>0
\},
$$
$$
S^-=\{
(x,y)\ :\ x\neq y \text{ and } \braket{\psi|x}\bra{x}H\ket{y}\braket{y|\psi}\leq0
\},
$$
for $x,y \in \{0,1\}^n$. Then, the fixed-node Hamiltonian $F^{H,\psi}$ with respect to a real-valued state $\ket{\psi}$ is defined by 
$$
\bra{x}F^{H,\psi}\ket{y}= 
\begin{cases}
    0  &\text{ if } (x,y)\in S^+,\\ 
    \bra{x}H\ket{y}&\text{ if } (x,y)\in S^-,\\ 
    \bra{x}H\ket{x}+ \sum_{z:(x,z)\in S^+} \bra{x}H\ket{z} \frac{\braket{z|\psi}}{\braket{x|\psi}} &\text{ if } x=y.
\end{cases}
$$
    
\end{definition}

\begin{remark}
    In the above definition, $\ket{\psi}$ need not be the ground state of $H$. 
    Also, $\ket{\psi}$ need not be a non-negative state.
\end{remark}

The following lemma is known about the fixed-node construction, see~\cite{ten1995proof,bravyi2023rapidly,jiang2023local}. 

\begin{lemma}
\label{lemma:fixed_node_Hamiltonian}
Let $F^{H,\psi}$ be a fixed-node Hamiltonian with respect to a real-valued state $\ket{\psi}$. Then, the following hold:
\begin{itemize}
    \item $F^{H,\psi} \ket{\psi}= H\ket{\psi}$. 
    \item If $\ket{\psi}$ is a ground state of $H$, then, $\lambda(H)=\lambda(F^{H,\psi})$ and $\ket{\psi}$ is a ground state of $F^{H,\psi}$. 
    \item $\bra{\phi}F^{H,\psi}\ket{\phi}\geq \bra{\phi}H\ket{\phi}$ for any $\phi$. 
\end{itemize}
\end{lemma}

\subsection{Simplicial complexes and combinatorial Laplacian}
\label{sec:preliminaries:sc}

In this section, we introduce preliminaries on simplicial complexes, homology, and  combinatorial Laplacians. 
An abstract simplicial complex is a family of sets of vertices (simplices) that is closed under taking subsets.  
A particularly relevant class of simplicial complexes in this work is the clique complexes, which are a collection of cliques of a (simple) graph. 

\begin{definition}[Clique complex]
    Given an undirected graph $G=(V,E)$, the clique complex of $G$ is the set of cliques in $G$. We denote the clique complex of $G$ by $\mathrm{Cl}(G)$ and the set of $d$-cliques as $\mathrm{Cl}_{d-1}(G)$.
\end{definition}
Note that $d$-cliques, which are composed of $d$-number of vertices, are $d-1$-simplices. 
Clearly, the subsets of cliques are also cliques of the graph. 
A convenient property of clique complexes that is suitable for the study of computational complexity is that it can be succinctly described by the graph (1-skeleton of the complex), even if there are superpolynomially many simplices in the number of vertices.

In this paper, the orientations of simplices play a crucial role. 
We consider two orientations for any $d$-simplices with $d\geq 1$. 
Two oriented simplices with the same vertices are regarded as equivalent if they are equivalent up to even permutation. 
If they are equivalent up to odd permutation, they are said to have opposite orientations. 
In oriented simplicial complexes, the vertices of simplices are ordered in a fixed way. 
For $\sigma \in X$ where $X$ is an oriented simplicial complex over $n$-vertices, we denote the counterpart with the opposite orientation by $\bar\sigma$. 

We introduce the faces and cofaces of simplices. 
\begin{definition}[Faces and cofaces]
For $\sigma \in X_d$, the set of faces and cofaces of $\sigma$ is defined as follows:
\begin{itemize}
    \item $\mathrm{face}(\sigma)$: set of $d-1$ simplices contained in $\sigma$.
    \item $\mathrm{coface}(\sigma)$: set of $d+1$ simplices those contain $\sigma$.  
\end{itemize}
\end{definition}
Next, we introduce several notions of adjacency of $d$-simplices:
\begin{itemize}
    \item {Up-adjacent:} we say two oriented simplices $\sigma,\sigma'$ are up-adjacent (denote by $\sigma \sim_\uparrow\sigma'$) if $\sigma$ and $\sigma'$ share a coface and induce different orientations on it. 
    \item Down-adjacent: we say two oriented simplices are down-adjacent (denote by $\sigma \sim_{\downarrow} \sigma'$) if $\sigma$ and $\sigma'$ share a common face and induce opposite orientations on it. 
\end{itemize}
Note that $\sigma \sim_\uparrow\sigma'$ implies $\sigma \sim_\downarrow\overline{\sigma'}.$
We use $\sigma \sim\sigma'$ to denote that $\sigma,\sigma'$ share a face (i.e., either $\sigma \sim_\downarrow\sigma'$ or $\bar\sigma \sim_\downarrow\sigma'$ holds).  
The degree of a $d$-simplex $\sigma$, denoted by $\mathrm{deg}(\sigma)$, is the number of ($d+1$ dimensional) cofaces of $\sigma$ in $X$.


Next, we introduce chain space and operators on it. 
For a simplicial complex $X$, let 
$$
C_d = \mathrm{Span}(\ket{\sigma}: \sigma \in X_d)
$$
be the space spanned by the $d$-simplices in $X$ with complex coefficients.  
The boundary operator $\partial_d:C_d \rightarrow C_{d-1}$ is defined by 
$$
\partial_d \ket{\sigma} = \sum_{\tau \in \mathrm{face}(\sigma)} (-1)^{[\sigma:\tau]} \ket{\tau}.
$$
Here, $[\sigma:\tau]=0$ if the induced orientation for $\tau$ by $\sigma$ is the same with the orientation of $\tau$ and $[\sigma:\tau]=1$ otherwise. 
It is equivalent to define for $\sigma=v_0v_1...v_d$ that 
$
\partial_d \ket{\sigma} = \sum_{i} (-1)^i \ket{\sigma\backslash v_i}
$
under the convention $\ket{\bar\sigma}:=-\ket{\sigma}$. 

For unweighted simplicial complex $X$, we define inner products by $\braket{\sigma|\sigma'}=\delta_{\sigma,\sigma'}$ for $\sigma,\sigma' \in X_d$. 
Then, we define coboundary operator $\delta_d:C_d \rightarrow C_{d+1}$ by $\delta_d:= \partial_{d+1}^\dagger$. 
It holds that 
$$
\delta_d \ket{\sigma}= \sum_{\tau \in \mathrm{coface}(\sigma)} (-1)^{[\tau:\sigma]} \ket{\tau}.
$$
The homology $H_d$ of $X$ is defined by $H_d:= \ker \partial_d / \mathrm{Im}\partial_{d-1}$. 

Next, we define useful positive semi-definite Hermitian operators on $C_d$: 
\begin{itemize}
    \item (Up-Laplacian) $\Delta_d^{\mathrm{up}}:=\delta_d^\dagger\delta_d$.
    \item (Down-Laplacian) $\Delta_d^{\mathrm{down}}:= \partial_d^\dagger\partial_d.$
    \item (Full/combinatorial/Hodge Laplacian) $\Delta_d:=\Delta_d^{\mathrm{up}}+\Delta_d^{\mathrm{down}}$.
\end{itemize}

Then, the following representation of matrix elements of Laplacians is known~\cite{goldberg2002combinatorial}
\begin{align*}
\bra{\sigma'}\Delta_d^{\mathrm{up}}\ket{\sigma}=
    \left\{
    \begin{array}{ll}
    \mathrm{deg}(\sigma) & \text{if\ } \sigma=\sigma', \\
    -1 & \text{if } {\sigma} \sim_\uparrow\sigma',\\
    0 & \text{otherwise}.
    \end{array}
    \right.
    \ \ \ \ \ \ 
\bra{\sigma'}\Delta_d^{\mathrm{down}}\ket{\sigma}=
    \left\{
    \begin{array}{ll}
    d+1 & \text{if\ } \sigma=\sigma', \\
    1 & \text{if } \bar{\sigma} \sim_{\downarrow} \sigma', \\
    -1 & \text{if } \sigma\sim_{\downarrow} \sigma',\\
    0 & \text{otherwise}.
    \end{array}
    \right.
\end{align*}

\begin{align*}
\bra{\sigma'}\Delta_{d}\ket{\sigma}=
    \left\{
    \begin{array}{ll}
    \mathrm{deg}(\sigma)+d+1 & \text{\ if\ } \sigma=\sigma', \\
    1 & 
    \begin{array}{l}
    \text{if } \bar{\sigma}\sim_{\downarrow}\sigma' \text{ and } \sigma, \sigma' \text{ do not share a coface}, \\
    \end{array}
     \\
    -1 &     
    \begin{array}{l}
    \text{if } {\sigma}\sim_{\downarrow}\sigma' \text{ and } \sigma, \sigma' \text{ do not share a coface} ,
    \end{array}\\
    0 & \text{\ otherwise}.
    \end{array}
    \right.
\end{align*}
It is known that 
$H_d\cong \ker\Delta_d$ because $\ker\Delta_d=\ker\Delta^{\text{up}}_d \cap \ker\Delta^{\text{down}}_d= \ker\partial_d \cap \mathrm{Im}(\delta_d)^\perp$. 
Therefore, $\ker\Delta_d$ is called the harmonic homology space. A vector in $\ker\Delta_d$ that represents a hole is called a harmonic representative for it\footnote{A hole may be represented by a cycle $\ket{c}$ which cannot be written as a boundary. A harmonic representative for $\ket{c}$ is $\ket{\phi}=\ket{c}+\ket{b}\in \ker\Delta_d$ with $\ket{b}\in \mathrm{Im}(\partial_{d+1})$. See also \cite{basu2024harmonic}}. 

Now we move on to the weighted case. Suppose that we are given the weighting of the vertices.
Then, we consider the vertex-weighting of simplices in a way that is used in \cite{king:qma}. 
The weight of simplices is given as the product of the weight for vertices as 
simplices $\sigma=v_0v_1...v_d$ in $X$ is weighted by $w(\sigma)=w(v_0)w(v_1)\times\cdots\times w(v_{d})$. 
Now we define the inner product as 
$$
\braket{\sigma|\sigma'}= w(\sigma)w(\sigma') \delta_{\sigma,\sigma'}.
$$
With this inner product, 
the orthonormal basis for $C_d$ is given by
\begin{equation}
\label{eq:normalizedbasis}
    \left\{ \ket{[\sigma]}:= \frac{1}{w(\sigma)}  \Ket{\sigma}\right\}_{\sigma\in X_d}.
\end{equation}
Throughout this paper, we use $\ket{[\sigma]}$ for the normalized basis of the chain space. 
For the normalized basis, the boundary and coboundary act as 
$$
\partial_d \ket{[\sigma]} = \sum_{\tau \in \mathrm{face}(\sigma)} (-1)^{[\sigma:\tau]} {w([\sigma:\tau])}\ket{[\tau]}
$$
$$
\delta_d \ket{[\sigma]}= \sum_{\tau \in \mathrm{coface}(\sigma)} (-1)^{[\tau:\sigma]} {w([\tau:\sigma])}\ket{[\tau]}
$$
where $w([\sigma:\tau])$ is the weight of the vertex that is to be removed from $\sigma$ to obtain $\tau$. 


Then, the Laplacian matrix elements can be written for the orthonormal basis as 
\begin{align*}
\bra{[\sigma']}\Delta_d^{\mathrm{up}}\ket{[\sigma]}=
    \left\{
    \begin{array}{ll}
    \sum_{u\in \mathrm{up}(\sigma)} w(u)^2 & \text{if\ } \sigma=\sigma', \\
    w(v_\sigma)w(v_{\sigma'}) & \text{if } \bar{\sigma} \sim_\uparrow\sigma', \\
    -w(v_\sigma)w(v_{\sigma'}) & \text{if } {\sigma} \sim_\uparrow\sigma',\\
    0 & \text{otherwise}.
    \end{array}
    \right.
    \end{align*}
\begin{align*}
\bra{[\sigma']}\Delta_d^{\mathrm{down}}\ket{[\sigma]}=
    \left\{
    \begin{array}{ll}
    \sum_{v\in \sigma} w(v)^2 & \text{if\ } \sigma=\sigma', \\
    w(v_\sigma)w(v_{\sigma'}) & \text{if } \bar{\sigma} \sim_{\downarrow} \sigma', \\
    -w(v_\sigma)w(v_{\sigma'}) & \text{if } \sigma\sim_{\downarrow} \sigma',\\
    0 & \text{otherwise}.
    \end{array}
    \right.
\end{align*}

\begin{align*}
\bra{[\sigma']}\Delta_{d}\ket{[\sigma]}=
    \left\{
    \begin{array}{ll}
    \left( \sum_{u\in \mathrm{up}(\sigma)} w(u)^2 \right) + \left( \sum_{v\in \sigma} w(v)^2 \right) & \text{\ if\ } \sigma=\sigma', \\
    w(v_\sigma)w(v_{\sigma'}) & 
    \begin{array}{l}
    \text{if } \bar{\sigma}\sim_{\downarrow}\sigma' \text{ and } \sigma, \sigma' \text{ do not share a coface}, \\
    \end{array}
     \\
    -w(v_\sigma)w(v_{\sigma'})  &     
    \begin{array}{l}
    \text{if } {\sigma}\sim_{\downarrow}\sigma' \text{ and } \sigma, \sigma' \text{ do not share a coface} ,
    \end{array}\\
    0 & \text{\ otherwise}.
    \end{array}
    \right.
\end{align*}
Here, $u\in \mathrm{up}(\sigma)$ are vertices s.t. $\sigma\cup\{v\} \in X_{d+1}$ and $v\in \sigma$ are vertices in $\sigma$. For $\sigma, \sigma'$ that share a common face, $v_\sigma$ and $v_{\sigma'}$ are the vertices to be removed from $\sigma$ and $\sigma'$ to obtain the common face.

\subsection{Orientability and disorientability of simplicial complexes}
\label{sec:preliminaries:orientable}

The concept of orientation has been discussed in the literature of simplicial complexes, which are also crucial in our work. 
We introduce the existing concepts concerning the orientations of simplices, although they are different from those utilized in our work. 

The {\it orientability} and {\it disorientability} of simplices are defined as follows. 
A $d$-dimensional simplicial complex is said to be disorientable if there is a choice of orientations of simplices such that whenever two $d$-dimensional simplices intersect, they induce the same orientation on the shared $d-1$-dimensional simplex~\cite{mukherjee2016random}. This is a concept defined for maximal simplices and considered to be a higher-order generalization of bipartiteness~\cite{mukherjee2016random,eidi2024higher}. 
A $d$-dimensional simplicial complex is said to be orientable if there is a choice of orientations s.t. if two $d-1$-dimensional simplices share a $d$-dimensional coface, they induce opposite orientations on it. 
Orientability is related to the manifold-like property of simplicial complexes. For example, for the triangulation  of a 2-dimensional manifold, one can imagine a choice of orientations for the triangles such that they induce opposite orientations on their common edges. 
Note that the (up)-degree of $d-1$-dimensional simplices is at most 2, so that the simplicial complex can be orientable, i.e., any $d-1$-simplices can be a face of at most two $d$-dimensional simplices.

\section{Problem definition and main result}
\label{sec:problem_main}

In this section, we introduce the concept of orientable filtration, our main problem, and our main result.
Let $X$ be a $D$-dimensional simplicial complex, i.e., a simplicial complex whose maximal simplices are $D$-dimensional.  
Let
$$
X^0 \subseteq X^1 \subseteq \cdots \subseteq X^N
$$
be a filtration of simplicial complexes i.e., each of $\{X^j\}_{i=0}^N$ forms a subcomplex of a $d$-dimensional simplicial complex $X$.
The filtration induces filtrations of $d$-simplices for $d\in [D]$ as
$$
X_d^0 \subseteq X_d^1 \subseteq \cdots \subseteq X_d^N,
$$
where $X_d^i$ is the set of $d$-simplices in $X^i$. 
Then,  $$\tilde X_d^i:=X_d^i\backslash X_d^{i-1}$$ is the set of $d$-simplices that are added at the $i$-th level of the filtration. 

We would like to define a notion of orientability for the filtration of a simplicial complex for a dimension $d\in [n-1]$, that is defined through orientability for subsets that appear in the filtration. 
Roughly speaking, the notion of orientable filtration concerns (1) orientability of $d$-simplices that are added at each of the steps of the filtration, and (2) orientability among the different subsets in the filtration:
\begin{enumerate}
    \item Each of $\tilde X_d^i$ is orientable: inside of $\tilde X_d^i$, $d$-simplices induces opposite orientations on their common faces in $X_{d-1}$. (\textbf{down-orientability})
    \item For adjacent subsets $\tilde X_d^i$ and $\tilde X_d^j$, $d$-simplices induces opposite orientations on the common cofaces in $X_{d+1}$ that ``connects'' subsets $\tilde X_d^i$ and $\tilde X_d^j$. (\textbf{up-orientability})
\end{enumerate}
An example of such filtrations can be seen in Figure~\ref{fig:orientable_filtration}.  
\begin{figure}
     \centering
     \includegraphics[width=0.5\linewidth]{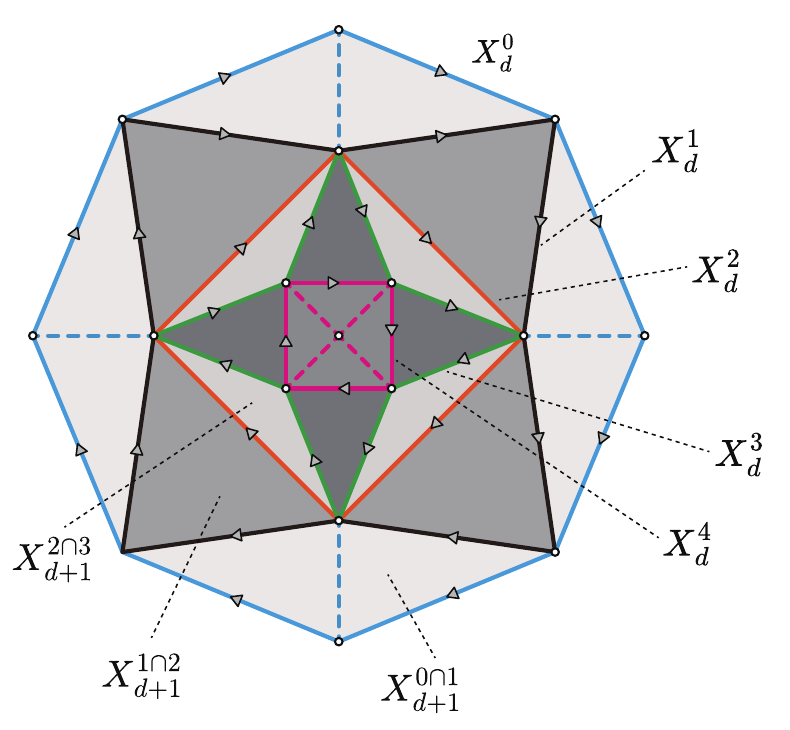}
     \caption{An example of an orientable filtration ${X}_d^0\subseteq \cdots \subseteq {X}_d^4$ with $d=1$. 
     Edges with the same color belong to the same $\tilde X_1^i$. 
     Arrows indicate the orientation that makes simplices in different subsets induce opposite orientations on the triangles. 
     The dashed edges are not the elements in $\hat X_1^i$. 
     They are regarded as ``internal'' edges, and the adjacent triangles can be effectively regarded as a single cell.
     Here, $f^{0,1}=1,f^{1,0}=1,f^{1,2}=2,f^{2,3}=1,f^{3,2}=2$ and so on. 
     }
     \label{fig:orientable_filtration}
 \end{figure} 
These are basic properties of orientable filtration. 
However, in order to treat more general situations, 
we will need some additional technical requirements.
Specifically, we would like to include situations in which there are simplices that we do not need to care about the orientation, which we call internal simplices:
\begin{definition}[Internal simplices of $X_d^i$]
    For $d\in[D]$ and $i\in [N]$, we say that $\sigma \in \tilde X_d^i$ is an internal simplex of $\tilde X_d^i$ if there are two distinct $d$ simplices $\sigma',\sigma'' \in \tilde X_d^i$ s.t. there are distinct cofaces of $\sigma'$ and $\sigma''$ in any $d+1$-simplices in $X$ that contains $\sigma$ as a face. 
    We define $\hat X_d^i$ to be the set of $d$ simplices that is obtained by removing internal simplices from $\tilde X_d^i$.
\end{definition}

An example of such internal simplices can be seen in Figure~\ref{fig:orientable_filtration} as dashed edges. 
We say the disjoint union $\bigsqcup_{i=0}^N\hat X_d^i$ to be the set of {\it essential $d$-simplices} because they are mostly relevant to the homology problem with orientable filtration. 
Now, we are ready to define a notion of orientability for a filtration of a simplicial complex. 
The motivation behind this definition is that such internal simplices indeed do not appear in the harmonics of interest (therefore, we do not need to care about the orientations for them).
We formally define an orientable filtration as below.  

\begin{definition}[Orientable $d$-filtration]
\label{def:orientable_filtration}
A simplicial complex $X$ is said to have an orientable $d$-filtration of length $N$ if there is a filtration
$
X^0 \subseteq X^1 \subseteq \cdots \subseteq X^N\subseteq X
$ where $N\in poly(n)$
and 
a choice of orientations for $X_d$
s.t. the following conditions hold:
\begin{itemize}
    \item Each of the subset of essential simplices $\{\hat X_{d}^{i}\}_{i=0}^N$ is orientable i.e., for any $i\in \{0,...,N\}$, whenever two simplices in $\hat X_{d}^{i}$ share a $d-1$-dimensional face, they induces opposite orientations on it. 
    \item For any $i,j\in \{0,...,N-1\}$, whenever two simplices in $\hat X_{d}^{i}$ and {$\hat X_{d}^{j}$} share a $d+1$-dimensional coface $\tau$, they induces opposite orientations on it. Moreover, the faces of such $\tau$ is partitioned into $\tilde X_{d}^{i}$ and $\tilde X_{d}^{j}$.
\end{itemize}
\end{definition}

We say that a filtration and orientations for $X_d$ is an \textit{oriented filtration} if they satisfy the conditions in Definition~\ref{def:orientable_filtration}. 
We remark that the (down) orientability for $\hat X_{d}^{i}$ implies that $\hat X_{d}^{i}$ is down degree-2, i.e., any $d-1$ simplices is a face of at most two simplices in $\hat X_{d}^{i}$. 
Similarly, there is a degree-2-like property in the upward connectivity: the faces of any $d+1$ simplex are contained in at most two subsets in the filtration $\tilde X_{d}^{i}$ and $\tilde X_{d}^{j}$.

We also define a ``uniform version'' of the above orientable filtration. 
Let us denote $X_{d+1}^{i\cap j}$ as the set of $d+1$ simplices whose faces are partitioned into $\tilde X_d^i$ and $\tilde X_d^{j}$. 
In other words, $X_{d+1}^{i\cap j}$ is the set of $d+1$-simplices that makes $\tilde X_d^i$ and $\tilde X_d^{j}$ adjacent as 
\begin{align*}
& X_{d+1}^{i\cap j}\\
    \tilde X_d^i \nearrow &\hspace{1cm} \nwarrow\tilde X_d^j
\end{align*}
where the arrows show inclusion of simplices as cofaces. 
An example can be seen in Figure~\ref{fig:orientable_filtration}. 
Then, the uniform version of the orientable filtration is defined as follows. 

\begin{definition}[Uniform orientable $d$-filtration]
\label{def:uniform_filtration}
A simplicial complex $X$ is said to have a uniform orientable $d$-filtration 
if it has an oriented $d$-filtration and with that filtration and orientations, 
    \begin{itemize}
        \item Each of the simplices in $\tilde X_d^i$ has a uniform weight $w^i$. 
        \item 
       For any $i,j$ and any $\tau\in X_{d+1}^{i\cap j}$,
        the number of faces in $\hat X_d^i$ is the same. We will denote the number of faces of $\tau\in X_{d+1}^{i\cap j}$ in $\hat X_d^i$ by $f^{i,j}$, i.e.,
        $$
        f^{i,j} = |\{\sigma \in \mathrm{face(\tau)}: \sigma \in \hat X_d^i \}|. 
        $$
        (The number of faces in $\hat X_d^j$ is denoted by $f^{j,i}=|\{\sigma \in \mathrm{face(\tau)}: \sigma \in \hat X_d^j \}|$.)
    \end{itemize}
\end{definition}
We say that a filtration and orientations for $X_d$ is a \textit{uniform oriented filtration} if they satisfy the conditions in Definition~\ref{def:uniform_filtration} as uniformities for weights and relative degree $f^{i,j}$ are imposed. 
Figure~\ref{fig:orientable_filtration} indeed describes a uniform case. 
The relative degree $f^{i,j}$ is important in the construction of harmonics because we need to cancel out the contributions on the cofaces between different layers.

Then, the main problem that we study in this paper is defined as follows. 

\begin{definition}[Promise Clique Homology problem with a uniform orientable filtration]
\label{def:problem} \ \\ 
\textbf{Input:} 
\begin{itemize}
    \item $d\in [n-1]$, $\epsilon >1/poly(n)$.
    \item
    An oriented clique complex $\mathrm{Cl}_d(G)$ over $n$-vertices described by a weighted graph $G$. 
    and a filtration 
 $
X_d^0 \subseteq X_d^1 \subseteq \cdots \subseteq X_d^N \subseteq \mathrm{Cl}_d(G)
$
of length $N\in poly(n)$. 
The orientations and filtration are specified by the access to the following classical circuits:
\begin{itemize}
    \item We are given a $poly(n)$-size description of a classical circuit that returns an oriented clique in $\mathrm{Cl}_d(G)$ by inputting an unoriented clique composed of the same vertices. 
    \item We are also given a $poly(n)$-size  
description of a classical circuit that returns an index of the filtration at which a simplex $\sigma \in X_d$ is added, or returns $\perp$ if $\sigma \in \mathrm{Cl}_d(G) \backslash X_d^N$.
\end{itemize}
\end{itemize}
\textbf{Promise:} Either of the below holds:
\begin{itemize}
    \item (YES instance) There is a non-negative and homologous cycle $\ket{c}$ supported on $\hat X_d^0$. 
   Moreover, the harmonics $\ket{\phi}$ that represent the same hole as $\ket{c}$ are supported only on $X_d^N$.

    \item (NO instance)
    The minimum eigenvalue of $\Delta_d$ satisfies
    $\lambda(\Delta_d)\geq \epsilon$
\end{itemize}
\textbf{Output:} 
1 for YES instances and 0 for NO instances.
\end{definition}
Note that we say that a cycle $\ket{c}\in \ker \partial_d$ is homologous if $\ket{c}\notin \mathrm{Im}\partial_{d+1}$. 

Essentially, our problem is a {\it persistence} problem \cite{gyurik2024quantum} in which we are required to decide whether a hole (supposed to exist) in $X_d^0$ finally persists in $X_d^N$ or not under the condition of the uniform orientable filtration as well as the promises. 
There is a possibility that even if a hole in $X_d^0$ does not persist in $X_d$, there are still some holes in $X_d^N$. 
For example, a new hole can be born in the process of the filtration.  
However, we are prohibiting the existence of {\it any} holes in the NO instances.  
This will be an important point in the gadget construction in the reduction from an MA-hard problem. 

Our main result can be stated as follows. 

\begin{theorem}
\label{thm:main}
    The problem in Definition~\ref{def:problem} is MA-complete.
\end{theorem}

The proof will be provided in the subsequent sections. 
In Section~\ref{sec:containment_MA}, we prove the containment in MA. 
In Sections~\ref{sec:hardness_1},~\ref{sec:hardness_2},~\ref{sec:hardness_3}, we prove MA-hardness. 
Section~\ref{sec:hardness_1} provides the construction of the clique complex from an MA-hard problem with new gadget constructions.
In Section~\ref{sec:hardness_2}, we prove the lower bound for the minimal eigenvalue of the Laplacian in the NO instances. 
In Section~\ref{sec:hardness_3}, we prove the requirement for the homologous cycle and harmonics in YES instances.

\section{Containment in MA}
\label{sec:containment_MA}

In this section, we prove the containment of the homology problem with uniform orientable filtration in MA.

\subsection{Fixed-node construction for non-negative states}
\label{sec:inMA:1}

We first recall the definition of the fixed-node Hamiltonian for the chain space. 
In Definition~\ref{def:fixed_node_Hamiltonian}, we have introduced the fixed-node Hamiltonian for $H\in \mathbb{R}^{2^n\times 2^n}$. 
Here, instead of $n$-qubit space, we consider $H: C_d(X)\rightarrow C_d(X)$ where $C_d(X)$ is the space spanned by {\it oriented} $d$-simplices as introduced in Section~\ref{sec:preliminaries:sc}. 
We only consider real Hamiltonians on $C_d(X)$. Therefore, it holds 
$$
\bra{[\sigma]} H\ket{[\sigma']} \in \mathbb{R}
$$
for any $\sigma,\sigma' \in X_d$, where $\ket{[\sigma]}$ is the normalized basis state defined in eq.~\eqref{eq:normalizedbasis}. 

Then, the fixed-node Hamiltonian can be formulated as 
$$
\bra{[\sigma]}F^{H,\psi}\ket{[\sigma']}= 
\begin{cases}
    0  &\text{ if } (\sigma,\sigma')\in S^+,\\ 
    \bra{[\sigma]}H\ket{[\sigma']}&\text{ if } (\sigma,\sigma')\in S^-,\\ 
    \bra{[\sigma]}H\ket{[\sigma']}+ \sum_{\sigma'':(\sigma,\sigma'')\in S^+} \bra{[\sigma]}H\ket{[\sigma'']} \frac{\braket{[\sigma'']|\psi}}{\braket{[\sigma]|\psi}} &\text{ if } x=y.\\ 
\end{cases}
$$
where for $\sigma,\sigma' \in X_d$, 
$$
S^+:=\{(\sigma,\sigma')\ :\ \sigma\neq \sigma' \text{ and } \braket{\psi|\sigma}\bra{\sigma}H\ket{\sigma'}\braket{\sigma'|\psi}>0\},
$$
$$
S^-:=\{(\sigma,\sigma')\ :\ \sigma\neq \sigma' \text{ and } \braket{\psi|\sigma}\bra{\sigma}H\ket{\sigma'}\braket{\sigma'|\psi}\leq0\}.
$$

We will apply the fixed-node construction for the combinatorial Laplacian w.r.t. some {\it non-negative state} $\ket{\psi}$.
When $\ket{\psi}$ is a non-negative state, 
$$
(\sigma,\sigma') \in S^+ \Leftrightarrow \bra{\sigma}H\ket{\sigma'}>0 \text{ and } \sigma,\sigma' \in \mathrm{Supp}(\ket{\psi})
$$
and
$$
(\sigma,\sigma') \in S^- \Leftrightarrow \bra{\sigma}H\ket{\sigma'}\leq0 \text{ or } \sigma,\sigma' \notin \mathrm{Supp}(\ket{\psi}).
$$
In the subsequent sections, we identify a suitable choice for the fixed-node state.

\subsection{Local decomposition of up Laplacian and down Laplacian}
\label{sec:inMA:2}

In this subsection, we define {\it good} and {\it bad} simplices. 
For $\tau \in X_{d-1}$, define a restriction of the boundary operator $\partial_{d,\tau}:C_d(X)\rightarrow C_{d-1}(X)$ by 
$$
\partial_{d,\tau}:= 
\Pi_\tau\partial_d
$$
where $\Pi_\tau:= \ket{[\tau]}\bra{[\tau]}$ and $\partial_d$ is the boundary operator.  
Similarly, for $\tau \in X_{d+1}$, define a restriction of the coboundary operator by
$$
\delta_{d,\tau}:= 
\Pi_\tau\delta_d
$$
where $\delta_d$ is the coboundary operator. 
Then, the following lemma about the decompositions of the up and down Laplacians holds.

\begin{lemma}
For any integer $0\leq d \leq D-1$, the following equations hold:
\begin{enumerate}
    \item         $
        \delta_d= \sum_{\tau \in X_{d+1}} \delta_{d,\tau}
        $
        \item 
         $
        \partial_d= \sum_{\tau \in X_{d}} \partial_{d,\tau}
        $
        \item 
        $
    \Delta_d^{\mathrm{up}}=\sum_{\tau\in X_{d+1}} \Delta_{d,\tau}^{\mathrm{up}}
    $
    \item  $
    \Delta_d^{\mathrm{down}}=\sum_{\tau\in X_{d-1}} \Delta_{d,\tau}^{\mathrm{down}}
    $
\end{enumerate}
    where
    $
    \Delta_{d,\tau}^{\mathrm{up}}:= (\delta_{d,\tau})^*\delta_{d,\tau}
    $
    and 
    $
    \Delta_{d,\tau}^{\mathrm{down}}
    := (\partial_{d,\tau})^*  \partial_{d,\tau}.
    $
\end{lemma}

\begin{proof}
Because 
$$
\sum_{\tau \in X_{d-1}} \Pi_\tau = I(C_{d-1}(X))
$$
where $I(C_{d-1}(X))$ is the identity on $C_{d-1}(X)$, 
$$
\partial_d= \sum_{\tau \in X_{d-1}} \Pi_\tau \partial_d= \sum_{\tau \in X_{d}} \partial_{d,\tau}. 
$$
Moreover, 
\begin{align*}
    \Delta_{d,\tau}^{\mathrm{down}}= \partial_d^\dagger \partial_d 
    = \sum_{\tau,\tau' \in X_{d-1}}\partial_d^\dagger \Pi_{\tau'}  \Pi_\tau \partial_d
    = \sum_{\tau \in X_{d-1}}\partial_d^\dagger \Pi_{\tau}  \Pi_\tau \partial_d = \sum_{\tau \in X_{d-1}}\partial_{d,\tau}^\dagger  \partial_{d,\tau}. 
\end{align*}
The remaining two claims can be shown similarly. 
\end{proof}

We say that $\ket{\psi}\in C_d(X)$ is a non-negative ground state of $\Delta_{d,\tau}^{\mathrm{down}}$ if $\ket{\psi}\neq 0$ is a (non-zero and) non-negative state s.t. $\Delta_{d,\tau}^{\mathrm{down}}\ket{\psi}=0$.
Similarly, we say that $\ket{\psi}\in C_d(X)$ is a non-negative ground state of $\Delta_{d,\tau}^{\mathrm{up}}$ if $\ket{\psi}\neq0$ is a non-negative state s.t. $\Delta_{d,\tau}^{\mathrm{up}}\ket{\psi}=0$. 
Examples for non-negative local ground states can be seen in Figure~\ref{fig:example_non_negative_gs}.

\begin{figure}
    \centering
    \includegraphics[width=0.55\linewidth]{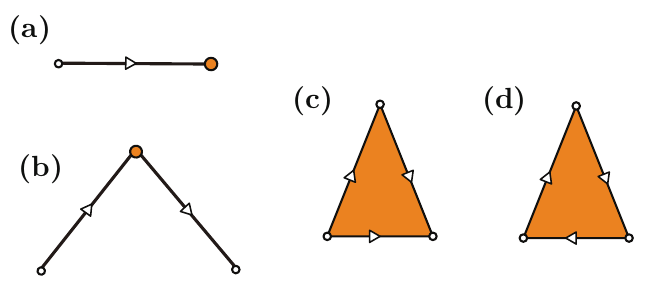}
    \caption{Examples for non-negative ground states in the 1-dimensional case. The arrows on the edge indicate the given orientation.  (a) The down Laplacian for the colored vertex does not have a non-negative ground state, while (b) has a non-negative ground state. The up Laplacian for the colored triangle in (c) has a non-negative ground state, while (d) does not. }
    \label{fig:example_non_negative_gs}
\end{figure}

\begin{definition}[Good and bad simplices]
\label{def:good_bad}
For a given filtration and orientation of $X_d$ that forms an orientable filtration $
X_d^0 \subseteq X_d^1 \subseteq \cdots \subseteq X_d^N
$, we say that $\sigma\in  X_d^N$
is a good simplex if the following two conditions hold:
\begin{itemize}
    \item For any $\tau \in \mathrm{face}(\sigma)$, 
there is a non-negative ground state for $\Delta_{d,\tau}^{\mathrm{down}}$ supported on $\hat{X}_d^i$ for $i$ s.t. $\sigma \in \tilde X_d^i$. 
\item 
For any $\tau\in \mathrm{coface}(\sigma)$, 
there is a non-negative ground state for $\Delta_{d,\tau}^{\mathrm{up}}$ supported on $\bigsqcup_{i=0}^N\hat{X}_d^i$ 
\end{itemize}
Simplices in $X_d^N$ that are not good are called bad simplices. 
\end{definition}

Examples of good and bad simplices can be seen in Figure~\ref{fig:eg_good_bad}. It can be observed that under the uniform orientable filtration, the badness of simplices tells us that ``there is no hole'' with local information.  

\begin{claim}
\label{claim:good_bad}
    For any $\sigma \in \bigsqcup_{i=0}^N\hat{X}_d^i$, we can efficiently decide whether $\sigma$ is good or bad. Then, we can efficiently check the conditions for every simplex in $\mathrm{face}(\sigma)$ and $\mathrm{coface}(\sigma)$.
\end{claim}

\begin{proof}
    We first compute the index $i$ s.t. $\sigma \in \hat X_d^i$. and list all the adjacent simplices. 
    Then, 
\end{proof}

\subsection{Connected components of good simplices}
\label{sec:inMA:3}

In this section, we introduce $\ket{\phi_{\text{good}}^{\sigma_0}}$, which will be used as a ``fixed-node''. 
Let $$
X^0_d \subseteq X^1_d \subseteq \cdots \subseteq X^N_d
$$ be a filtration and let us suppose that $X_d$ is oriented such that the filtration forms a uniform oriented filtration. 
Let $S_{\text{good}}(X_d)$ be the set of good simplices in $X_d$. 

Let 
$\sigma_0$ be an arbitrary simplex in $\hat X^0_d$. 
Then, let $S_{\text{good}}^{\sigma_0}\subseteq S_{\text{good}}(X_d)$ be the set of good simplices connected to $\sigma_0$. 
Here, we say that two good simplices $\sigma,\sigma'$ are connected if there is a sequence of connected simplices 
$\sigma \sim \eta_1 \sim \eta_2 \sim \cdots \sim \eta_m \sim \sigma'$ 
where $\eta_1,...,\eta_m \in S_{\text{good}}(X_d)$ and $\sigma\sim \eta$ means that $\sigma$ and $\eta$ have a common face.
Then, let 
$X_d^{\sigma_0}(0)$ be the set of good simplices in $\hat X_d^0$ connected to $\sigma_0$. (Connected only through elements in $\hat X_d^0$, i.e., there is a sequence $\sigma \sim \eta_1 \sim \eta_2 \sim \cdots \sim \eta_m \sim \sigma'$ s.t. $\eta_1,...,\eta_m \in \hat X_d^0$.)
Now, $X_d^{\sigma_0}(0)$ induces a filtration
    $$
    X_d^{\sigma_0}(0) \subseteq X_d^{\sigma_0}(1) \subseteq X_d^{\sigma_0}(2) \subseteq \cdots \subseteq X_d^{\sigma_0}(M) 
    $$
    where at each step $t+1$, $d$-simplices in $\bigsqcup_{i=1}^N\hat{X}_d^i$
    that share cofaces with simplices in $X_d^{\sigma_0}(t)$ are added to form $X_d^{\sigma_0}(t+1)$. 
    Formally:
    $$
    \hat X_d^{\sigma_0}(i):= X_d^{\sigma_0}(i)\backslash X_d^{\sigma_0}(i-1) 
    $$
    and
    $$
    \hat X_d^{\sigma_0,j}(i):= \{\sigma \in \tilde X_d^{\sigma_0}(i) : \sigma \in \hat X_d^j \}.
    $$
    The index $j$ is required above because there can be a ``branching'' into several subsets with different indices in the filtration. 

We define  $\ket{\phi_{\text{good}}^{\sigma_0}}$ as 
$$
\Ket{\phi_{\text{good}}^{\sigma_0}}:= 
\sum_{i=0}^M \ket{c^{\sigma_0}(i)},
$$
where each
$\ket{c^{\sigma_0}(i)}$ is constructed as 
$$\ket{c^{\sigma_0}(0)}=\sum_{\sigma\in \hat X_d^{\sigma_0}(0)} \ket{\sigma} 
    \left(=\sum_{\sigma\in \hat X_d^{\sigma_0}(0)} w(\sigma)\ket{[\sigma]}\right) $$
    with
    $$
    \ket{c^{\sigma_0}(i)}
    = \sum_j \sum_{\sigma \in \hat X_d^{\sigma_0,j}(i)} \frac{|\hat X_d^{\sigma_0}(0)|}{|\hat X_d^{\sigma_0,j}(i)|} \ket{\sigma}
    \left(=
    \sum_j \sum_{\sigma \in \hat X_d^{\sigma_0,j}(i)} \frac{w(\sigma)|\hat X_d^{\sigma_0}(0)|}{|\hat X_d^{\sigma_0,j}(i)|} \ket{[\sigma]}\right).
    $$
    Note that $\ket{\phi_{\text{good}}^{\sigma_0}}$ is a non-normalized state.

    Therefore, $\ket{\phi_{\text{good}}^{\sigma_0}}$ is composed of non-negative terms $\left\{\sum_{\sigma \in \hat X_d^{\sigma_0,j}(i)} \frac{|\hat X_d^{\sigma_0}(0)|}{|\hat X_d^{\sigma_0,j}(i)|} \ket{\sigma}\right\}_{i,j}$.
We denote 
$$
S^{\sigma_0}_{\text{good}}:= \mathrm{Supp}(\phi_{\text{good}}^{\sigma_0})=X_d^{\sigma_0}(M) 
$$
where $\mathrm{Supp}(\phi_{\text{good}}^{\sigma_0})$ is the support of $\phi_{\text{good}}^{\sigma_0}$ defined in eq.~\eqref{eq:def:support}.

An example of a state that can be constructed in this way is shown in Figure~\ref{fig:harmonics_eg}. In this example, there is no branching, and there are also no bad simplices. The coefficient for the cycle comes from the relative number of simplices between each of the layers. 
\begin{figure}
    \centering
    \includegraphics[width=1\linewidth]{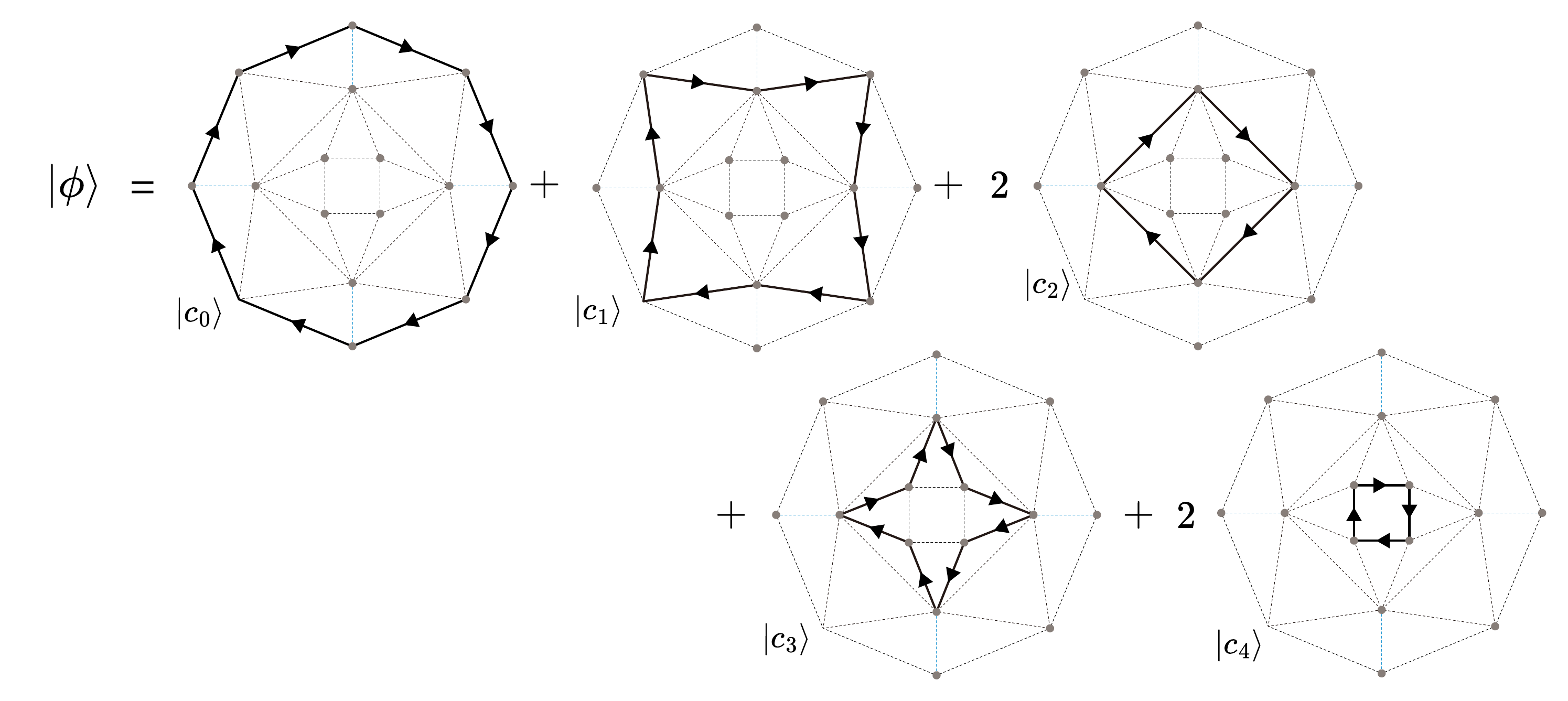}
    \caption{An example of a harmonic state associated with the uniform orientable filtration. Here, we are considering an unweighted setting. The coefficients are determined by the contributions of coboundaries on the shared cofaces, which can be calculated with eq.~\eqref{eq:relative_amplitude}.}
    \label{fig:harmonics_eg}
\end{figure}

We can observe the following properties of $\ket{\phi_{\text{good}}^{\sigma_0}}$: 
first, in each of $\sum_{\sigma \in \hat X_d^{\sigma_0,j}(i)}\ket{\sigma}$, the boundaries vanish on faces shared between simplices. (If $\sum_{\sigma \in \hat X_d^{\sigma_0,j}(i)}\ket{\sigma}$ is a cycle, the boundary completely vanishes.)
 Second, for any two adjacent components (i.e., two components whose simplices in the support share cofaces) 
    $$\sum_{\sigma \in \hat X_d^{\sigma_0,j}(i)} \frac{|\hat X_d^{\sigma_0}(0)|}{|\hat X_d^{\sigma_0,j}(i)|} \ket{\sigma} \ \
    \text{ and  } 
    \sum_{\sigma \in \hat X_d^{\sigma_0,j'}(i+1)} \frac{|\hat X_d^{\sigma_0}(0)|}{|\hat X_d^{\sigma_0,j'}(i+1)|} \ket{\sigma},
    $$ the coboundary vanishes on their common cofaces. 
    It also holds that 
    $$
    \frac{|\hat X_d^{\sigma_0,j'}(i+1)|}{|\hat X_d^{\sigma_0,j}(i)|}=\frac{f^{j',j}}{f^{j,j'}}
    $$
    by the condition on the degree in the uniform orientable filtration. 
Therefore, for two simplices $\sigma,\sigma' \in S^{\sigma_0}_{\text{good}}$ that share a coface s.t. $\sigma \in \hat X_d^i$ and $\sigma' \in \hat X_d^j$ for some $i\neq j$, it holds that
$$
\frac{\braket{\sigma'|\phi_{\text{good}}^{\sigma_0}}}{\braket{\sigma|\phi_{\text{good}}^{\sigma_0}}}=\frac{f^{i,j}}{f^{j,i}}
$$
and therefore, 
\begin{equation}
\label{eq:relative_amplitude}
\frac{\braket{[\sigma']|\phi_{\text{good}}^{\sigma_0}}}{\braket{[\sigma]|\phi_{\text{good}}^{\sigma_0}}}=\frac{f^{i,j}w(\sigma)}{f^{j,i}w(\sigma')}\in poly(n).
\end{equation}
Using this relationship, we can compute the relative amplitude of adjacent simplices efficiently. 

Next, we show that the relative amplitudes for adjacent simplices are bounded by $poly(n)$. 
\begin{claim}
\label{claim:relative_amplitude}
    For any $\sigma \sim \sigma'$ s.t. $\sigma \in \hat X_d^i$ and $\sigma' \in \hat X_d^j$ for some $i,j \in [N]$,  $$\frac{\braket{[\sigma']|\phi_{\text{good}}^{\sigma_0}}}{\braket{[\sigma]|\phi_{\text{good}}^{\sigma_0}}} \in poly(n)$$
    and it can be efficiently evaluated. 
\end{claim}

\begin{proof}
    For the case $\sigma,\sigma'$ shares a coface, we have already checked the claim in eq.~\eqref{eq:relative_amplitude}.
    There are cases when $\sigma \in \hat X_d^i$ and $\sigma' \in \hat X_d^j$ do not share a coface but share a face. 
For such cases, there are intermediate simplices 
$
\sigma \sim \cdots \sim  \sigma'
$
where intermediate simplices are either (1) in $\hat X_d^i$ or  $\hat X_d^j$ or (2) in in $\hat X_d^i$, $\hat X_d^k$ or  $\hat X_d^j$ for some $k \neq i,j$. 

In case (1), the situation is almost the same as eq.~\eqref{eq:relative_amplitude} because the amplitudes in the same subset $\hat X_d^i$ are the same. 

In case (2), 
$$\frac{\braket{[\sigma']|\phi_{\text{good}}^{\sigma_0}}}{\braket{[\sigma]|\phi_{\text{good}}^{\sigma_0}}}
= \frac{\braket{[\sigma']|\phi_{\text{good}}^{\sigma_0}}}{\braket{[\sigma'']|\phi_{\text{good}}^{\sigma_0}}} 
\cdot \frac{\braket{[\sigma'']|\phi_{\text{good}}^{\sigma_0}}}{\braket{[\sigma]|\phi_{\text{good}}^{\sigma_0}}}
$$
where $\sigma'' \in \hat X_d^k$. 
In this case, it is clear that the relative amplitude is bounded by $poly(n)$ and efficiently computable. 
\end{proof}

We show that in YES instances, $\ket{\phi_{\text{good}}^{\sigma_0}}$ is a superposition of homologous cycles that represent the same hole, and it is indeed a harmonic state in $\ker \Delta_d$. 

\begin{claim}
\label{claim:good_supp}
In the yes instance, 
there exists $\sigma_0$ s.t. 
$\ket{\phi_{\text{good}}^{\sigma_0}}\in \ker(\Delta_d)$. 
Moreover, there is no bad simplices in $S_{\text{good}}^{\sigma_0}$ and simplices that are  adjacent to simplices in $S_{\text{good}}^{\sigma_0}$. 
\end{claim}

\begin{proof}
    In the YES instances, there is a non-negative homologous cycle on $\hat X_d^0$. 
    Therefore, by choosing $\sigma_0$ supported on the cycle, $\ket{c^{\sigma_0}}$ will be a cycle. 
    By the condition of the uniform orientable filtration and by the construction, $\{\ket{c^{\sigma_0}(i)}\}_i$ are also cycles. 
    Moreover, $\ket{\phi_{\text{good}}^{\sigma_0}}$ is constructed such that the contribution of the coboundary on the shared cofaces between $\ket{c^{\sigma_0}(i)}$ and $\ket{c^{\sigma_0}(i+1)}$ cancel out. 
    If there is a bad simplex in the adjacent simplices of $\mathrm{Supp}(\ket{\phi_{\text{good}}^{\sigma_0}})$, then there is a cycle  
    $$\sum_{\sigma \in \hat X_d^{\sigma_0,j}(i)} \frac{|\hat X_d^{\sigma_0}(0)|}{|\hat X_d^{\sigma_0,j}(i)|} \ket{\sigma}$$
    that appears as a boundary. This implies that $\ket{c^{\sigma_0}}$ is also a boundary, which is a contradiction. 
\end{proof}

\subsection{Fixed-node Laplacian and the Markov transition Matrix}
\label{sec:transition_matrix}
Define the fixed-node Laplacian relative to the state $\ket{\phi_{\text{good}}^{\sigma_0}}$, which is determined by $\sigma_0\in \hat X_d^0$, by 
$$\bra{[\sigma]}F^{\sigma_0}_d\ket{[\sigma']}:=
\begin{cases}
    0  &\text{ if } (\sigma,\sigma')\in S^+,\\ 
    \bra{[\sigma]}\Delta_d\ket{[\sigma']}&\text{ if } (\sigma,\sigma')\in S^-,\\ 
    \bra{[\sigma]}\Delta_d\ket{[\sigma]}+ \sum_{\sigma'':(\sigma,\sigma'')\in S^+} \bra{[\sigma]}\Delta_d\ket{[\sigma'']} \frac{\Braket{[\sigma'']|\phi_{\text{good}}^{\sigma_0}}}{\Braket{[\sigma]|\phi_{\text{good}}^{\sigma_0}}} &\text{ if } \sigma=\sigma'.\\ 
\end{cases}
$$
Because $\ket{\phi_{\text{good}}^{\sigma_0}}$ is a non-negative state, for $\sigma,\sigma' \in \mathrm{Supp}(\ket{\phi_{\text{good}}})$
$$
(\sigma,\sigma') \in S^+ \Leftrightarrow \bra{\sigma}\Delta_d\ket{\sigma'}>0
$$
and 
$$
(\sigma,\sigma') \in S^- \Leftrightarrow \bra{\sigma}\Delta_d\ket{\sigma'}\leq0.
$$

Recall the Laplacian matrix element
\begin{align*}
\bra{[\sigma']}\Delta_{d}\ket{[\sigma]}=
    \left\{
    \begin{array}{ll}
    \sum_{u\in \mathrm{up}(\sigma)} w(u)^2  +  \sum_{v\in \sigma} w(v)^2  & \text{\ if\ } \sigma=\sigma', \\
    w(v_\sigma)w(v_{\sigma'}) & 
    \begin{array}{l}
    \text{if } \bar{\sigma}\sim_{\downarrow}\sigma' \text{ and } \sigma, \sigma' \text{ do not share a coface}, \\
    \end{array}
     \\
    -w(v_\sigma)w(v_{\sigma'})  &     
    \begin{array}{l}
    \text{if } {\sigma}\sim_{\downarrow}\sigma' \text{ and } \sigma, \sigma' \text{ do not share a coface} ,
    \end{array}\\
    0 & \text{\ otherwise},
    \end{array}
    \right.
\end{align*}
where $\bar{\sigma}$ is the same simplex with $\sigma$ with the opposite orientation.
It should be noted that the Laplacian element will be zero if $\sigma$ and $\sigma'$ share a coface.

We define a Markov transition matrix $P^{\sigma_0}$ for a given $\sigma_0$ on the state space $S^{\sigma_0}_{\text{good}}$ whose matrix elements are 
$$
P_{\sigma\rightarrow \sigma'}^{\sigma_0}= 
\frac{\Braket{[\sigma']|\phi_{\text{good}}^{\sigma_0}}}{\Braket{[\sigma]|\phi_{\text{good}}^{\sigma_0}}}
 \braket{[\sigma'] | I-\beta F_d^{\sigma_0}|[\sigma]}.
$$
Although $P^{\sigma_0}$ is dependent on $\sigma_0$ and $\beta$, we hide the dependency on $\beta$ in the expression $P^{\sigma_0}$ for simplicity. 
We confirm several properties of this matrix. 


\begin{claim}
\label{claim:beta}
There is $\beta>1/poly(n)$ s.t.  
   $P_{\sigma\rightarrow \sigma'}^{\sigma_0}\geq 0$ for all $\sigma \in \mathrm{Supp}(\ket{\phi_{\text{good}}^{\sigma_0}}$). 
\end{claim}

\begin{proof}
    For $\sigma\neq\sigma' \in \mathrm{Supp}(\ket{\phi_{\text{good}}^{\sigma_0}})$, 
    it is clear that $\bra{\sigma}F_d^{\sigma_0}\ket{\sigma'}\leq 0$. 
    It remains to prove 
    $\bra{[\sigma]} F^{\sigma_0}_d\ket{[\sigma]} \in poly(n)$, which follows from  the fact that 
    $\frac{\Braket{[\sigma'']|\phi_{\text{good}}^{\sigma_0}}}{\Braket{[\sigma]|\phi_{\text{good}}^{\sigma_0}}} \in poly(n)
    $
    for any $\sigma\in \mathrm{Supp}(\ket{\phi_{\text{good}}^{\sigma_0}})$ and $\sigma\sim \sigma'' \in \mathrm{Supp}(\ket{\phi_{\text{good}}^{\sigma_0}})$. Therefore we can choose $\beta > 1/poly(n)$ so that $P_{\sigma\rightarrow \sigma'}^{\sigma_0}\geq 0$. 
\end{proof}

\begin{claim}
    $\sum_{\sigma' \in X_d}P_{\sigma\rightarrow \sigma'}^{\sigma_0}=1$ for all  $\sigma \in \mathrm{Supp}(\ket{\phi_{\text{good}}^{\sigma_0}})$. 
\end{claim}

\begin{proof}
   \begin{align*}
   \sum_{\sigma' \in X_d}P_{\sigma\rightarrow \sigma'}^{\sigma_0} 
   &= 
   \sum_{\sigma' \in X_d}\frac{\Braket{[\sigma']|\phi_{\text{good}}^{\sigma_0}}}{\Braket{[\sigma]|\phi_{\text{good}}^{\sigma_0}}}
 \braket{[\sigma'] | I-\beta F_d^{\sigma_0}|[\sigma]}\\
 &=
 \frac{\bra{\phi_{\text{good}}^{\sigma_0}}I-\beta F_d^{\sigma_0}\ket{[\sigma]}}{{\Braket{[\sigma]|\phi_{\text{good}}^{\sigma_0}}}}.
   \end{align*}
   Therefore, it suffices to show 
   $$
   \bra{[\sigma]}F_d^{\sigma_0}\ket{\phi_{\text{good}}^{\sigma_0}}=0. 
   $$
   Indeed, $F_d^{\sigma_0}\ket{\phi_{\text{good}}^{\sigma_0}}$ has support only outside of the set of good simplices, and therefore $\bra{[\sigma]}F_d^{\sigma_0}\ket{\phi_{\text{good}}^{\sigma_0}}=0$ for $\sigma \in \mathrm{Supp}(\ket{\phi_{\text{good}}^{\sigma_0}}$. 
\end{proof}

\paragraph{Computation of matrix elements}

We show that we can efficiently compute the {\it laziness} and the matrix elements of $P^{\sigma_0}$.
The laziness of the random walk for $\sigma$ (i.e., the probability of staying at $\sigma$ after the transition from $\sigma$) is
\begin{align*}
\bra{[\sigma]} P^{\sigma_0}\ket{[\sigma]}
&=1-\beta \bra{[\sigma]}F^{\sigma_0}_d\ket{[\sigma]} \\&
=
1-\beta 
\bra{[\sigma]}\Delta_d\ket{[\sigma']}
-\beta \sum_{\sigma'':(\sigma,\sigma'')\in S^+} \bra{[\sigma]}\Delta_d\ket{[\sigma'']} \frac{\Braket{[\sigma'']|\phi_{\text{good}}^{\sigma_0}}}{\Braket{[\sigma]|\phi_{\text{good}}^{\sigma_0}}}.
\end{align*}
The non-zero matrix element for $\sigma\sim \sigma'$
is 
$$
\bra{[\sigma]} P^{\sigma_0}\ket{[\sigma']}
= \beta
\frac{\Braket{[\sigma']|\phi_{\text{good}}^{\sigma_0}}}{\Braket{[\sigma]|\phi_{\text{good}}^{\sigma_0}}}
 \bra{[\sigma']}\Delta_{d}\ket{[\sigma]}.
$$
Both quantities can be evaluated by computing the relative amplitudes in $\phi_{\text{good}}^{\sigma_0}$ 
with eq.~\eqref{eq:relative_amplitude} and Claim~\ref{claim:relative_amplitude}. 

\paragraph{Efficient sampling}

There is an efficient classical algorithm that simulates a random walk according to the transition matrix $P^{\sigma_0}$. 
Suppose the current simplex is $\sigma_t$ and we want to sample $\sigma_{t+1}$ from $P^{\sigma_0}\ket{[\sigma_t]}$.
Then, consider the following procedures. 
\begin{itemize}
    \item 
    First, compute
     $$N(\sigma_t)=\{\sigma':\sigma'\sim_\downarrow \sigma_t \text{ and } \sigma',\sigma_t \text{ do not share a coface } \}.$$
    Note that $|N(\sigma_t)|=\mathcal{O}(n)$ and adjacent simplices only differ from each other with one vertex, and therefore $N(\sigma_t)$ can be computed efficiently. 
    \item Compute $P^{\sigma_0}_{\sigma_t\rightarrow \sigma_t}$ and $P^{\sigma_0}_{\sigma_t\rightarrow \sigma'}$ for all $\sigma'\in N(\sigma_t)$. 
    \item Sample $\sigma_{t+1}$ according to $P^{\sigma_0}$.
\end{itemize}

\subsection{MA protocol}
\label{sec:verification_protocol}

The MA-verification protocol is described as follows. 

\begin{itemize}
    \item 
    The prover sends a simplex $\sigma_0 \in X_d$. 
    \item Verifier rejects if there is no $j$ s.t. $\sigma_0 \in \hat X_d^j$ or $\sigma_0$ is not a good simplex. \\
    {This can be efficiently verified by checking if $\sigma_0$ is an internal simplex or not, and by checking if it is good or not (Claim~\ref{claim:good_bad}).}
    \item Verifier sets
    $\beta>1/poly(n)$ and 
    $L\in poly(n)$ s.t. 
   $P_{\sigma\rightarrow \sigma'}^{\sigma_0}\geq 0$ for any $\sigma,\sigma'$, and 
    $$
\sqrt{|X_d|} (1-\beta \epsilon)^L \leq 1/3.
$$
We can always choose such $\beta$ and $L$ due to Claim~\ref{claim:beta} and  $|X_d|\leq \binom{n}{d+1}$. 
    \item Verifier repeats following for $t=0,1,...,L$: 
    \begin{itemize}
        \item Check if all 
        all adjacent simplices to $\sigma_{t+1}$ are good simplices. 
        Output reject if there are any bad simplices. 
        \item Sample $\sigma_{t+1}$ according to the Markov transition matrix $P^{\sigma_0}$.
        \item 
        Compute 
        $$
        r_t=\frac{\Braket{[\sigma_{t+1}]|\phi_{\text{good}}^{\sigma_0}}}{\Braket{[\sigma_t]|\phi_{\text{good}}^{\sigma_0}}}.
        $$
    \end{itemize}
    \item Check if all 
        all adjacent simplices to $\sigma_{L}$ are good simplices. Otherwise, reject. 
    \item Verify that $\prod_{t=1}^L r_t \leq 1 $. Otherwise, reject. 
    \item Verifier accepts.
\end{itemize}

This protocol is similar to the one that is used in~\cite{bravyi2010complexity} applied for the fixed-node Laplacian. 

\paragraph{Completeness}

In YES instances, 
there is $\sigma_0'\in X_d^0$ with which 
$
\Delta_d\ket{\phi_{\text{good}}^{{\sigma_0'}}}=0.
$
An honest prover chooses a simplex $\sigma_0$ s.t. 
whose amplitude 
$$
\bra{[\sigma_0]}{\phi_{\text{good}}^{{\sigma_0'}}}\rangle
$$
is maximum among all the support of ${\phi_{\text{good}}^{{\sigma_0'}}}$. 
Starting from $\sigma_0$, the random walk performed by the verifier transits only on the support of $\ket{\phi_{\text{good}}^{{\sigma_0'}}} \in \ker(\Delta_d)$. 
By Claim~\ref{claim:good_supp}, 
$\{\sigma_t\}_{t=0}^L$ and their adjacent simplices are all good simplices. 
Moreover, as the prover sends $\sigma_0$ with maximal amplitude, 
$\prod_{t=1}^L r_t \leq 1 $. 
Therefore, the verifier always accepts.

\paragraph{Soundness}
We would like to show that the verifier rejects with high probability in NO instances for an arbitrary initial state $\sigma_0$. 
We first compute the probability $P_{\text{good}}$ that is a probability that in $L$ steps of a random walk, all $\sigma_t$ and simplices adjacent to $\sigma_t$ are good simplices. 
Let us denote $N(\sigma)$ for the set that is composed of $\sigma$ and adjacent simplices to $\sigma$.  

\begin{align*}
    P_{\mathrm{good}} &=
    \sum_{\substack{\sigma_1,...,\sigma_L: \\ N(\sigma_1),\dots, N(\sigma_L) \in S_\mathrm{good}^{\sigma_0}}} P_{\sigma_{L-1} \rightarrow \sigma_L} \cdots  P_{\sigma_0 \rightarrow \sigma_1}\\
    &= \sum_{\substack{\sigma_1,...,\sigma_L: \\ N(\sigma_1),\dots, N(\sigma_L) \in S_\mathrm{good}^{\sigma_0}}} 
    \prod_{t=0}^L  \left(\frac{\braket{[\sigma_{t+1}]|\phi_{\text{good}}^{\sigma_0}}}{\braket{[\sigma_t]|\phi_{\text{good}}^{\sigma_0}}}
 \braket{[\sigma_{t+1}] | I-\beta F_d^{\sigma_0}|[\sigma_t]} \right)
\end{align*}
The verifier accepts only if $N(\sigma_1),\dots, N(\sigma_L)$ are good and $\prod_{t=0}^L  \frac{\braket{[\sigma_{t+1}]|\phi_{\text{good}}^{\sigma_0}}}{\braket{[\sigma_t]|\phi_{\text{good}}^{\sigma_0}}} \leq 1$. 
Therefore, 
\begin{align*}
    P_{\text{acc}}&=P\left(N(\sigma_1),\dots, N(\sigma_L) \text{ are good} \ \cap  \prod_{t=0}^L  \frac{\braket{[\sigma_{t+1}]|\phi_{\text{good}}^{\sigma_0}}}{\braket{[\sigma_t]|\phi_{\text{good}}^{\sigma_0}}} \leq 1\right)\\
    &\leq 
    P\left(  \prod_{t=0}^L  \frac{\braket{[\sigma_{t+1}]|\phi_{\text{good}}^{\sigma_0}}}{\braket{[\sigma_t]|\phi_{\text{good}}^{\sigma_0}}} \leq 1 \ \Big|\ N(\sigma_1),\dots, N(\sigma_L) \text{ are good} \ \right)
\end{align*}
Therefore, the acceptance probability satisfies 
\begin{align*}
    P_{\text{acc}} &\leq 
    \sum_{\substack{\sigma_1,...,\sigma_L: \\ N(\sigma_1),\dots, N(\sigma_L) \in S_\mathrm{good}^{\sigma_0}}} 
    \prod_{t=0}^L  
 \braket{[\sigma_{t+1}] | I-\beta F_d^{\sigma_0}|[\sigma_t]} \\&
 \leq \sum_{\sigma_1,\dots, \sigma_L \in {X_d}}
    \prod_{t=0}^L  
 \braket{[\sigma_{t+1}] | I-\beta F_d^{\sigma_0}|[\sigma_t]}\\&
 = \sqrt{|X_d|} \left(\frac{1}{\sqrt{|X_d|}}\sum_{\sigma_L \in X_d}\bra{[\sigma_L]}\right)(I-\beta F_d^{\sigma_0})^L \ket{[\sigma_0]}
\end{align*}
where we have used $I=\sum_{\sigma \in {X_d} }\ket{[\sigma]} \bra{[\sigma]}$. 
Because by {Lemma~\ref{lemma:fixed_node_Hamiltonian}} and the promise in NO instances, the minimal eigenvalue of $F_d^{\sigma_0}$ is larger than $\epsilon$.
This allows the verifier to choose 
$L\in poly(n)$ s.t.
$
\sqrt{|X_d|} (1-\beta \epsilon)^L \leq 1/3.
$
With such a choice of $L$, it is ensured that 
$P_{\text{acc}}\leq 1/3$. 
This concludes the containment in MA.

\section{MA-hardness: Gadget construction}
\label{sec:hardness_1}

In this section, we give a construction of the simplicial complexes for the MA-hardness. 

\subsection{Qubit graph}
\label{sec:qubit_graph}

We first introduce a clique complex to which we can encode the $n$-qubit Hilbert space $\mathbb{C}^{2^n}$. 
We define an $n$-qubit graph $\mathcal{G}_n$ as the $n$-fold tensor product of a graph that is composed of two connected squares as follows: 

\begin{figure}[h!]
    \label{fig:qubit_graph}
    \centering
    \includegraphics[width=1\linewidth]{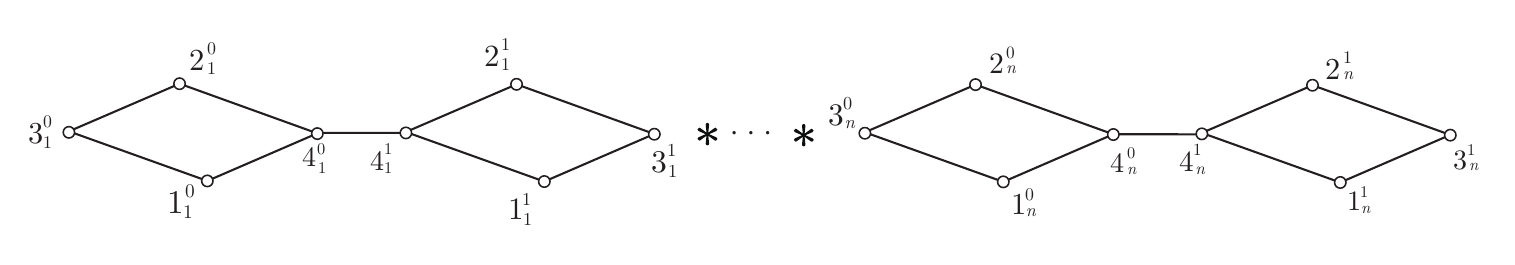}
    \caption{Qubit graph that we use in the MA-hardness construction.}
\end{figure}

Each of the graphs at the $i$-th position has $\dim H_1=2$, $\dim \tilde H_0=0$, where $\tilde H_0$ is the reduced 0-homology. 
Here the join product $G*G'$ of two graphs $G=(V,E)$ and $G'=(V',E')$ is composed of vertices 
$$
V \cup V'
$$
and edges 
$$
E \cup E' \cup \{(u,v):u\in V, v\in V'\}.
$$
The join of two simplicial complexes $K$ and $L$ is defined by 
$$K*L=\{\sigma\cup\tau: \sigma\in K, \tau \in L\}.$$
The clique complex of a join of graphs satisfies 
$$
\mathrm{Cl}(G*G')= \mathrm{Cl}(G)*\mathrm{Cl}(G'). 
$$
By the Kunneth formula 
\begin{equation}
\label{eq:kunneth}
   H_d(K*L)\cong \bigoplus_{i+j=d-1}H_i(K)\otimes H_j(L), 
\end{equation}
it can be seen that 
the clique complex of $\mathcal{G}_n$ has $2^n$-dimensional $2n-1$-th homology. 
Indeed,
$$
\ker(\Delta_{2n-1})= \underbrace{\ker(\Delta_1(\mathcal{G}))\otimes \cdots \otimes \ker(\Delta_1(\mathcal{G}))}_{n \text{ times}}.
$$
As we explicitly define later, the $n$-qubit Hilbert space is encoded into a harmonic subspace $\ker(\Delta_{2n-1})$. ($\ker(\Delta_{2n-1})$ is equivalent to the cycle subspace because there is no $2n$-dimensional simplices at this moment. )

Our base graph $\mathcal{G}$ is different from that used in \cite{king:qma} :
$
\basegraph
$
and also different from that used in~\cite{crichigno2024clique}. 
The reason that we use a different graph is that for $\mathcal{G}$, the {\it degree of the vertices will become 2} after removing the vertices between $4^0_i-4^1_i$, which is favorable for the requirement of orientability.  

\paragraph{Orientation of the clique complex of the qubit graph}

Let us define the ordering of the vertices of $\mathcal{G}_n$ as 
$$
1_1^0<2_1^0<3_1^0<4_1^0<1_1^1<2_1^1<3_1^1<4_1^1<1_2^0<2_2^0< \cdots .
$$
Then, we give the orientation of $2n-1$-dimensional simplices of $\mathrm
{Cl}(\mathcal{G}_n)$ as follows. 

\begin{definition}[Orientation of the qubit graph clique complex]
we orient $\mathrm{Cl}_{2n-1}(\mathcal{G}_n)$ 
s.t. for $v_1<v_2<...<v_{2n}$, 
$\sigma=(-1)^{\sum_{i=1}^{2n}v_i}[v_1v_2...v_{2n}] \in \mathrm{Cl}_{2n-1}(\mathcal{G}_n)$
if $v_1v_2...v_{2n}$ forms an $2n$-clique in $\mathcal{G}_n$. 
(In $\sum_{i=1}^{2n}v_i$, the subscripts and superscripts are ignored.)  
\end{definition}

Then, we define an encoding map from $n$-qubit space to the chain space, s.t. non-negative states in the qubit space appear as non-negative states in the chain space as follows. 

\begin{definition}[Encoding map]
\label{def:encoding_map}
    We define an encoding map 
$$\mathrm{Enc}:\mathcal{H}_n\rightarrow C_{2n-1}(\mathcal{G}_n)$$
where $\mathcal{H}_n$ is the $n$-qubit Hilbert space and 
$C_{2n-1}(\mathcal{G}_n):=C_{2n-1}(\mathrm{Cl}(\mathcal{G}_n))$
$$
\mathrm{Enc}(\ket{x})= \ket{c_{1,x_1}}\otimes \cdots \otimes \ket{c_{2,x_n}}
$$
for any $x\in\{0,1\}^n$
where 
$$
\ket{c_{i,0}}:= \frac{1}{2} (\ket{[1_i^03_i^0]}+\ket{[3_i^02_i^0]}+\ket{[2_i^04_i^0]}+\ket{[4_i^01_i^0]})
$$
$$
\ket{c_{i,0}}:= \frac{1}{2} (\ket{[1_i^13_i^1]}+\ket{[3^1_i2^1_i]}+\ket{[2^1_i4^1_i]}+\ket{[4^1_i1^1_i]}).
$$
\end{definition}

The tensor product, 
for $\sigma_1,\sigma_2,...,\sigma_n\in \mathrm{Cl}(\mathcal{G})$ when there is no overlap of vertices in $\sigma_1,\sigma_2,...,\sigma_n$ and $\sigma_1\sqcup\sigma_2\sqcup...\sqcup\sigma_n \in \mathrm{Cl}(\mathcal{G}_n)$,
is defined by
$$
\ket{[\sigma_1]}\otimes \cdots \otimes \ket{[\sigma_n]}= 
\ket{[\sigma_1\sqcup\sigma_2\sqcup...\sqcup\sigma_n]}.
$$

It can be seen that
$\{\mathrm{Enc}(\ket{x})\}_{x\in\{0,1\}^n}$ 
forms an orthonormal basis of $\ker(\Delta_{2n-1}(\mathcal{G}_n))$. 
It can also be seen that with our choice of orientation in $\mathrm{Cl}_{2n-1}(\mathcal{G}_n)$, 
$\mathrm{Enc}(\ket{x})$ appears as a {\it non-negative state} for any $x\in\{0,1\}^n$. Note that each of the computational basis states is encoded into a cycle of a generalized octahedron in eq.~\eqref{eq:n_octa} below. 

\subsection{Set of projectors}
\label{sec:MAgadget:projectors}

The set of projectors for the MA-hard QSAT problem is obtained by applying the circuit-to-Hamiltonian construction for a restricted verifier that corresponds to classical computation~\cite{bravyi2010complexity}. 
We explicitly introduce a set of rank-1 stoquastic projectors with that the QSAT problem is MA-hard. 

\begin{lemma}[A set of MA-hard rank-1 projectors]
\label{lemma:setofprojectors}
Let $S_{\mathrm{stoq}}$ be a set of stoquastic projectors onto 
\begin{itemize}
    \item Single-qubit state: $\ket{0}$
    \item Two-qubit state: $\ket{01}$
    \item  Three-qubit state: $\ket{001}, \ket{011}, \ket{-}\otimes\ket{01}$ 
    \item Five-qubit state: 
    $\frac{1}{\sqrt{2}}(\ket{000}-\ket{001})\otimes\ket{01}$, $\frac{1}{\sqrt{2}}(\ket{001}-\ket{101})\otimes\ket{01}$, $\frac{1}{\sqrt{2}}(\ket{001}-\ket{101})\otimes\ket{01}$
    \item Six-qubit state: 
    $\frac{1}{\sqrt{2}}(\ket{0111}-\ket{1110})\otimes\ket{01}$, $\frac{1}{\sqrt{2}}(\ket{0011}-\ket{1111})\otimes\ket{01}$
\end{itemize}
Then, the satisfiability problem with projectors chosen from $S_{\mathrm{stoq}}$ is MA-hard. 
\end{lemma}

\begin{proof}

Let $U=U_L \cdots U_1$ be a quantum circuit with 3-qubit Toffoli gates with initial states $\ket{\phi}$ that are restricted to product states with $\ket{0}, \ket{1}, \ket{+}$. 

The system is composed of 
$N_w$ witness qubits and $N_a$ ancilla qubits. The $k$-th ancilla qubit is labeled as $a(k)$ with $k=1,...,N_a$. 
Then, the constraints of the QSAT are composed of 
$$
\mathcal{C}=
\{
H_k^{\mathrm{int}}, H_j^{\mathrm{prop}},H_l^{\mathrm{clock}}, H^{\mathrm{meas}}
\}
$$
where 
\begin{align}
    H_k^{\mathrm{int}} &= (I-\ket{\phi_k}\bra{\phi_k})_{a(k)}\otimes \ket{10}\bra{10}_{\mathrm{Cl}(0),\mathrm{Cl}(1)},\  k= 1,...,N_a \\
    H_j^{\mathrm{prop}} &= \frac{1}{2}\ket{1}\bra{1}_{\mathrm{Cl}(j-1)}\otimes 
    \left(
    \ket{1}\bra{1}_{\mathrm{Cl}(j)}+\ket{0}\bra{0}_{\mathrm{Cl}(j)}-\ket{1}\bra{0}_{\mathrm{Cl}(j)} \otimes U_j  - \ket{0}\bra{1}_{\mathrm{Cl}(j)}\otimes U^\dagger_j
    \right) \\ & \ \ \  \otimes \ket{0}\bra{0}_{\mathrm{Cl}(j+1)}, \ j=1,...,L \\ 
   H_0^{\mathrm{clock}} &= \ket{0}\bra{0}_{\mathrm{Cl}(0)} \\
   H_l^{\mathrm{clock}} &= \ket{01}\bra{01}_{\mathrm{Cl}(l-1),\mathrm{Cl}(l)}, \ l=1,...,L \\
   H^{\mathrm{meas}} &= \ket{0}\bra{0}_{\text{out}}\otimes \ket{1}\bra{1}_{\mathrm{Cl}(L)}.
\end{align}
With our choice of gate set, the set of projectors is 
\begin{itemize}
    \item Single-qubit:  $\ket{0}\bra{0}$
    \item Two-qubit: $\ket{01}\bra{01}$
    \item Three-qubit: 
    $\ket{010}\bra{010}, \ket{110}\bra{110}, \ket{-}\bra{-}\otimes \ket{10}\bra{10}$
    \item Six-qubit: 
    $\frac{1}{2}\ket{10}\bra{10}\otimes \Big(\ket{0}\bra{0}+\ket{1}\bra{1}-(\ket{1}\bra{0}+\ket{0}\bra{1})\otimes U_{\mathrm{tof}}\Big)$
\end{itemize}
where $U_{\mathrm{tof}}$ is the Toffoli gate. 

Let 
$$2 \Pi_{\mathrm{tof}}= \ket{0}\bra{0}+\ket{1}\bra{1}-(\ket{1}\bra{0}+\ket{0}\bra{1})\otimes U_{\mathrm{tof}}.$$
Observe that
\begin{align*}
    \Pi_{\mathrm{tof}}= 
    I_4 - &(\ket{1}\bra{0}+\ket{0}\bra{1})\otimes (\ket{00}\bra{00}+\ket{01}\bra{01}+\ket{10}\bra{10})\otimes I \\
    &- (\ket{1}\bra{0}+\ket{0}\bra{1})\otimes (\ket{110}\bra{111}+\ket{111}\bra{110}).
\end{align*}
Therefore, the rank-8 projector $\Pi_{\mathrm{tof}}$ can be decomposed into the projector onto the states 
\begin{itemize}
    \item $\ket{1110}-\ket{0111}$
    \item $\ket{1111}-\ket{0110}$
    \item $(\ket{100}-\ket{000})\otimes I$
    \item $(\ket{101}-\ket{001})\otimes I$
    \item $(\ket{110}-\ket{010})\otimes I$
\end{itemize}
Then, we can conclude that the set of projectors in the statement of the lemma suffices for MA-hardness. 
\end{proof}

In the following subsections, we give our gadget construction for projectors onto $\ket{x}-\ket{y}$ in a general way. 
The rough overview of the construction is:
\begin{enumerate}
    \item As $\mathrm{Enc}(\ket{x})$ are cycles of a generalized octahedron $K^x$, we prepare a copy $K^{'x}$. Similarly, prepare a copy for $K^y$, which we denote $K^{'y}$. 
    \item We glue $K^x$ and $K^{'x}$, and also glue $K^y$ and $K^{'y}$. 
    \item Introduce another generalized octahedron $K''$. Then, glue $K''$ to $K^{'x}$ and $K^{'x}$. 
    \item Add axial simplices.
\end{enumerate}
With this construction, the cycles $\mathrm{Enc}(\ket{x})$ and $\mathrm{Enc}(\ket{y})$ are effectively connected and $\mathrm{Enc}(\ket{x}+\ket{y})$ becomes boundary. 
As can be seen from the above overview, the central ingredient in the construction is a procedure of gluing two generalized octahedra. Therefore, we first introduce this procedure in the next subsection.  

\subsection{Procedure of gluing two generalized octahedra}
\label{sec:General_argument}

\paragraph{Construction of the gadget that glues two generalized octahedra}

To construct gadgets for the projectors onto $\ket{x}-\ket{y}$, we first establish a procedure of gluing two generalized octahedra. 

Let 
\begin{equation}
\label{eq:n_octa}
   K:= \{1,2\} * \{3,4\} \cdots * \{2n-1,2n\}
\end{equation}
be a generalized octahedron with $d-1$-dimensional maximal simplices. 
($\{1,2\} * \{3,4\} * \{5,6\} $ is a usual three dimensional octahedron.)
The set of $n-1$ dimensional simplices of $K$ is given by
\begin{equation}
\label{eq:n_octa_order}
   \{1,2\}\times \{3,4\}\times \cdots \times \{2n-1,2n\}. 
\end{equation}

In the following, we treat $K$ as an oriented simplicial complex. 
The orientation of $K_{n-1}$ is defined such that 
$$(-1)^{\sum_{i=1}^n v_i}[v_1v_2...v_n]\in K_{n-1}$$
if $v_1<v_2<...<v_n$ and $v_i \in \{2i-1,2i\}$. 

For example, in the case of $n=2$, 
$$
[1,3], -[1,4]=[4,1],-[2,3]=[3,2],[2,4]
$$
are the elements of $K_1$.

With this choice of orientation, it can be seen that 
$$
\partial_{n-1} \left( \sum_{\sigma\in K_{n-1}} \ket{\sigma} \right) = 0
$$

In \cite{king:qma}, a general procedure of ``thickening'' of simplicial complexes is introduced. 
This can be utilized for the purpose of gluing two generalized octahedra.

\begin{definition}[Thickening~\cite{king:qma}]
\label{def:thickening}
    Let $K$ be a simplicial complex. Order vertices $K_0$. 
    Let $L$ be the simplicial complex with vertices 
    $L_0=K_0 \times \{0,1\}$ and simplices 
    $$[(u_1,0)(u_2,0)...(u_a,0)]$$
    $$[(u_1,1)(u_2,1)...(u_a,1)]$$
    whenever $[u_1u_2...u_a]\in K$,
    and 
    $$[(u_1,0)(u_2,0)...(u_a,0)(v_1,1)...(v_b,1)]$$
    whenever
    \begin{itemize}
        \item $u_1<\cdots <u_a\leq v_1<\cdots <v_b$
        \item $[u_1...u_a]\in K$
        \item $[v_1...v_b]\in K$
        \item if $u_a=v_1$ then $[u_1...u_av_2...v_b]\in K$
        \item if $u_a\neq v_1$ then $[u_1...u_av_1...v_b]\in K$.
    \end{itemize}
    
\end{definition}

\begin{lemma}[\cite{king:qma}]
\label{lemma:thickening_graph}
Let $K$ be a generalized octahedron of \eqref{eq:n_octa} with the ordering of vertices in \eqref{eq:n_octa_order}. Then, the thickening $L$ is a triangulation of $K\times I$. 
Moreover, $L$ is a
clique complex of a graph 
with vertices $L_0$ and edges 
\begin{align*}
    &\{((u,0),(v,0)):(u,v)\in K_1\}\\
    \sqcup&\{((u,1),(v,1)):(u,v)\in K_1\}\\
    \sqcup&\{((u,0),(u,1)):u\in K_0\}\\
    \sqcup&\{((u,0),(v,1)):(u,v)\in K_1\}, u<v\}.
\end{align*}
. 
\end{lemma}

\begin{remark}
    We emphasize that the ordering of the vertices of $K$ is important in the thickening. In particular, the above lemma works well with our specific choice of ordering of vertices. 
\end{remark}

In this section, 
we use $i,i'$ to denote $(i,0),(i,1)$ for simplicity. 
We define for $k=1,2,...,n$,
\begin{equation}
\label{eq:filtration_thickenig}
  L_n^k := \{1,2\}\times \cdots \times \{2k-3,2k-2\} \times \{(2k-1)(2k-1'),(2k)(2k')\} \times \{2k+1',2k+2'\} \times \cdots \times \{2n-1',2n'\}.  
\end{equation}
Note that 
$$
L_n^1 = \{1,2\}\times \{3,4\}\times \cdots \times \{(2n-1)(2n-1'),(2n)(2n')\}
$$
$$
L_n^n = \{11',22'\}\times \{3',4'\}\times \cdots \times \{2n-1',2n'\}.
$$
Namely, $L_n^k$ is a subset of $L_n$ with $k-1$-number of vertices indexed with a prime. 
It holds that 
\begin{equation}
\label{eq:decomposition_thickening}
   L_n= L_n^1 \sqcup ...\sqcup L_n^n.  
\end{equation}
Let $L_{n-1}^k$ be the set of simplices that is obtained by removing any one vertex $(v,0)$ from $L_{n}^k$ for $k=1,2,...,n$. Let $L_{n-1}^0:= X_{n-1}\times \{0\}$ be the $n-1$ simplices of the initial generalized octahedron. 

We similarly define for $k=1,2,...,n$,
\begin{equation*}
    \hat{L}_{n-1}^k := \{1,2\}\times \cdots \times \{2k-3,2k-2\} \times \{2k-1',2k'\} \times \{2k+1',2k+2'\} \times \cdots \times \{2n-1',2n'\}
\end{equation*}
which can be obtained by removing vertices with labels $2k-1$ or $2k$.

\paragraph{Filtration of unweighted simplicial complex $L$}
In this subsection, we consider $L$ to be a unweighted simplicial complex. 

\begin{lemma}
\label{lemma:orientable_filtration_octa}
{There is a uniform orientable filtration for $L_n$.}
\end{lemma}
\begin{proof}
Let 
$$
X_{n-1}^k=L_{n-1}^1\sqcup \cdots \sqcup L_{n-1}^k.
$$
Then, 
$$X_d^0 \subseteq X_d^1 \subseteq \cdots \subseteq X_d^n$$
is a filtration of $X_d$. 
The orientation of ${L}_{n-1}^0$ is defined such that it is equivalent to the orientation of $K_{n-1}$ as defined in Section~\ref{sec:qubit_graph}. 
We can choose the orientation of $\hat{X}_{n-1}^k=\hat{L}_{n-1}^k$ for $k=1,...,n$ such that it induces the opposite orientations on the common cofaces of $\hat{L}_{n-1}^{k-1}$. 
We take the orientation of $L_n^k$ such that simplices in $L_n^k$ have the same orientation as the orientations induced by simplices in $\hat{L}_{n-1}^{k-1}$.
Note that with such orientation, $\sigma\in \hat{X}_{n-1}^{k-1}$ and $ \sigma' \in \hat{X}_{n-1}^{k}$ induce opposite orientations on their common cofaces. 

The orientation of simplices other than 
$$\hat{L}_{n-1}^1\sqcup \cdots \sqcup \hat{L}_{n-1}^k$$
can be taken arbitrarily. 
With this orientation, 
$$
X_{n-1}^k=L_{n-1}^1\sqcup \cdots \sqcup L_{n-1}^k.
$$
is an orientable filtration. 
Moreover, the uniformity also holds by construction. 
\end{proof}

\paragraph{Characterizing the harmonics of the glued generalized octahedra}

Let
$$\hat{L}_{n-1}^0:= \{1,2\}\times \{3,4\}\times \cdots \times \{2n-1,2n\}.$$
Let also
$$\ket{c_{n-1}}:= \sum_{\sigma\in K_{n-1}}\ket{\sigma}$$
$$\ket{c_{n-1}'}:= \sum_{\sigma\in K_{n-1}'}\ket{\sigma}.$$

First, we show that $\ket{c_{n-1}}-\ket{c_{n-1}'}$ is a boundary. 

\begin{lemma}
\label{lemma:boundary}
    $$\partial_n \left(
    \sum_{\tau\in L_n} \ket{\tau}
    \right) = \ket{c_{n-1}}-\ket{c_{n-1}'}.$$
\end{lemma}

\begin{proof}
Note that 
    $$L_n = L_n^1 \sqcup \cdots L_n^n$$
    and $L_n$ is down-degree 2. 
    By definition, the simplices in $L_n^k$ are adjacent to simplices only in $L_n^{k-1}$ or $L_n^{k+1}$. 
    When $\tau \in L_n^{k-1}$ and $\tau' \in L_n^{k}$ have $\sigma$ as a common face, $\sigma$ is contained only in ${\tau}$ and $\tau'$. 
    Moreover, the orientation of $\tau$ and $\tau'$ is taken such that they induce the opposite orientations on $\sigma$. 
    The $D-1$-simplices in $\hat{L}_{n-1}^0$ and $\hat{L}_{n-1}^n$ are the only simplices in $L_{n-1}$ that are contained only by one simplices in $L_n$. 
    Therefore, the claim follows.
\end{proof}

We can characterize the harmonics of $L$ as follows.

\begin{lemma}
\label{lemma:harmonics_octa}
Let
\begin{equation}
\label{eq:harmonics_octa}
    \ket{\phi_{n-1}}= \sum_{\sigma_0 \in \hat{L}_{n-1}^0} \ket{\sigma_0} + \sum_{\sigma_1 \in \hat{L}_{n-1}^1} \ket{\sigma_1} + \cdots + \sum_{\sigma_n \in\hat{L}_{n-1}^n} \ket{\sigma_n}.
\end{equation}
    Then $\mathrm{Span}(\ket{\phi_{n-1}})= \ker(\Delta_{n-1}(L))$.
\end{lemma}

\begin{proof}
First, there are two orthogonal harmonics in the two copies of $K$ that are $\ket{c_{n-1}}$ and $\ket{c_{n-1}'}$ before gluing these two cycles. 
As we have seen in Lemma~\ref{lemma:boundary}, $\ket{c_{n-1}}-\ket{c_{n-1}'}$ is a boundary in $L$. Therefore, $\dim({\Delta_{n-1}(L)})\leq 1$ because no additional holes are introduced in the gluing.

Note that 
$$
\partial_{n-1}\left(\sum_{\sigma_k \in \hat{L}_{n-1}^k} \ket{\sigma_k}\right) =0
$$
for each $k=0,1,...,n$ by construction.
This means $\partial_{n-1}\ket{\phi_{n-1}}=0$.
Therefore, in order to verify that $\Delta_{n-1}\ket{\phi_{n-1}}=0$, it is enough to see that $\delta_{n-1}\ket{\phi_{n-1}}=0$. 

It can be seen that 
$\hat{L}_{n-1}^k$ is contained only in $L_n^k\sqcup L_n^{k+1}$ as faces for any $k=0,1,...,n$ where $L_n^0=L_n^{n+1}=\{\emptyset\}$. 
Moreover, the intersection of the set of cofaces of $\hat{L}_{n-1}^k$ and $\hat{L}_{n-1}^{k+1}$ is 
precisely $L_n^{k+1}$. 
Therefore, it is enough to show that 
\begin{equation}
\label{eq:vanish}
\Pi_{k}\cdot  \delta_{n-1}
\left(
\sum_{\sigma_{l} \in \hat{L}_{n-1}^{k-1}} \ket{\sigma_{k-1}}  + \sum_{\sigma_n \in \hat{L}_{n-1}^{k}} \ket{\sigma_{k}}\right)=0
\end{equation}
where $\Pi_{k+1}$ is the projector onto $\mathrm{Span}\left(\{\ket{\tau}\}_{\tau\in L_n^k}\right)$. 
Recall that
$$
L_n^k = \{1,2\}\times \cdots \times \{2k-3,2k-2\} \times \{(2k-1)(2k-1'),(2k)(2k')\} \times \{2k+1',2k+2'\} \times \cdots \times \{2n-1',2n'\}
$$
$$
\hat{L}_{n-1}^k = \{1,2\}\times \cdots \times \{2k-3,2k-2\} \times \{2k-1',2k'\} \times \{2k+1',2k+2'\} \times \cdots \times \{2n-1',2n'\}. 
$$
Therefore, every simplex $\tau \in L_n^k$ has precisely one simplex $\sigma_{k-1} \in \hat{L}_{n-1}^{k-1}$ as a face and also one simplex $\sigma_{k} \in \hat{L}_{n-1}^{k}$ as a face. The orientations are taken such that $\sigma_{k-1},\sigma_k$ induces opposite orientations on $\tau$. 
This implies \eqref{eq:vanish} holds.

\end{proof}

\subsection{Gluing Gadget for projectors onto $\ket{x}-\ket{y}$}
\label{sec:MAgadget:gluing}

\paragraph{Gluing $\ket{x},\ket{y}$ for  $x,y \in \{0,1\}^m$}

Let $h$ be a projector onto 
$$\ket{x}-\ket{y}$$
$x,y \in \{0,1\}^{m}$. 
Our strategy for implementing the projector for $h$ is to connect two (disjoint) $\mathrm{Enc}(\ket{x})$ and $\mathrm{Enc}(\ket{y})$ cycles by a ``wormhole''. 

First, we fix the labeling of the vertices relevant to $\ket{x}$ and $\ket{y}$. 
For $\mathrm{Enc}(\ket{x})$, we label the  vertices such that the maximal faces of the corresponding generalized octahedron are given by 
$$
K^x=
\{1^{x_1}_1,2^{x_1}_1\} * \{3^{x_1}_1,4^{x_1}_1\} * \cdots *
\{1^{x_m}_m,2^{x_m}_m\} * \{3^{x_m}_m,4^{x_m}_m\}. 
$$
Similarly, we label the vertices of $\mathrm{Enc}(\ket{y})$ with 
$$
K^y=
\{1^{y_1}_1,2^{y_1}_1\} * \{3^{y_1}_1,4^{y_1}_1\} * \cdots * 
\{1^{y_m}_m,2^{y_m}_m\} * \{3^{y_m}_m,4^{y_m}_m\}. 
$$
It is important to 
note that when $x_i=y_i$, $1^x_i=1^y_i$, $2^x_i=2^y_i$, $3^x_i=3^y_i$ and $4^x_i=4^y_i$, i.e., there may be some overlaps of vertices between $K^x$ and $K^y$. 

We introduce three intermediate generalized octahedra $K'^x, K'^y, K''$ composed as
$$
K'^x=
\{1'^x_1,2'^x_1\} * \{3'^x_1,4'^x_1\} * \cdots *
\{1'^x_m,2'^x_m\} * \{3'^x_m,4'^x_m\}
$$
$$
K'^y=
\{1'^y_1,2'^y_1\} * \{3'^y_1,4'_1\} * \cdots  * 
\{1'^y_m,2'^y_m\} * \{3'^y_m,4'_m\}. 
$$
$$
K''=
\{1''_1,2''_1\} * \{3''_1,4''_1\} * \cdots  * 
\{1''_m,2''_m\} * \{3''_m,4''_m\}. 
$$
Note that the vertices of $K'^x$ and $K'^y$ are completely disjoint, even if some overlaps of vertices between $K^x$ and $K^y$.  
The introduced generalized octahedra are glued as 
\begin{align}
\label{eq:glue4}
   & K^x=\left(\bigotimes_{i=1}^m \xsinglegraph\right)\ {\longrightarrow}\ \
K'^x=\left(\bigotimes_{i=1}^m \xxsinglegraph\right)  \longleftarrow K''=\left(\bigotimes_{i=1}^m \primesinglegraph\right) \\
& \longrightarrow K'^y=\left(\bigotimes_{i=1}^m \yysinglegraph\right) \longleftarrow 
K^y=\left(\bigotimes_{i=1}^m \ysinglegraph\right). \nonumber
\end{align}
The arrow indicates the ordering of vertices in the gluing procedure:
$K^x$ is glued {\it to} $K'^x$  w.r.t the ordering of vertices 
$$
1^{x_i}_i<2^{x_i}_i<3^{x_i}_i<4^{x_i}_i<1'^{x}_i<2'^{x}_i<3'^{x}_i<4'^{x}_i.
$$
We may write this relation as $K^x<K'^x$. Then, other gluing is performed w.r.t the ordering $K''<K'^x$, $K''<K'^y$, and $K^y<K'^y$. 

We need to add more simplices in order to prevent us from generating new holes in the gadget construction, see Figure~\ref{fig:newhole}. 
Figure~\ref{fig:newhole} shows the graph of the $i$-th component that appears after the above gluing procedure for the case $x_i=y_i$ and $x_i\neq y_i$.  
As can be seen from the figure, in the case $x_i=y_i$, there is a new hole $4_i^{x_i}4_i^{'x}+4_i^{'x}4''_i+4''_i4_i^{'y}+4_i^{'y}4_i^{y_i}+4_i^{y_i}4_i^{x_i}$! 
The edge $4_i^{x_i}4_i^{y_i}$ was introduced so that the qubit graph does not consist of many connected components. However, this leads to the production of a new hole in the gadget's construction. 
In the case $x_i= y_i$, there also appears a new hole $4_i^{x_i}4_i^{'x}+4_i^{'x}4''_i+4''_i4_i^{'y}+4_i^{'y}4_i^{x_i}$! The condition $x_i= y_i$ leads to making a torus. 
These new holes in the $i$-th component will lead to new $2m-1$-dimensional holes in the whole gadget complex as well. 

\begin{figure}
    \centering
    \includegraphics[width=1\linewidth]{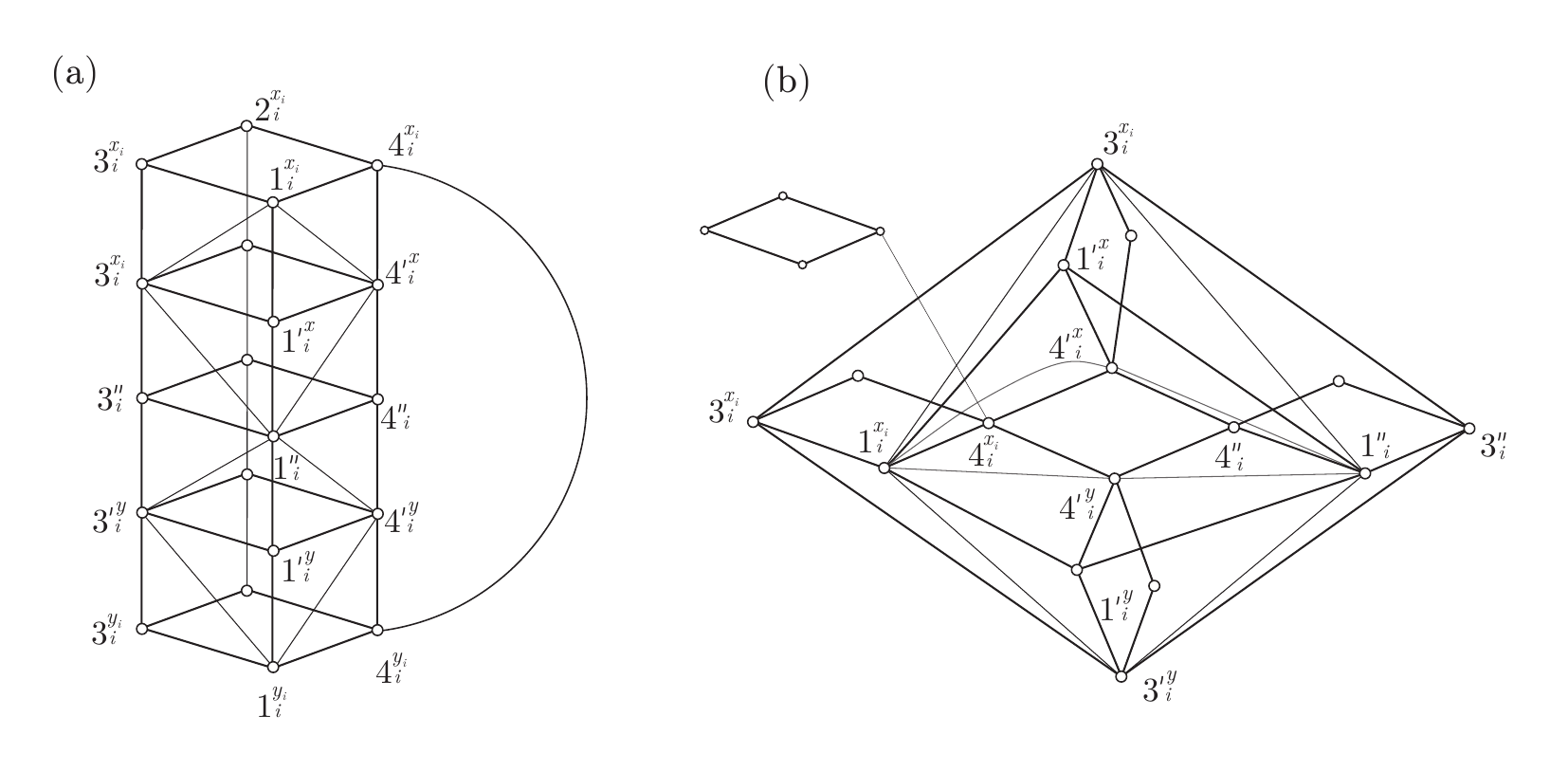}
    \caption{Graph of the $i$-th component after gluing generalized octahedra. (a) The case $x_i=y_i$. (b) The case $x_i\neq y_i$.  In both (a) and (b), not all the edges are described for maintaining visibility. }
    \label{fig:newhole}
\end{figure}

In order to prevent the gadget from creating new holes, we add more simplices. 
We modify the construction so that the $i$-th component in the gadget graph looks like Figure~\ref{fig:gadgetgraph_modified}. 
It can be seen that the newly generated holes in Figure~\ref{fig:newhole} now appear as boundaries. 

\begin{remark}
    Our construction introduces three intermediate generalized tetrahedra rather than two intermediate ones. 
    There are two reasons for this. The first reason is that we can use a symmetric ordering of vertices for $x$ and $y$. 
    Second, if we choose two intermediate ones, there are unwanted triangles that appear automatically. For example, $4_i^{x_i}4_i^{'x}4_i^{'y}$ becomes triangle in Figure \ref{fig:newhole} (b) if we construct without $4''_i$. 
    Such triangles read to new 2-dimensional homology classes in Figure \ref{fig:newhole} (b). In order to avoid such confusion, we have chosen a construction with three intermediate copies of the cycles. 
\end{remark}

\begin{figure}
    \centering
    \includegraphics[width=1\linewidth]{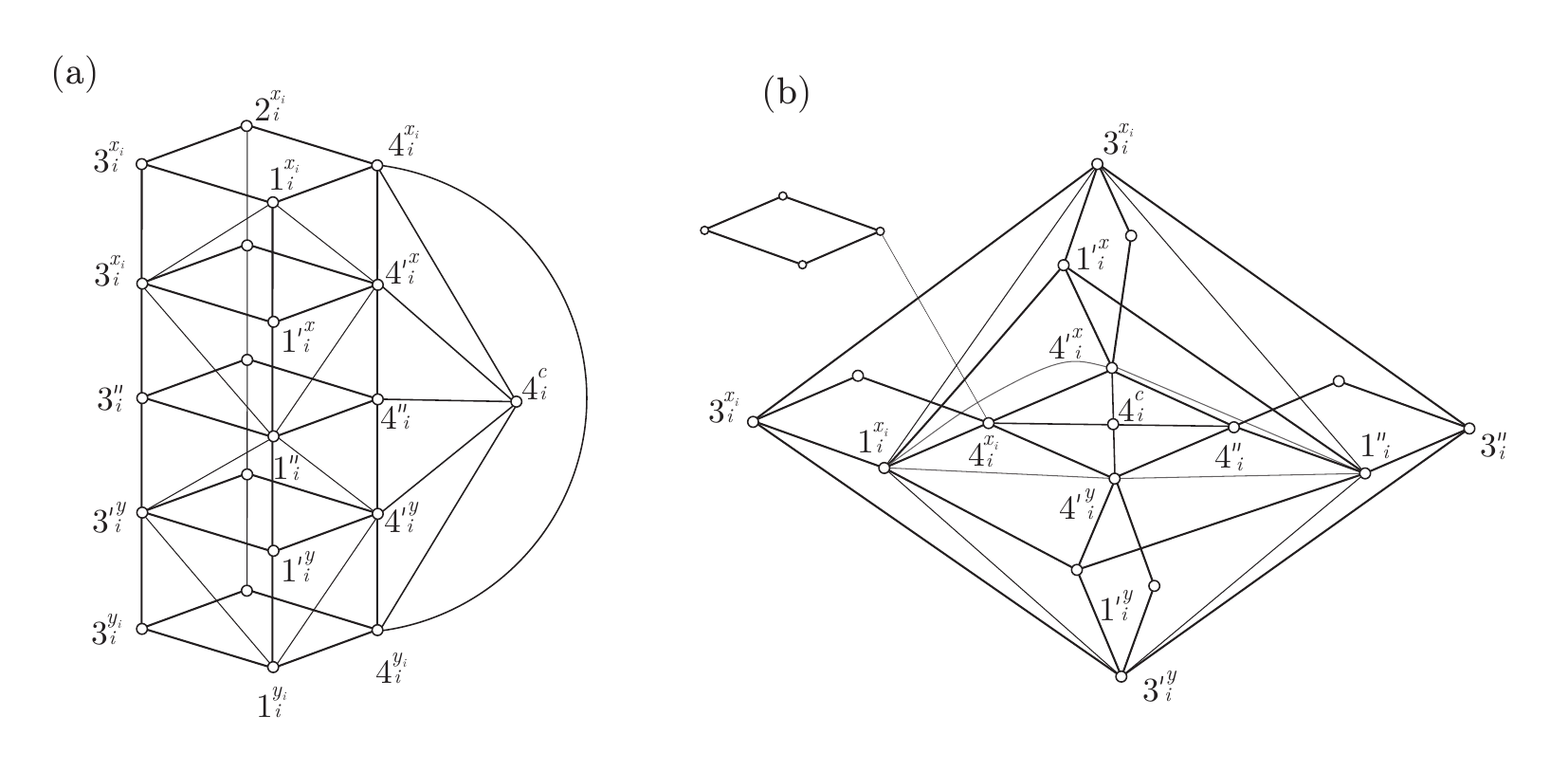}
    \caption{We add central vertex $4_i^c$ for each clock $i$ in the final gadget construction.}
    \label{fig:gadgetgraph_modified}
\end{figure}

\paragraph{Description of the gadget graph and simplicial complex}

Now, we describe the gadget simplicial complex. 
Recall the simplicial complex of eq.~\eqref{eq:glue4}
$$
K^x \rightarrow K'^x \leftarrow K'' \rightarrow K'^y \leftarrow K^y.
$$
According to Lemma~\ref{lemma:thickening_graph}, there is a corresponding graph $\tilde{\mathcal{G}}_m$
 for this clique complex. 
For each $i=1,...,m$, we add a vertex $4_i^c$. 
Each of $4_i^c$ is connected to 
$$4_i^{x_i},4_i^{'x},4_i^{''},4_i^{'x},4_i^{y_i}$$ inside the $i$-th block (it may be that $4_i^{x_i}=4_i^{y_i}$). 
Moreover, each of $4_i^c$ is connected to every vertices outside of the $i$-th block that at least one of  $4_i^{x_i},4_i^{'x},4_i^{''},4_i^{'x},4_i^{y_i}$ is connected to. 
Let us denote the graph constructed in this way as $\hat{\mathcal{G}}_m$. 

\paragraph{Weighting}
All the vertices that are not the vertices of the original qubit graph are weighted by $\lambda\ll 1$.

\subsection{Gadgets for filling computational basis states}
\label{sec:filling_gadgets}

For projectors onto any computational basis state $\ket{x}$ with $x\in\{0,1\}^m$, we use the same construction as that of \cite{king:qma}. 
Consider a projector $\ket{x}\bra{x}$. 
The state $\ket{x}$ is encoded into a cycle of a generalized octahedron $K^x$.  We glue $K^x$ to a copy of the generalized octahedron $K'^x$. Then we put a center vertex $v^c$ and fill the hole of $K'^x$. 
The vertices of $K^x$ are weighted by $1$, and the other vertices are weighted by $\lambda$. 

It is easy to see that there is an orientable filtration for this gadget as well. This is because we can similarly filtrate the gadget between $K^x$ and $K'^x$, and the orientation of the simplices that touch the central vertex can be taken arbitrarily (recall the center part of Figure~\ref{fig:orientable_filtration}).  

The following is shown in \cite{king:qma}.  This is a lemma 
for $\mathrm{Cl}(\hat{\mathcal{G}}_m)$, where $\hat{\mathcal{G}}_m$ is the qubit graph of the target $m$ qubits combined with the gadget graph for $\ket{x}\bra{x}$. 

\begin{lemma}[\cite{king:qma}]
\label{lemma:spectrum_filling}
    Let $\hat{\Delta}_{2m-1}$ be a Laplacian for $\mathrm{Cl}_{2m-1}(\hat{\mathcal{G}}_m)$
    where $\hat{\mathcal{G}}_m$ is the gadget graph for the projector onto a computational basis state $\ket{x}$ for $x\in\{0,1\}^m$. It holds that 
    \begin{itemize}
        \item $\hat{\Delta}_{2m-1}$ has a $(2^m-1)$-dimensional kernel, which is a $\mathcal{O}(\lambda)$-perturbation of the subspace 
        \item The first excited state of $\hat{\Delta}_{2m-1}$ is a $\mathcal{O}(\lambda)$-perturbation of $\mathrm{Enc}(\ket{x})$ and it has energy $\Theta(\lambda^{4m+2})$.
        \item The next lowest eigenvectors have eigenvalues $\Theta(\lambda^2)$, and they are $\mathrm{O}(\lambda)$-perturbation of sums of $(2m-1)$-simplices touching the central vertex $v^c$.
        \item The rest of the eigenvalues are $\Theta(1)$.
    \end{itemize}
\end{lemma}

\section{Spectral sequence and the spectrum of the gadget complex}
\label{sec:hardness_2}

We investigate the spectral property of the combinatorial Laplacian $\hat{\Delta}_d:C_d(\mathrm{Cl}(\hat{\mathcal{G}}_m))\rightarrow C_d(\mathrm{Cl}(\hat{\mathcal{G}}_m))$. 

There is a natural filtration of the chain space with the weight of simplices\footnote{This filtration is different from the filtration for uniform orientable filtration.}. 
Let 
$$C_d^k:= \mathrm{Span}(\{\ket{\sigma}: \sigma \in X_d\ \ \text{s.t.\ } w(\sigma)\in \{\lambda^k, \lambda^{k+1},...,\lambda^{d+1}\}  \}).$$ 
Then, it holds that
$$
C_d=C_d^0 \supseteq C_d^1 \supseteq \cdots  \supseteq C_d^{d+1}
$$
and 
$$
\delta_d(C_d^k) \subseteq C_{d+1}^k.
$$
Therefore, we can apply the analysis with spectral sequences. Lemma~\ref{lemma:perturbation} allows us to understand the spectral property of $\hat{\Delta}_d$ with the analysis of spectral sequences. 

The vector spaces for the $k$-th page $e_d^{k,l}$ are defined as follows. 
The 0th page is given by 
$$
e_d^{0,l}= C_d^l / C_d^{l+1} 
$$
and therefore, $$e^{0,l}_d \cong \mathrm{Span}(\ket{\sigma}:\sigma \in \mathrm{Cl}_d^l(\hat{\mathcal{G}}_m)),$$
where $\mathrm{Cl}_d^l(\hat{\mathcal{G}}_m))$ is the set of weight $\lambda^l$ simplices in $\mathrm{Cl}_d(\hat{\mathcal{G}}_m))$. 
There is a map induced by $\delta_d$:
$$
\delta_d^{0,l}:e_d^{0,l} \rightarrow e_{d+1}^{0,l}.  
$$
Then, the 1st page is given by 
$$
e_d^{1,l}= \ker(\delta_d^{0,l}) / \mathrm{Im}(\delta_{d-1}^{0,l}).
$$
Then, the coboundary map $\delta_d$ induces 
$$
\delta_d^{1,l}:e_d^{1,l} \rightarrow e_{d+1}^{1,l+1}. 
$$

The general $k$-th page can be introduced as follows. Suppose we have an induced coboundary map 
$$
\delta_d^{k,l}:e_d^{k,l} \rightarrow e_{d+1}^{k,l+k}.
$$
Then the $k+1$-th page can be defined as 
$$
e_d^{k+1,l}= \ker(\delta_d^{k,l}) / \mathrm{Im}(\delta_{d-1}^{k,l-k}).
$$
We can also consider an induced boundary map 
$$
\partial_d^{k,l}: e_d^{k,l}\rightarrow e_d^{k,l-k}.
$$
The $k+1$-th page can also be characterized by this induced boundary map as
$$
e_d^{k+1,l}= \ker(\partial_d^{k,l})/\mathrm{Im}(\partial_{d+1}^{k,l+k}). 
$$

Let us introduce a notation about the perturbation of subspaces.
\begin{definition}[Perturbation of subspaces \cite{king:qma}]
Consider a subspace $\mathcal{U}\subseteq \mathcal{V}$ of a complex vector space $\mathcal{V}$. 
Let $\mathcal{U}_\lambda\subseteq \mathcal{V}$ be a family of subspaces parameterized by a continuous parameter $\lambda \in[0,1]$. 
$\mathcal{U}_\lambda$ is said to be a $\mathcal{O}(\lambda)$-perturbation of 
$\mathcal{U}$ if 
there exists orthonormal basis $\{\ket{u}\}_u$ for $\mathcal{U}$ and $\{\ket{u,\lambda}\}_u$ for each $\mathcal{U}_\lambda$ s.t. for all $\ket{u}$,
$$
\|\ket{u,\lambda}-\ket{u}\|=\mathcal{O}(\lambda). 
$$
\end{definition}

By computing the spectral sequences, we can know about the spectral property of the combinatorial Laplacian following \cite{forman1994hodge, king:qma}.

\begin{lemma}[\cite{forman1994hodge, king:qma}]
\label{lemma:perturbation}
    The subspace
    $$\mathrm{Span}(\{\ket{\psi}:\ket{\psi} \text{ is an eigenvector of } \Delta_d \text{ with eigenvalue } \mathcal{O}(\lambda^k)\})$$
    $\mathcal{O}(\lambda)$-perturbation of $E_d^{k}$.
    Here, 
    $$
    E_d^{k}=\bigoplus_l E_d^{k,l}
    $$
    and $E_d^{k,l}$ is a space s.t. there is an isomorphism from $E_d^{k,l}$ to $e_d^{k,l}$. 
\end{lemma}
The corresponding map from $e_d^{k,l}$ to $E_d^{k,l}$ is obtained by projecting a representative onto $$\mathrm{Span}(\ket{\sigma}:\sigma \in \mathrm{Cl}_d^l(\hat{\mathcal{G}}_m)).$$
We show the spectrum of the gluing gadget as follows.

\begin{lemma}
\label{lemma:singlegadget_spectrum}
Let $h_i$ be a projector onto $\ket{x}-\ket{y}$. Then, for the Laplacian $\hat{\Delta}_{2m-1}(h_i):C_d(\mathrm{Cl}(\hat{\mathcal{G}}_{2m-1}))\rightarrow C_d(\mathrm{Cl}(\hat{\mathcal{G}}_{2m-1}))$, the following holds. 
\begin{itemize}
    \item $\hat{\Delta}_{2m-1}$ has a $2^m-1$ dimensional kernel that is a $\mathcal{O}(\lambda)$-perturbation of $\mathrm{Span}(\{\mathrm{Enc}(\ket{x}+\ket{y}),\mathrm{Enc}(\ket{z}):z\neq x,y\}$.
    \item The first excited state of $\hat{\Delta}_{2m-1}$ is an $\mathcal{O}(\lambda)$-perturbation of $\mathrm{Span}(\{\mathrm{Enc}(\ket{x}-\ket{y})\})$ and with energy $\Theta(\lambda^{4m+2})$.
    \item The next lowest eigenvectors have eigenvalues $\Theta(\lambda^2)$. That is an $\mathcal{O}(\lambda)$-perturbation of $\mathrm{Span}
\Big( \ket{\sigma}: \sigma \in X^{\mathrm{core}}_{2m-1}
\Big)$ of eq.~\eqref{eq:core}.
    \item The other the eigenvalues are $\Theta(1)$.
\end{itemize}
\end{lemma}

We prove this lemma in the following subsections. 

\subsection{0th page}

As we have seen, on the 0th page of the spectral sequences, 
$$
e_d^{0,l} \cong \mathrm{Span}(\ket{\sigma}: \sigma \in  \mathrm{Cl}_d^l(\hat{\mathcal{G}}_m))=: E_d^{0,l}
$$
for all $d,l$, where $\mathrm{Cl}_d^l(\hat{\mathcal{G}}_m)$ is the set of weight $\lambda^l$ $d$-simplices in $\mathrm{Cl}(\hat{\mathcal{G}}_m)$. 
Figure~\ref{fig:filtration_gadget} shows a filtration according to the weight of simplices for the $i$-th block.

\begin{figure}
    \centering
    \includegraphics[width=1\linewidth]{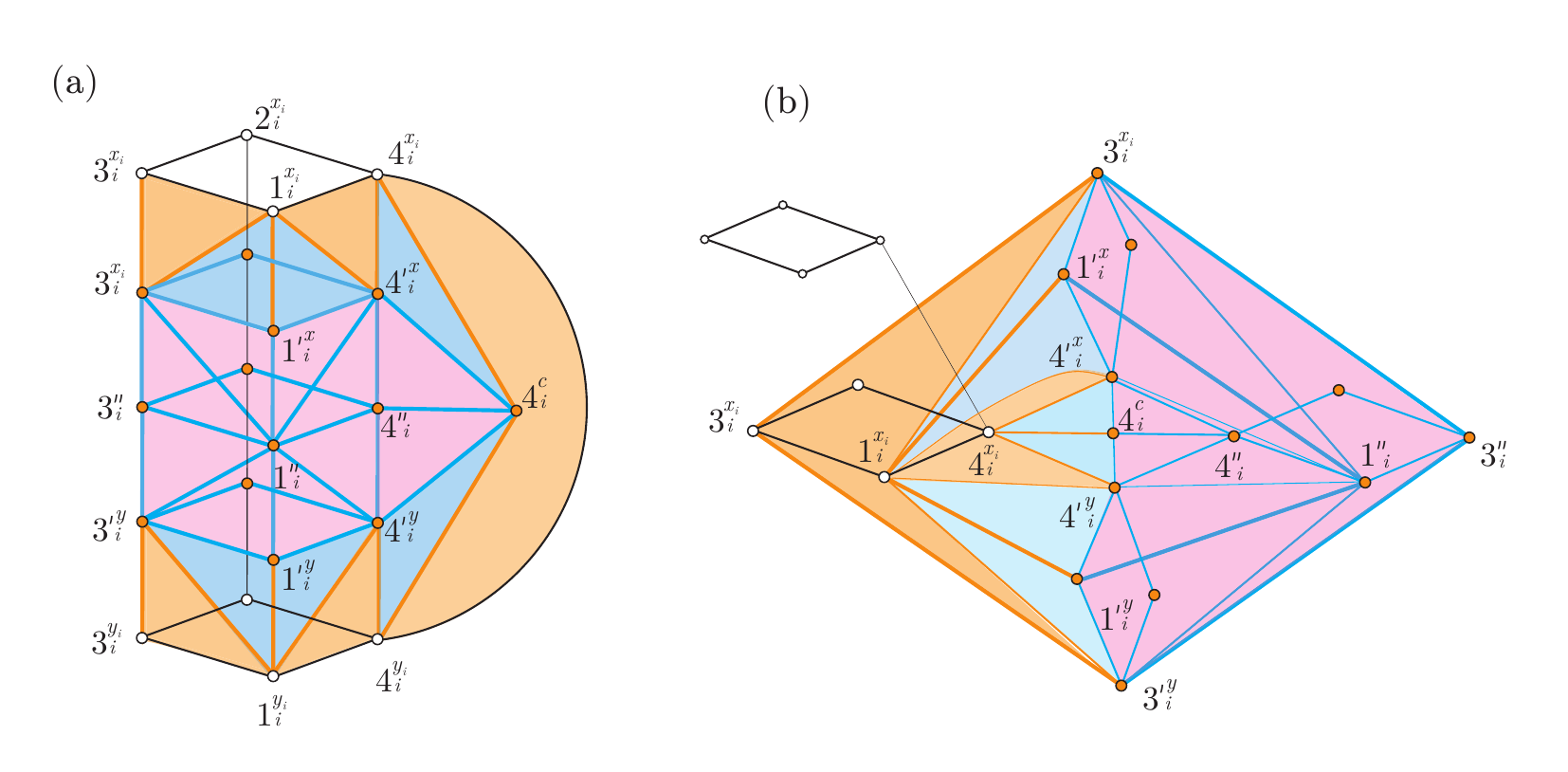}
    \caption{Filtration of the gadget according to the weight of simplices for the $i$-th block. Simplices with the same weight are colored with the same color. }
    \label{fig:filtration_gadget}
\end{figure}

\subsection{1st page}

\paragraph{Leftmost column $e_d^{1,0}$}

For the ``leftmost column'' $e_d^{1,0}$, the only non-trivial subspace appears when $d=2m-1$ because this is simply the homology of the qubit gadget clique complex. 
By the construction, each block of the qubit graph only has a single connected component and two 1-dimensional holes. 
Therefore, 
$$
e_{2m-1}^{1,0} \cong \mathrm{Enc}(\mathcal{H}_m) =: E_{2m-1}^{1,0}
$$
and $e_{d}^{1,0}\cong 0$ for all $d\neq 2m-1$. 

\paragraph{Off-diagonal elements $e_d^{1,l}$ where $l \neq d+1$}

We introduce the following lemma. 
\begin{lemma}[\cite{king:qma}]
\label{lemma:subcomplex}
    Let $\mathcal{P}$ be a simplicial complex and $\mathcal{Q}\subseteq \mathcal{P}$ be a subcomplex. 
    If $\mathcal{Q}$ has no $d-1$-cohomology, 
    $$
    C_d(\mathcal{Q})^\perp \cap \delta_{d-1}  C_{d-1}(\mathcal{Q})^\perp= 
    C_d(\mathcal{Q})^\perp \cap \delta_{d-1}  C_{d-1}(\mathcal{P}).
    $$
\end{lemma}

Then, we claim the following:
\begin{claim}
\label{claim:off_diag}
    $e_d^{1,l}\cong 0$ for $1\leq l \leq d$. 
\end{claim}

\begin{proof}
    We consider a modified gadget with axial vertices $\{0_i^{x_i},0_i^{y_i}\}_i$ with weight $1$ as in Figure~\ref{fig:ancilla}. 
    Then, we can show the claim in a way similar with~\cite{king:qma}.

    \begin{figure}
        \centering
        \includegraphics[width=1\linewidth]{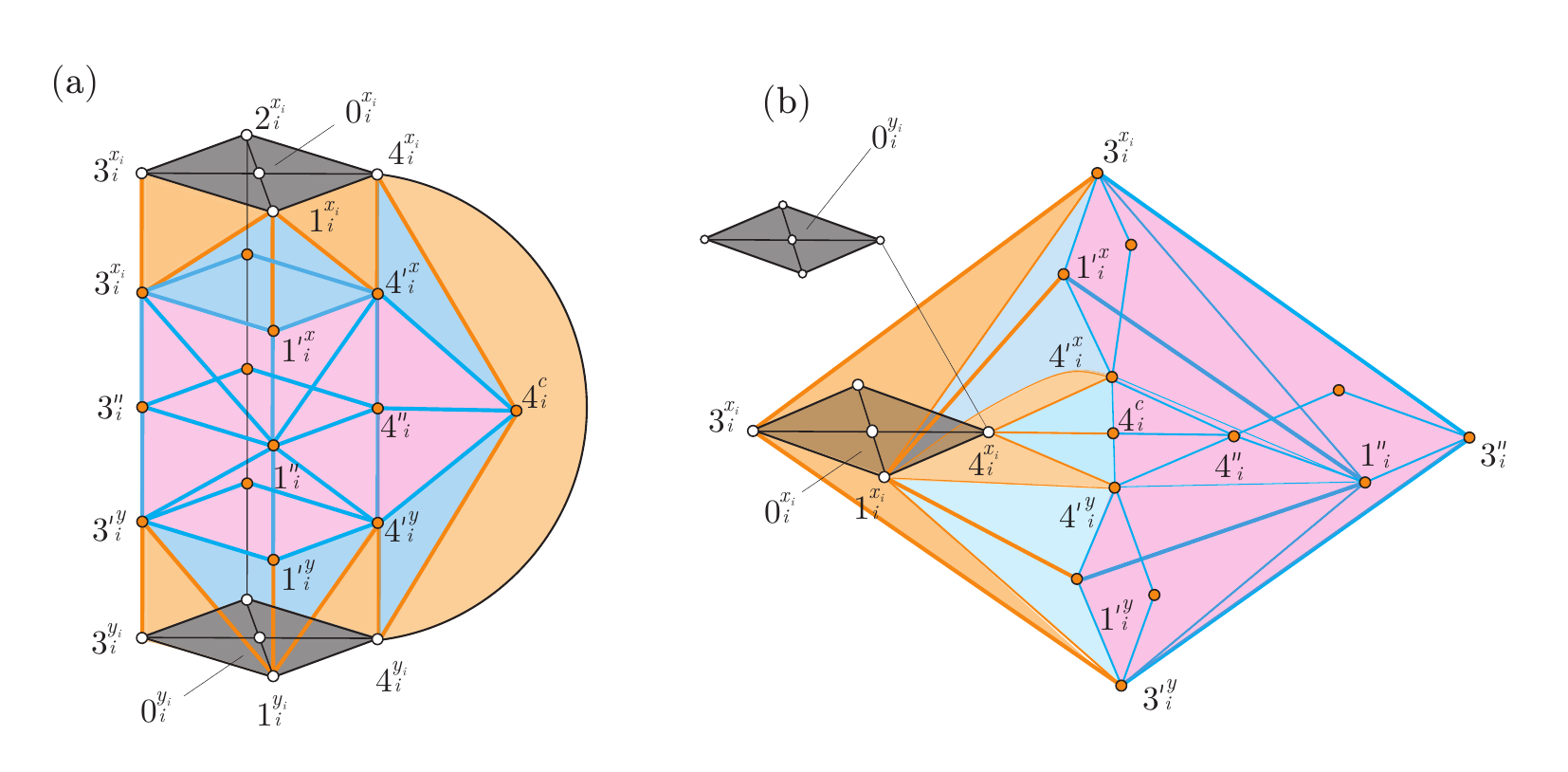}
        \caption{Modification of the gadget with auxiliary vertices $0_i^{x_i}$ and $0_i^{y_i}$ for each block $i$.}
        \label{fig:ancilla}
    \end{figure}

    Let us denote the modified simplicial complex as $\mathrm{Cl}(\hat{\mathcal{G}}'_m))$. Then, let 
    \begin{equation}
       \label{eq:PQ}\mathcal{Q}=\ket{\sigma}: \sigma\in  \mathrm{Cl}^{\leq l-1}(\hat{\mathcal{G}}'_m),\ \ 
    \mathcal{P}=\ket{\sigma}: \sigma\in  \mathrm{Cl}^{\leq l}(\hat{\mathcal{G}}'_m)
    \end{equation}
    where $\mathrm{Cl}^{\leq l}(\hat{\mathcal{G}}'_m))$ is the set of weight $\leq \lambda^l$ simplices in $\mathrm{Cl}(\hat{\mathcal{G}}'_m))$. 
    After the modification with axially qubits, the cohomology of $\mathcal{Q}$ is trivial for any dimension. 

    For any $\ket{\alpha} \in E_d^{0,l}$ s.t. $\ket{\alpha}\in \ker(\delta_{d}^{0,l})$, 
    $\delta_{d} \ket{\alpha}$ must be only supported outside of $\mathrm{Cl}_{d+1}^l(\hat{\mathcal{G}}_m)$. 
    As we are considering $\ket{\alpha} \in E_d^{0,l}$ with $1\leq l \leq d$, 
    no support of $\ket{\alpha}$ is contained in simplices with the newly introduced axial vertices. Therefore, the support of $\delta_{d} \ket{\alpha}$ and $\mathrm{Cl}_{d+1}^{\leq l-1}(\hat{\mathcal{G}}'_m)$ do not overlap as well. 
    However, because the $d$-cohomology of $\mathrm{Cl}^{\leq l-1}(\hat{\mathcal{G}}'_m)$ is trivial, $\ket{\alpha}$ is also a coboundary in $\mathrm{Cl}^{\leq l-1}(\hat{\mathcal{G}}'_m)$. 
    Using Lemma~\ref{lemma:subcomplex} with subcomplexes \eqref{eq:PQ}, 
    as $\ket{\alpha}$ is a coboundary that comes from $ E_{d-1}^{0,l}$ and therefore
    $\ket{\alpha}\in \mathrm{Im} \delta_{d-1}^{0,l}$. 
    Therefore, we can conclude $\ker\delta_{d}^{0,l}/\mathrm{Im} \delta_{d-1}^{0,l}\cong 0$.

\end{proof}



\paragraph{Diagonal elements $e_d^{1,l}$ with $l = d+1$}

For the ``diagonal'' subspaces $e_d^{1,l}$ with $l = d+1$, 
let 
$X^{\mathrm{core}}$ be the set of simplices 
that contains at least one vertex from 
$$
1_i'',2''_i,3''_i,4''_i
$$
for some $i\in[m]$. 
Then, the following holds. 
\begin{claim}
    \begin{equation}\label{eq:core}
    e_d^{1,d+1} \cong \mathrm{Span}
\Big( \ket{\sigma}: \sigma \in X^{\mathrm{core}}_{d}
\Big) =: E_d^{1,d+1}
. 
\end{equation}
\end{claim}

\begin{proof}
    For any $\sigma \in X^{\mathrm{core}}_{d}$, there is no faces or cofaces of $\sigma$ with the same weight of $\sigma$. This means that we cannot add or subtract weight $1$ vertices from such $\sigma$. In contrast, any other weight $\lambda^{d+1}$ $d$-dimensional simplices appears as a boundary operation that removes a weight $1$ vertex. 
\end{proof}

\subsection{2nd page}

On the 2nd page, the leftmost column does not change, and we see that in the diagonal elements, the only subspace that remains non-trivial is 
$e_{2m}^{2,2m+1}$.

\paragraph{Diagonal elements $e_d^{2,d+1}$ with $d \leq  2m-1$}

\begin{claim}
    For $d \leq  2m-1$, $e_d^{2,d}\cong 0$
\end{claim}

\begin{proof}
This can be proven in a way similar to the proof of Claim~\ref{claim:off_diag}.
Again, we consider the modification of Figure~\ref{fig:ancilla}. 

For $\ket{\alpha}\in\mathrm{Span}
\Big( \ket{\sigma}: \sigma \in X^{\mathrm{core}}_{d}
\Big)$, 
assume $\ket{\alpha}\in \ker (\delta_d^{1,d+1})$. 
In this case, $\ket{\alpha}\in \ker \delta_d$ as well. 
Although we have added axial vertices, simplices in the support of $\ket{\alpha}$, the axial vertices do not form valid simplices. Therefore, $\ket{\alpha}$ is a cocycle in $\mathrm{Cl}(\hat{\mathcal{G}}_m' )$. However, $\mathrm{Cl}(\hat{\mathcal{G}}_m' )$ do not have non-trivial homology for any dimension. 
Therefore, $\ket{\alpha}$ must be a coboundary. 

We apply Lemma~\ref{lemma:subcomplex} with 
$
\mathcal{P} =  \mathrm{Cl}(\hat{\mathcal{G}}_m' ),\ \ \mathcal{Q} = \mathrm{Cl}(\hat{\mathcal{G}}_m' ) \backslash X^{\mathrm{core}} 
.
$ Then, we have $\ket{\alpha}\in \delta_{d}(E_d^{1,d+1})$ which means $\ket{\alpha}\in \mathrm{Im}(\delta_{d-1}^{1,d})$.
Therefore, $\ker (\delta_d^{1,d+1}) = \mathrm{Im}(\delta_{d-1}^{1,d}) $.
\end{proof}

\paragraph{Diagonal element $e_{2m}^{2,2m+1}$}
Let us define 
\begin{align*}
    \Ket{E_{2}^{2,3}}_i^x:= \ket{1''_i1'^x_i3'^x_i}-\ket{1''_i3''_i3'^y_i}&+\ket{2''_i3''_i3'^x_i}-\ket{2''_i2'^x_i3'^x_i}
    \\&+\ket{2''_i2'^x_i4'^x_i}-\ket{2''_i4'^x_i4'^x_i}+\ket{1''_i4'^x_i4'^x_i}-\ket{1''_i1'^x_i4'^x_i}
\end{align*}
and
\begin{align*}
    \Ket{E_{2}^{2,3}}_i^y:= \ket{1''_i1'^y_i3'^y_i}-\ket{1''_i3''_i3'^y_i}&+\ket{2''_i3''_i3'^y_i}-\ket{2''_i2'^y_i3'^y_i}
    \\&+\ket{2''_i2'^y_i4'^y_i}-\ket{2''_i4'^y_i4'^y_i}+\ket{1''_i4'^y_i4'^y_i}-\ket{1''_i1'^y_i4'^y_i}.
\end{align*}
These are the sums of triangles described in Figure~\ref{fig:2ndpage_surface} (a).
We also introduce
 $$
\Ket{E_{1}^{2,2}}_{i}^{''}:= \ket{1''_i3''x_i}+\ket{3''_i2''_i}+\ket{2''_i4''_i}+\ket{4''_i1''_i},
$$
 $$
\Ket{E_{1}^{2,2}}_{i}^x:= \ket{1'^x_i3'^x_i}+\ket{3'^x_i2'^x_i}+\ket{2'^x_i4'^x_i}+\ket{4'^x_i1'^x_i},
$$
 $$
\Ket{E_{1}^{2,2}}_{i}^y:= \ket{1'^y_i3'^y_i}+\ket{3'^y_i2'^y_i}+\ket{2'^y_i4'^y_i}+\ket{4'^y_i1'^y_i}.
$$
These are 1-dimensional cycles described in Figure~\ref{fig:2ndpage_surface} (b).

\begin{figure}
    \centering
    \includegraphics[width=0.65\linewidth]{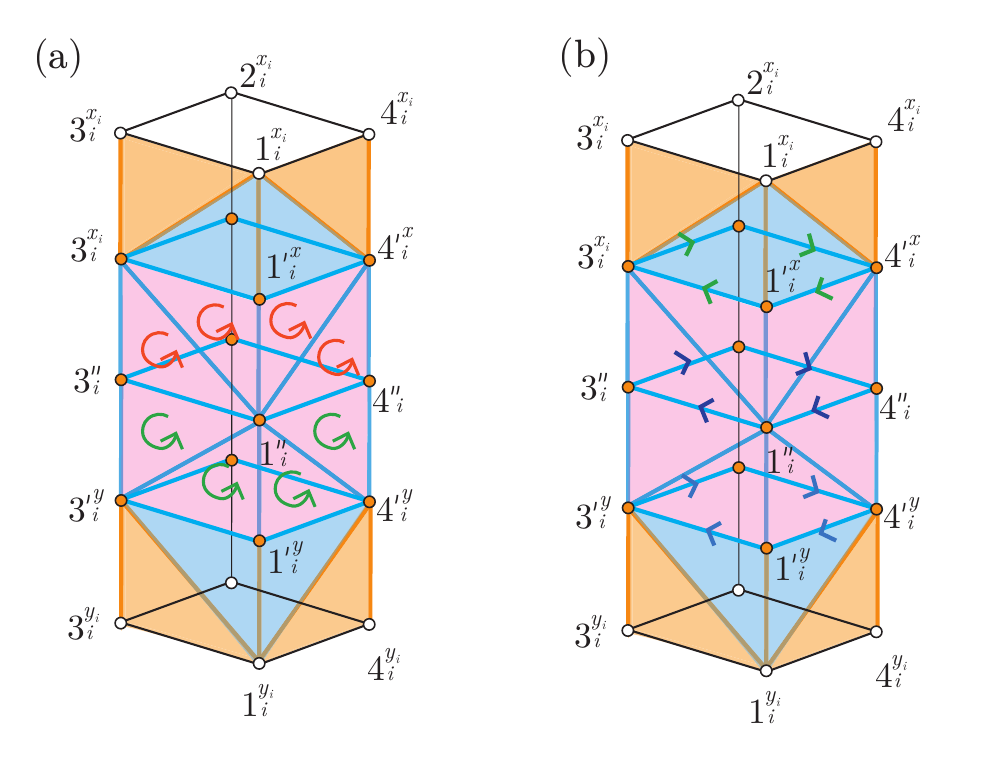}
    \caption{Components that appear in the 2nd page. We are omitting the vertex $4_i^c$ because it is irrelevant at this point. The cases whether $x_i=y_i$ or $x_i\neq y_i$ do not matter here as well.}
    \label{fig:2ndpage_surface}
\end{figure}

Then, we introduce a state
$$
\Ket{E_{2m}^{2,2m+1}}:=
\left(
\bigotimes_{j=i}^{i-1}\Ket{E_{1}^{2,2}}_{j}^{''}\right)\Ket{E_{2}^{2,3}}_i^x\left(\bigotimes_{j=i+1}^m\Ket{E_{1}^{2,2}}_{j}^x\right)
-
\left(
\bigotimes_{j=i}^{i-1}\Ket{E_{1}^{2,2}}_{j}^{''}\right)\Ket{E_{2}^{2,3}}_i^y\left(\bigotimes_{j=i+1}^m\Ket{E_{1}^{2,2}}_{j}^y\right)
$$
whose boundary is 
$$
\partial_{2m}^{1,2m+1}\Ket{E_{2m}^{2,2m+1}}=\bigotimes_{j=i+1}^m\Ket{E_{1}^{2,2}}_{j}^x-\bigotimes_{j=i+1}^m\Ket{E_{1}^{2,2}}_{j}^y \cong 0
$$
because it is completely outside of $E_{2m-1}^{1,2m}$. 
{There is no other cycle under $\partial_{2m}^{1,2m+1}$.} 
Therefore, 
$e_{2m}^{2,2m+1}\cong \mathrm{Span}\Big(\Ket{E_{2m}^{2,2m+1}}\Big)$.

\subsection{After 3rd page}

From 3rd page to $2m+1$st pages, no change occurs for every $e_d^{k,l
}$. 

In $2m+2$nd page, $e_{2m}^{2m+2, 2m+2}$ becomes trivial and 
$$
e_{2m-1}^{2m+2,0}\cong \mathrm{Span}\left(\{\mathrm{Enc}(\ket{x'})\}_{x'\neq x,y \in \{0,1\}^m} \right) \oplus \mathrm{Span} (\mathrm{Enc}(\ket{x}+\ket{y})).
$$
This is because  
$$
\delta_{2m-1}^{2m+1,0} \mathrm{Enc} (\ket{x}-\ket{y}) = \Ket{E_{2m}^{2,2m+1}}.
$$
As a consequence, $\Ket{E_{2m}^{2,2m+1}}$ now becomes a coboundary under $\delta_{2m-1}^{2m+1,0}$ and $\mathrm{Enc} (\ket{x}-\ket{y})$ becomes a boundary under $\partial_{2m}^{2m+1,2m+1}$. 

It is clear that after page $2m+2$, there will be no change. 

\subsection{Summary}

We can obtain $E_d^{k,l}$ by taking representative of $e_d^{k,l}$ in the space $\mathrm{Span}(\mathrm{Cl}_d^l(\hat{\mathcal{G}}_m))$. 
We conclude that 
\begin{itemize}
    \item $E_{2m-1}^0= \bigoplus_{l=0}^{2m} \mathrm{Span}(\ket{\sigma}: \sigma \in  \mathrm{Cl}_{2m-1}^l(\hat{\mathcal{G}}_m))$
    \item 
$E_{2m-1}^1=\mathrm{Enc}(\mathcal{H}_m)
    \oplus 
    \mathrm{Span}
\Big( \ket{\sigma}: \sigma \in X^{\mathrm{core}}_{2m-1}
\Big)    .$
    \item 
    $E_{2m-1}^2=E_{2m-1}^3=\cdots = E_{2m-1}^{2m+1}=\mathrm{Enc}(\mathcal{H}_m)$
    \item $E_{2m-1}^j= \mathrm{Span}\left(\{\mathrm{Enc}(\ket{z})\}_{z\neq x,y \in \{0,1\}^m} \right) \oplus \mathrm{Span} (\mathrm{Enc}(\ket{x}+\ket{y}))$ 
    for $j\geq 2m+2$. 
\end{itemize}
These concludes Lemma~\ref{lemma:singlegadget_spectrum}.

\subsection{Construction and analysis for combined gadgets}

Next, we combine the gadgets for single terms and analyze the spectral properties of the resulting simplicial complex. 

Recall that we reduce from
$H=\sum_{i=1}^t h_i$ where each term comes from $ S_{\mathrm{stoq}}$ and 
$\hat{\mathcal{G}}_{m_i}$ 
is the gadget graph for $h_i$. 
Suppose that $h_i$ is a projector on qubits $\mathcal{I}_i$ where $|\mathcal{I}_i|=m_i$. 
Then, corresponding to the tensor product with identity 
$$
h_i\otimes I_{\bar{\mathcal{I}_i}},
$$
where $I_{\bar{\mathcal{I}_i}}$ is the identity operator for other qubits than these $h_i$ acts non-trivially, 
we construct a graph
$$\hat{\mathcal{G}}_n(h_i)= \hat{\mathcal{G}}_{m_i} * \mathcal{G}_{\bar{\mathcal{I}_i}} $$
where 
$\mathcal{G}_{\bar{\mathcal{I}_i}}$ are the join product of $n-m_i$ graphs of ${j\in [n]\backslash {\mathcal{I}_i}}$.
Now, for the Laplacian $\hat{\Delta}'_{2n-1,i}$ of $\mathrm{Cl}(\hat{\mathcal{G}}_n(h_i))$, 
it holds that 
$$
\hat{\Delta}'_{2n-1,i}=\hat{\Delta}_{2m_i-1}\otimes I + I\otimes \Delta_{2(n-m)-1}
$$
where $\Delta_{2(n-m)-1}$ is the Laplacian for $\mathrm{Cl}_{2(n-m)-1}(\mathcal{G}_{\bar{\mathcal{I}_i}})$.

Let 
$$\tilde{G}_i:= \hat{\mathcal{G}}_n(h_i) \backslash \mathcal{G}_n.$$ 
Let also 
$$\hat{\mathcal{G}}_n(H):= \mathcal{G}_n  \sqcup \tilde{G}_1 \sqcup \cdots \sqcup \tilde{G}_t.$$ 
For each $h_i$, we define $\mathcal{T}_i$ as 
$$
\mathcal{T}_i:= \mathrm{Cl}(\hat{\mathcal{G}}_{m_i})\backslash \mathrm{Cl}(\mathcal{G}_{n})
$$
where $\mathrm{Cl}(\mathcal{G}_{n})$ is the clique complex of the qubit graph $\mathcal{G}_n$. 
Then, 
$$\mathrm{Cl}(\hat{\mathcal{G}}_{n}(H))
= \mathrm{Cl}(\mathcal{G}_{n}) \sqcup \mathcal{T}_1 \sqcup \cdots \sqcup\mathcal{T}_t.
$$
Therefore, for $d=0,1,...,2n$, 
$$C_{d}(\hat{\mathcal{G}}_n(H))= C_{d}({\mathcal{G}}_n)\oplus C_{d}({\mathcal{T}}_1)\oplus\cdots\oplus C_{d}({\mathcal{T}}_t),$$
where $C_{d}(\hat{\mathcal{G}}_n(H))$ and $C_{d}({\mathcal{G}}_n)$ are 
the $k$-th chain space of $\mathrm{Cl}(\hat{\mathcal{G}}_n)$ and $\mathrm{Cl}({\mathcal{G}}_n)$. 

The full $2n-1$-dimensional Laplacian is $\hat{\Delta}_{2n-1}:C_{2n-1}(\hat{\mathcal{G}}_n(H))\rightarrow C_{2n-1}(\hat{\mathcal{G}}_n(H))$. 
The up Laplacian can be decomposed into parts corresponding to each of the gadgets 
$$
\hat{\Delta}^{\text{up}}_{2n-1}=\hat{\Delta}^{\text{up}}_{2n-1,1}+\cdots+\hat{\Delta}^{\text{up}}_{2n-1,t}
$$
because $2n$-dimensional simplices do not overlap among different gadgets. 
We use 
$\hat{\delta_d}, \hat{\partial}_d$ to denote the $d$-coboundary and boundary operator on the final complex.

\paragraph{Proof of the lower bound of the minimum energy in the NO instances}

We show the lower bound for the minimal eigenvalue of the Laplacian in NO instances following the strategy of \cite{king:qma}. 

\begin{proposition}
    Let $H$, 
    $\hat{\mathcal{G}}_n$, $\hat{\Delta}_{2n-1}$, 
    for any $g>0$, 
    there exist a sufficiently small constant $\alpha>0$ s.t. with 
    $$\lambda=ct^{-1}g$$
    it holds that 
    if $\lambda_0(H)=0$ then $\lambda_0(\hat{\Delta}_{2n-1})=0$ and if $\lambda_0(H)\geq g$ then $\lambda_0(\hat{\Delta}_{2n-1})\geq c\lambda^{4m+2}t^{-1}g$.
\end{proposition}

\begin{proof}

The proof follows the proof of Theorem 10.1 in~\cite{king:qma}. 
Let us introduce
\begin{itemize}
\item $\Pi_{2n-1,0}$: projector onto the chain space of the qubit graph $C_{2n-1}({\mathcal{G}}_n)$.
\item $\{\Pi_{2n-1,i}\}_i\in [t]$: projector onto $C_{2n-1}({\mathcal{T}}_i)$. 
\end{itemize}
Then, $$I_{C_{2n-1}(\hat{\mathcal{G}}_n)}=\sum_{i=0}^t\Pi_{2n-1,i}$$
where $I_{C_{2n-1}(\hat{\mathcal{G}}_n)}$ is the identity on the full chain space. 
Then, we introduce the reformulation of the single gadget chain space $\Pi_{2n-1,0}+\Pi_{2n-1,i}$ as 
$$\Pi_{2n-1,0}+\Pi_{2n-1,i}= \Pi_i^{(\mathcal{A})}+\Pi_i^{(\mathcal{B})}+\hat{\Phi}_i+\hat{\Phi}_i^\perp$$
where the projectors in the rhs are defined through the spectral properties of each single gadget Laplacian $\hat{\Delta}_{2n-1,i}$ on $\hat{\mathcal{G}}_n(h_i)$ as
    \begin{itemize}
        \item $\Pi_i^{(\mathcal{A})}$ is the projector onto the space of eigenvectors with eigenvalue $\Theta(1)$, 
        \item $\Pi_i^{(\mathcal{B})}$ is the projector onto the space of eigenvectors with eigenvalue $\Theta(\lambda^2)$, 
        \item $\hat{\Phi}_i$ is the projector onto the space of eigenvectors with eigenvalue $\Theta(\lambda^{4m_i+2})$, 
        \item $\hat{\Phi}_i^\perp$ is the projector onto $\ker(\hat{\Delta}'_{2n-1,i})$. 
    \end{itemize}
Note that this works for both the gluing gadget and the gadget for filling a computational basis state because they share similar spectral properties, as can be seen from Lemma~\ref{lemma:spectrum_filling} and Lemma~\ref{lemma:singlegadget_spectrum}.

Suppose that there is a normalized state $\ket{\varphi}$ s.t.
$$
\bra{\varphi}\hat{\Delta}_{2n-1}\ket{\varphi}<E. 
$$
Then, it holds that \cite{king:qma}
\begin{align*}
    \braket{\varphi|\varphi} \leq 
    &(1-g)\bra{\varphi}\Pi^{(\mathcal{H})}_n\ket{\varphi}-(t-1)\bra{\varphi}(\Pi^{2n-1}_0-\Pi^{(\mathcal{H})}_n)\ket{\varphi}+\mathcal{O}(\lambda t) \\&+
    \bra{\varphi}\sum_i \hat{\Phi}_i \ket{\varphi} + \bra{\varphi}\sum_i \Pi^{(\mathcal{A})}_i \ket{\varphi}+\bra{\varphi}\sum_i \Pi^{(\mathcal{B})}_i \ket{\varphi}. 
\end{align*}
For each of the gadget $\hat{\mathcal{G}}_n(h_i)$, we define projector $\Pi^{[d]}_i$ as 
\begin{itemize}
    \item If $h_i$ is a projector onto a computational basis state, $\Pi^{[d]}_i$ is a projector onto space $C_d([\text{bulk}]_i)$ spanned by $d$-simplices containing the central vertex of the gadget construction for $h_i$. 
    \item If $h_i$ is a projector onto $\ket{x}-\ket{y}$, $\Pi^{[d]}_i$ is a projector onto the space $C_d([\text{bulk}]_i):=\mathrm{Span}(\ket{\sigma}: \sigma \in X^{\mathrm{core}}_{d})$. 
\end{itemize}

Then, we can show the following claims. 
\begin{claim}
\label{claim:fin0}
For $\ket{\hat\phi}$ that is $\mathcal{O}(\lambda)$-perturbation of the state penalized by some $h_i$ and an eigenvector with eigenvalue $\Theta(\lambda^{4m_i+2})$ for the single gadget Laplacian $\hat{\Delta}_{2m-1}(h_i)$, 
$
\hat{\partial}_{2n-1} \ket{\hat\phi} \otimes \mathrm{Enc}(\mathcal{H}_{n-m_i}) = 0,
$
where $\hat{\partial}_{2n-1}$ is the boundary operator for the single gadget complex. 

\end{claim}

\begin{claim}
\label{claim:fin1}
    For every $h_i$, all states $\ket{\psi}$ have $\|\Pi^{[2n-2]}_i\hat{\partial}_{2n-1}\ket{\psi}\|=\mathcal{O}(\lambda)\|\ket{\psi}\|$ and $\|\Pi^{[2n]}\hat{\delta}_{2n-2}\ket{\psi}\|=\mathcal{O}(\lambda)$. 
\end{claim}


\begin{claim}
\label{claim:fin2}
    For every $h_i$, 
    a normalized state $\ket{\psi}\in C_{2n-1}([\text{bulk}])$ has $\|\Pi^{[2n-2]}_i\hat{\partial}_{2n-1}\ket{\psi}\|=\Omega(\lambda)$ or $\|\Pi^{[2n-2]}\hat{\delta}_{2n-1}\ket{\psi}\|=\Omega(\lambda)$. 
\end{claim}

The above three claims have been shown for gadgets for filling computational basis states in \cite{king:qma}. 
These claims can be similarly shown for gluing gadgets as well.

Based on these Claims~\ref{claim:fin1},~\ref{claim:fin2}, we can show the following inequalities:

    $$\bra{\varphi}\sum_i \hat{\Phi}_i \ket{\varphi}=\mathcal{O}(\lambda^{-(4m+2)}Et)$$
    $$\bra{\varphi}\sum_i \Pi^{(\mathcal{A})}_i \ket{\varphi}=\mathcal{O}(\lambda^2t)$$
    $$\bra{\varphi}\sum_i \Pi^{(\mathcal{B})}_i \ket{\varphi}=\mathcal{O}(\lambda^2t).$$
These inequalities are shown in Lemmas 10.2, 10.3, 10.4 of \cite{king:qma} for gadgets for filling computational basis states, and can be similarly shown for gluing gadgets.

Therefore, 
\begin{align*}
    \braket{\varphi|\varphi} \leq 
    &(1-g)\bra{\varphi}\Pi^{(\mathcal{H})}_n\ket{\varphi}-(t-1)\bra{\varphi}(\Pi^{2n-1}_0-\Pi^{(\mathcal{H})}_n)\ket{\varphi}+\mathcal{O}(\lambda t) \\&
    +\mathcal{O}(\lambda^{-(4m+2)}Et) + \mathcal{O}(\lambda^2t). 
\end{align*}
With the choice of $\lambda=ct^{-1}g$ and $E=c\lambda^{4m+2}t^{-1}g$ for a sufficiently small constant $c$, 
\begin{align*}
    1&=\braket{\varphi|\varphi}\\&
    \leq (1-g)\bra{\varphi}\Pi^{(\mathcal{H})}_n\ket{\varphi}-(t-1)\bra{\varphi}(\Pi^{2n-1}_0-\Pi^{(\mathcal{H})}_n)\ket{\varphi}+\frac{1}{10}g \\
    &\leq 1-g+\frac{1}{10}g <1,
\end{align*}
which is a contradiction.

\end{proof}

\section{Uniform orientable filtration of the gadget simplicial complex}
\label{sec:hardness_3}

In this section, we first show the existence of a uniform orientable filtration for the constructed simplicial complexes with gadgets.  
Next, we 
 show several properties for the case of the reduction from YES instances of the stoquastic SAT problem.

\subsection{Orientable filtration of a gluing gadget}

\begin{proposition}
There is a uniform orientable filtration for 
    $\mathrm{Cl}_{2n-1}(\hat{\mathcal{G}}_n(H))$. 
\end{proposition}

\begin{proof}
First, let 
$$\mathcal{G}_{\text{disj}}= \bigotimes_{i=1}^n \left(\lsinglegraph\rsinglegraph\right) $$
be the disjoint version of the qubit graph $\mathcal{G}_n$ i.e., removing edges between $4_i^0$--$4_i^1$ for any $i$. 
Then, 
we let 
$$X_{2n-1}^0:=\mathrm{Cl}_{2n-1}(\mathcal{G}_{\text{disj}})$$ 
that is the set of simplices in the qubit gadget complex. 
It can be seen that $\mathrm{Cl}_{2n-1}(\mathcal{G}_{\text{disj}})$ is down-degree 2. 

For each $i\in [t]$, we continue the procedure of filtration for the corresponding projector $h_i$ as follows:

\paragraph{(1) The case $h_i=\frac{1}{2}(\ket{x}-\ket{y})(\bra{x}-\bra{y})$}

Recall the subsets in eq.~\eqref{eq:filtration_thickenig}
$$L_m^k := \{1,2\}\times \cdots \times \{2k-3,2k-2\} \times \{(2k-1)(2k-1'),(2k)(2k')\} \times \{2k+1',2k+2'\} \times \cdots \times \{2m-1',2m'\}$$
for $L$ that glues two target generalized octahedra and ${L}_{m-1}^k$, as well as 
\begin{equation*}
    \hat{L}_{m-1}^k := \{1,2\}\times \cdots \times \{2k-3,2k-2\} \times \{2k-1',2k'\} \times \{2k+1',2k+2'\} \times \cdots \times \{2m-1',2m'\}
\end{equation*}

Also, recall the structure of our single gadget:
$$
K^x \rightarrow K'^x \leftarrow K'' \rightarrow K'^y \leftarrow K^y.
$$
We use the following notations for each $h_i$:
\begin{itemize}
    \item $L^{xx}$: the simplicial complex that glues $K^x$ and $K'^x$. 
    \item $L^{yy}$: the simplicial complex that glues $K^y$ and $K'^y$.
    \item $L^{''x}$: the simplicial complex that glues $K''$ and $K'^x$.
    \item $L^{''y}$: : the simplicial complex that glues $K''$ and $K'^y$.
\end{itemize}
Then we can analogously define 
$$
(L^{xx})_{2{m_i}-1}^k, (L^{yy})_{2{m_i}-1}^k, (L^{''x})_{2{m_i}-1}^k, (L^{''y})_{2{m_i}-1}^k
$$
for $k=0,1,2,...,2m$. 
We can take the join with the clique complex of the disconnected qubit graph corresponding to taking the product with identity $h_i\otimes I$. 
Let
$$\bar{\mathcal{G}}_{\text{disj}}^{h_i}= \bigotimes_{i\in [n]\backslash [h_i]}^n \left(\lsinglegraph\rsinglegraph\right) $$
where the product is taken for the indices of the qubits $h_i$ acts trivially. 
Then let
\begin{align*}
   &
   (\tilde L^{xx})_{2{m_i}-1}^k:=(L^{xx})_{2{m_i}-1}^k*\mathrm{Cl}_{2n-2m_i-1}\big(\bar{\mathcal{G}}_{\text{disj}}^{h_i}\big),\ \ \ \ (\tilde L^{yy})_{2{m_i}-1}^k:=(L^{yy})_{2{m_i}-1}^k*\mathrm{Cl}_{2n-2m_i-1}\big(\bar{\mathcal{G}}_{\text{disj}}^{h_i}\big), \\
&(\tilde L^{''x})_{2{m_i}-1}^k:=(L^{''x})_{2{m_i}-1}^k*\mathrm{Cl}_{2n-2m_i-1}\big(\bar{\mathcal{G}}_{\text{disj}}^{h_i}\big), \ \ \ \ 
(\tilde L^{''y})_{2{m_i}-1}^k:=(L^{''y})_{2{m_i}-1}^k*\mathrm{Cl}_{2n-2m_i-1}\big(\bar{\mathcal{G}}_{\text{disj}}^{h_i}\big).
\end{align*}
These simplices will be gradually added to the subset $X_d^k$ before this procedure, as   
\begin{enumerate}
    \item  $X_{2n-1}^{k+1}=X_{2n-1}^{k}\ \sqcup\ (\tilde L^{xx})_{2{m_i}-1}^1\ \sqcup\ (\tilde L^{yy})_{2{m_i}-1}^1$
    \item $\cdots$
    \item $X_{2n-1}^{k+2m_i}=X_{2n-1}^{k+2m_i-1}\ \sqcup \ (\tilde L^{xx})_{2{m_i}}^{2{m_i}}\ \sqcup\ (\tilde L^{yy})_{2{m_i}}^{2{m_i}}$
    \item $X_{2n-1}^{k+2m_i+1}=X_{2n-1}^{k+2m_i}\ \sqcup \ (\tilde L^{''x})_{2{m_i}-1}^{2{m_i}-1}\ \sqcup\ (\tilde L^{yy})_{2{m_i}-1}^{2{m_i}-1}$
     \item $X_{2n-1}^{k+2m_i+2}=X_{2n-1}^{k+2m_i+1}\ \sqcup \ (\tilde L^{''x})_{2{m_i}-1}^{2{m_i}-2}\ \sqcup\ (\tilde L^{yy})_{2{m_i}-1}^{2{m_i}-2}$
    \item $\cdots$
    \item $X_{2n-1}^{k+4m_i-1}=X_{2n-1}^{k+4m_i-2}\ \sqcup \ (\tilde L^{xx})_{2{m_i}-1}^{1}\ \sqcup\ (\tilde L^{yy})_{2{m_i}-1}^{1}$
    \item $X_{2n-1}^{k+4m_i}=X_{2n-1}^{k+4m_i-1}\ \sqcup \ K''_{2m_i-1}*\mathrm{Cl}_{2n-2m_i-1}\big(\bar{\mathcal{G}}_{\text{disj}}^{h_i}\big)$
\end{enumerate}

\paragraph{(2) The case $h_i=\ket{x}\bra{x}$}

In this case, the generalized octahedron is first copied to $K'^x$
$$K^x \rightarrow K'^x$$
with a gluing procedure, and then the copied cycle is filled with a central vertex. 
Starting from $X_d^k$, we can similarly construct a filtration 
$$X_{2n-1}^k\subseteq X_{2n-1}^{k+1} \subseteq \cdots \subseteq X_{2n-1}^{k+2m_i}. $$
Note that the simplices with the central vertices are also added to $X_d^{k+2m_i}$. 

With the above procedure, 
we can construct 
\begin{equation}
\label{eq:uniform_filtration_gadget}
   X_{2n-1}^0\subseteq X_{2n-1}^1\subseteq \cdots \subseteq X_{2n-1}^N 
\end{equation}
for some $N\in poly(n)$. 

As can be seen from Lemma~\ref{lemma:orientable_filtration_octa}, there is a uniform orientable filtration for each of the gluing procedures. 
The same can be said for the combined gadgets and the constructed filtration of eq.~\eqref{eq:uniform_filtration_gadget}. 
The weighting of vertices keeps the filtration uniform as well. 

\end{proof}

We remark that the constructed filtration is different from the filtration that we utilized for the analysis with spectral sequences.
Also, there are remaining $2n-1$ simplices that are not contained in  $X_d^N$. 
These are any $2n-1$ simplices that contain vertices $4_i^04_i^1$ or $4''_i$ for some $i\in [n]$. 
As we will see in the next subsection, such simplices \uline{do not appear} in the harmonics that correspond to the homologous cycle in $X_d^0$ in the YES instances.

\subsection{Harmonics in the YES instances}

In this subsection, we prove that the Harmonics in the YES instances are only supported on $X_{2n-1}^N$ of eq.~\eqref{eq:uniform_filtration_gadget}. 

\paragraph{Harmonics of a single gadget}

We first investigate harmonics of a single gadget. 
Figure~\ref{fig:harmonics_gadget} shows the harmonics of a single gadget for the 2-dimensional case. 
\begin{figure}
    \centering
    \includegraphics[width=0.5\linewidth]{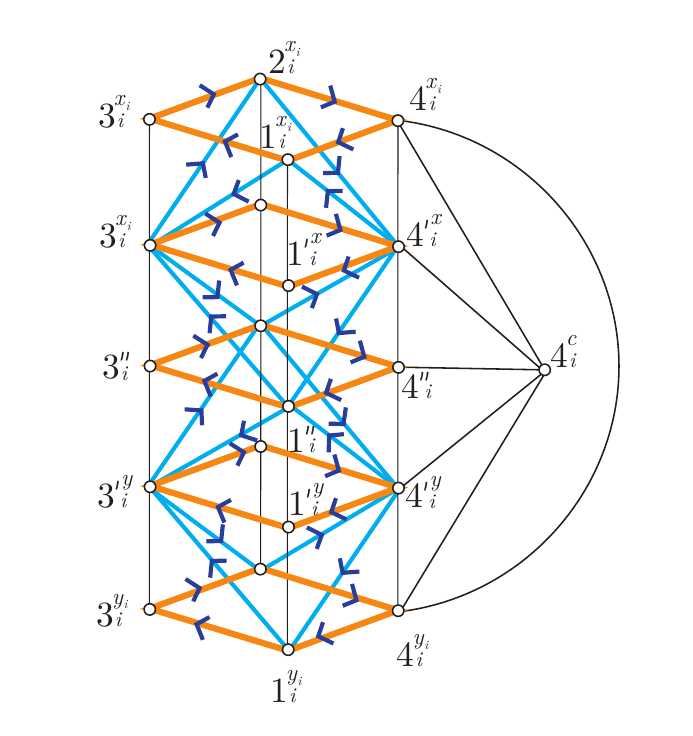}
    \caption{Harmonics of the gadget for the 2-dimensional case.}
    \label{fig:harmonics_gadget}
\end{figure}
The important point here is that the harmonics are not supported on the additional simplices with the vertex $4_i^c$. 
A similar property holds for gluing gadgets in general. 

Recall from Lemma~\ref{lemma:harmonics_octa}, 
the harmonics of $L_{2n-1}$ that glues $K$ and $K'$ is given by
\begin{equation}
  \sum_{\sigma_0 \in \hat{L}_{2n-1}^0} \ket{\sigma_0} + \sum_{\sigma_1 \in \hat{L}_{2n-1}^1} \ket{\sigma_1} + \cdots + \sum_{\sigma_n \in \hat{L}_{2n-1}^{2n}} \ket{\sigma_{2n}}.
\end{equation}
Notice that when the vertices are weighted, this is equivalent to 
\begin{equation}
  \sum_{\sigma_0 \in \hat{L}_{2n-1}^0} \frac{1}{w(\sigma_0)}\ket{[\sigma_0]} + \sum_{\sigma_1 \in \hat{L}_{2n-1}^1}\frac{1}{w(\sigma_1)} \ket{[\sigma_1]} + \cdots + \sum_{\sigma_n \in \hat{L}_{2n-1}^{2n}} \frac{1}{w(\sigma_{2n})}\ket{[\sigma_{2n}]}
\end{equation}
with a normalized basis. 
We can synthesize such harmonics for the construction
$$
K^x \rightarrow K'^x \leftarrow K'' \rightarrow K'^y \leftarrow K^y.
$$
to construct a harmonics $\ket{\phi^{h_i}}$ of the single gluing gadget $h_i$. 
Therefore, starting from 
$\mathrm{Enc}(\ket{x})+\mathrm{Enc}(\ket{y})$, we can ``harmonize'' this state as 
\begin{equation}
   \ket{\phi^{xy}}= \mathrm{Enc}(\ket{x})+ \ket{x\leftrightarrow y} +\mathrm{Enc}(\ket{y}), 
\end{equation}
where $\ket{x\leftrightarrow y}$ is the sum of cycles that appear as the propagation of $\mathrm{Enc}(\ket{x})$ and $\mathrm{Enc}(\ket{y})$. $\ket{\phi^{xy}}$ is only supported on the simplices for gluing. 

Now we show the following claim. 

\begin{proposition}
Let $H=\sum_i h_i$ be a local Hamiltonian composed of projectors from $S_{\text{stoq}}$. 
If $\ker H$ is non-empty (i.e., YES instance), then there is a homologous cycle $\ket{c}$ that is a uniform superposition of a subset of simplices in $\mathrm{Cl}_{2n-1}({\mathcal{G}}_n)$.  
\end{proposition}

\begin{proof}
Let $\ket{\psi}\in \ker H$ be a subset state that indeed exists in YES instances~\cite{bravyi2010complexity, aharonov2019stoquastic}. 
For any local projector $h_i=\ket{x}\bra{x}$ onto computational basis states, it must hold that 
$$
h_i \otimes I \ket{\psi}=0.
$$
Let us rewrite 
    $$
    h_i \otimes I= \sum_{z\in \{0,1\}^{n-m_i}} \ket{x\cup z}\bra{x\cup z}.
    $$
Then, because $\ket{\psi}$ is a non-negative state, it must hold 
$$\bra{x\cup z}\psi\rangle=0$$
for any $z$. 
This means that the bit strings in the support of $\ket{\psi}$ must not contain $x$. 
Therefore, 
for the encoding of $\ket{\psi}$ into the chain space of the qubit gadget complex $\mathrm{Enc}(\ket{\psi})$, 
the simplices in the support of $\mathrm{Enc}(\ket{\psi})$ do not overlap with the gadget complex for $h_i$. In other words, the coboundary $\delta_{2n-1}\mathrm{Enc}(\ket{\psi})$ is not supported on the gadget simplices added for $h_i$ s.t. $h_i$ is a projector onto a computational basis state. 

Next, consider projectors onto $\ket{x}-\ket{y}$ for some $x,y$. 
There may be projectors $\frac{1}{2}(\ket{x}-\ket{y})(\bra{x}-\bra{y})$ where the support of $\ket{\psi}$ do not contain $x$ and $y$. 
It is clear that $\delta_{2n-1}\mathrm{Enc}(\ket{\psi})$ is not supported on the gadget simplices for such projectors. 
Therefore, in the remainder of the proof, we only consider the case in which the support of $\ket{\psi}$ contains $x$ and $y$. 
Let $\tilde H^\psi=\sum_i h_i$ be the sum of such projectors. 
In this case, for a single term $h_i$, $\ket{\psi}$ must be written in the following form 
$$
\ket{\psi}= (\ket{x}+\ket{y}) \otimes \ket{\bar\psi}+ \sum_{z \neq x,y} \ket{z} \ket{\bar\psi'_z}
$$
for some (unnormalized) states $\ket{\bar\psi},\{\ket{\bar\psi'_z}\}_{z\neq x,y}$.
Now consider the encoding of $\ket{\psi}$ into the chain space
$$
\mathrm{Enc}(\ket{\psi})= \left(\mathrm{Enc}(\ket{x})+\mathrm{Enc}(\ket{y})\right) \otimes \mathrm{Enc}(\ket{\bar\psi})+ \mathrm{Enc}(\sum_{z \neq x,y} \ket{z} \ket{\bar\psi'_z}). 
$$
As only the first term has overlap with the gluing gadget for $h_i$, we can modify this state into 
$$
 \left(\mathrm{Enc}(\ket{x})+\ket{x\leftrightarrow y}+\mathrm{Enc}(\ket{y})\right) \otimes \mathrm{Enc}(\ket{\bar\psi})+ \mathrm{Enc}(\sum_{z \neq x,y} \ket{z} \ket{\bar\psi'_z}). 
$$
We can continue this procedure for every term in $\tilde H^\psi$ by adding the corresponding $\ket{x\leftrightarrow y}$ to construct a harmonic state. 
The constructed harmonics and $\mathrm{Enc}(\ket{\psi})$ represent the same ``hole'' because they can be translated with each other with only the procedure of adding boundaries and rescaling. 
\end{proof}



\section*{Acknowledgment}

The authors thank Jiaqing Jiang for explaining her work \cite{jiang2023local}.
RH thanks Atsuya Hasegawa for early discussions on the topic. 
 This work was supported by the Dutch National Growth Fund (NGF), as part of
the Quantum Delta NL programme. This publication is part of the project Divide \& Quantum (with project
number 1389.20.241) of the research programme NWA-ORC which is (partly) financed by the Dutch Research Council (NWO). This work was also partially supported by the Dutch Research Council (NWO/OCW), as part of the Quantum Software Consortium programme (project number 024.003.03), and co-funded by the European Union (ERC CoG, BeMAIQuantum, 101124342).
RH was supported by JSPS PRESTO Grant Number JPMJPR23F9 and MEXT KAKENHI Grant Number 21H05183, Japan.

\bibliographystyle{alpha}
\bibliography{main} 

\appendix

\section{Examples}

\subsection{Gluing 1-dimensional holes}

Let us first investigate the 1-dimensional case. 
We use $i$ for $(i,0)$ and $i'=(i,1)$ for simplicity where $i$ is the index for vertices. 
The original 1-dimensional simplices in the copies $K,K'$ are
$$[13],[32],[24],[41]$$
$$[1'3'],[3'2'],[2'4'],[4'1'].$$
Then, the two cycles are glued as Figure~\ref{fig:eg_1dim}.
\begin{figure}[h]
    \centering
    \includegraphics[width=0.3\linewidth]{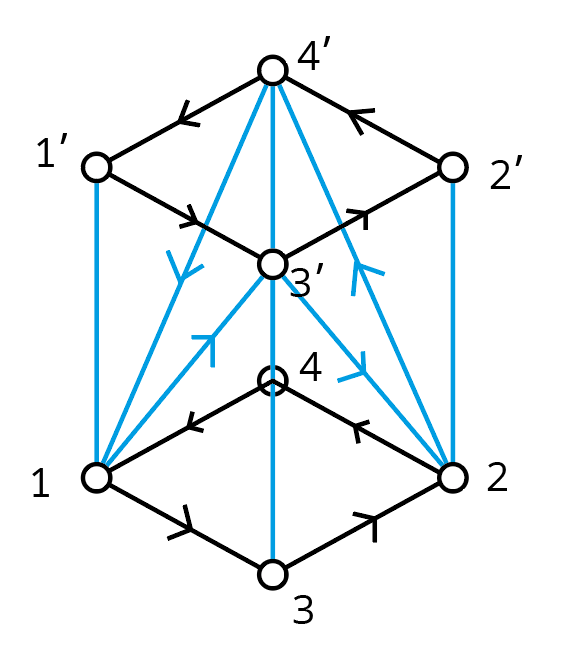}
    \caption{A graph for gluing two 1-dimensional holes. Arrows indicate the orientations in which a harmonic appears as a non-negative state.}
    \label{fig:eg_1dim}
\end{figure}

By the procedure of gluing the two cycles 
$$
\ket{c_1}=[13]+[32]+[24]+[41]
$$
and 
$$
\ket{c_1'}=
[1'3']+[3'2']+[2'4']+[4'1'].
$$
Let
$$\hat{L}_1^0=\{[13],[32],[24],[41]\}
$$
$$
\hat{L}_1^1= \{ [13'],[3'2],[2,4'],[4'1]
\}
$$
$$
\hat{L}_1^2=\{[1'2'],[2'3'],[3'4'],[4'1']\}.$$
Also let
$$
L_2^1=\{[133'],[323'],[244'],[414']\}
$$
and
$$
L_2^1=\{[13'1'],[3'22'],[24'2'],[4'11']\}.
$$
Then, the claims in Section~\ref{sec:General_argument} can be verified.

\subsection{Gluing 2-dimensional cycles}

We illustrate how to glue two octahedra. 
This example does not appear in the actual construction because we only treat $2m-1$-dimensional generalized octahedra. However, we treat the 2-dimensional case because it is higher-dimensional than the 1-dimensional case, but relatively easy to illustrate and understand intuitively. 

The target two octahedra are illustrated in Figure~\ref{fig:eg_2dim} (a). Note that we are working with the indexing of vertices like  Figure~\ref{fig:eg_2dim} (a) and {\it not} with the indexing like Figure~\ref{fig:eg_2dim} (b). In gluing two octahedra, there are 8 triangular prisms. One such triangular prism is illustrated in Figure~\ref{fig:eg_2dim} (c), and it will be divided into three tetrahedra. 
In the support of the harmonics that survive the gluing procedure, the lateral faces such as $155'$ and $11'5$ do not appear because they cancel out. 

\begin{figure}[h]
    \centering
    \includegraphics[width=0.65\linewidth]{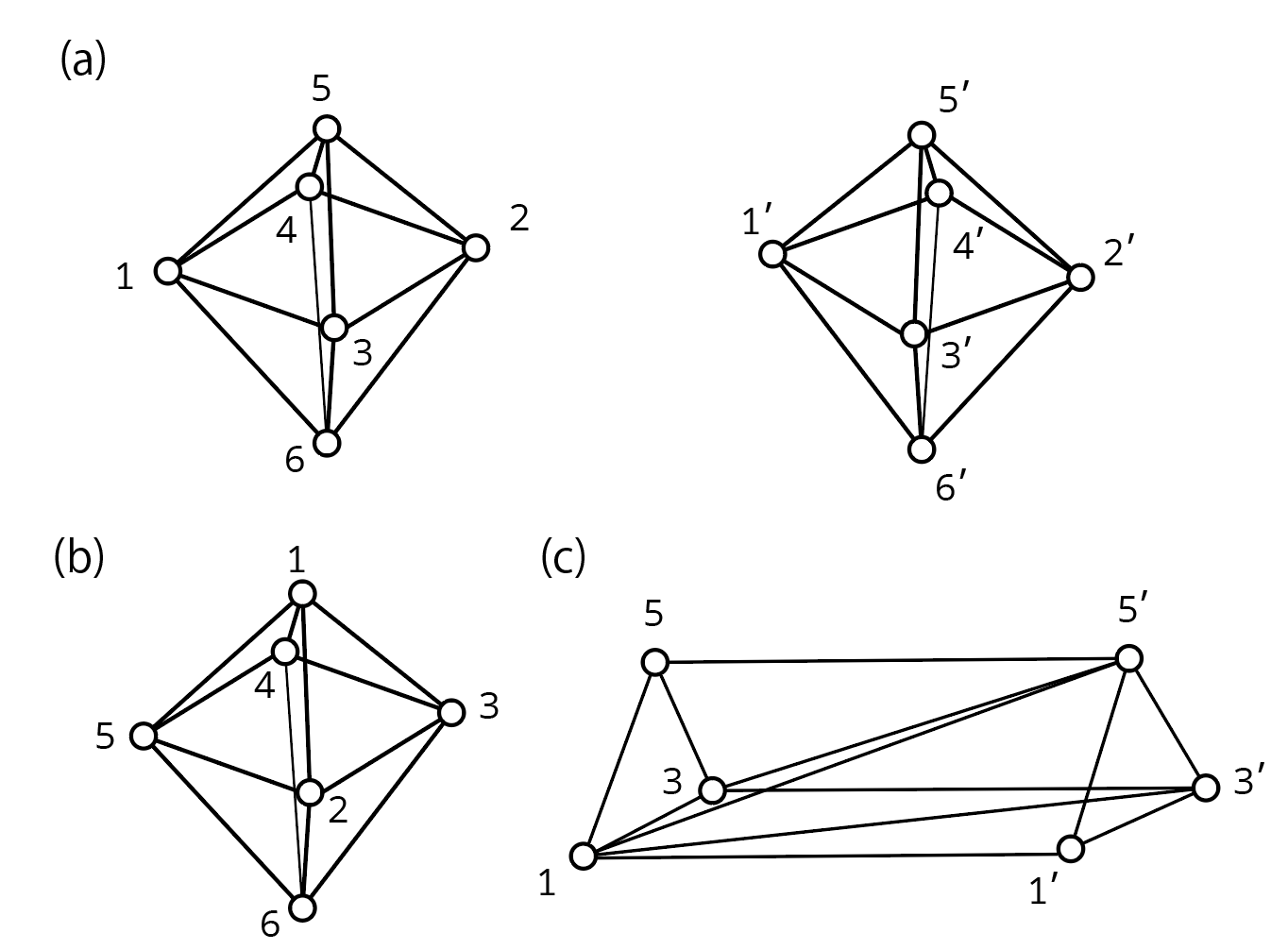}
    \caption{Gluing two octahedra.}
    \label{fig:eg_2dim}
\end{figure}

\end{document}